\documentclass[superscriptaddress, prx, longbibliography,twocolumn, nofootinbib, aps]{revtex4-2}

\usepackage{graphicx}
\usepackage{dcolumn}
\usepackage{bm}

\makeatletter
\newsavebox{\@brx}
\newcommand{\llangle}[1][]{\savebox{\@brx}{\(\m@th{#1\langle}\)}%
  \mathopen{\copy\@brx\kern-0.5\wd\@brx\usebox{\@brx}}}
\newcommand{\rrangle}[1][]{\savebox{\@brx}{\(\m@th{#1\rangle}\)}%
  \mathclose{\copy\@brx\kern-0.5\wd\@brx\usebox{\@brx}}}
\makeatother

\makeatletter
\def\l@subsubsection#1#2{}
\makeatother
\newcommand{\sket}[1]{|#1\rangle}
\newcommand{\sbra}[1]{\langle#1|}
\newcommand{\sbraket}[2]{\langle#1|#2\rangle}

\newcommand{\opket}[1]{|#1\rrangle}
\newcommand{\opbra}[1]{\llangle #1|}
\newcommand{\opbraket}[2]{\llangle #1 | #2 \rrangle}
\newcommand{\stkout}[1]{\ifmmode\text{\sout{\ensuremath{#1}}}\else\sout{#1}\fi}

\usepackage{dsfont}
\usepackage{wrapfig}
\usepackage{url}
\usepackage{physics}
\usepackage{amsmath,amssymb,amstext}
\usepackage{amsthm}
\usepackage{xcolor}
\usepackage[flushleft]{threeparttable}
\usepackage[bb=boondox]{mathalfa}
\usepackage{enumitem} 
\usepackage{comment}
\usepackage{float}

\usepackage[makeroom]{cancel}
\usepackage{hyperref}
\hypersetup{colorlinks,breaklinks,
            urlcolor=[rgb]{0.25,0.41,0.88},
            linkcolor=[rgb]{0.25,0.41,0.88},
            citecolor=[rgb]{0.25,0.41,0.88}}
\newtheorem{definition}{Definition}
\newtheorem{theorem}{Theorem}

\newtheorem{lemma}{Lemma}

\newtheorem{corollary}{Corollary}

\begin{document}

\title{Strong-to-Weak Symmetry Breaking Phases in Steady States of Quantum Operations}

\newcommand{\TUM}{\affiliation{Technical University of Munich, TUM School of Natural Sciences, Physics Department, 85748 Garching, Germany}}
\newcommand{\MCQST}{\affiliation{Munich Center for Quantum Science and Technology (MCQST), Schellingstr. 4, 80799 M{\"u}nchen, Germany}}
\newcommand{\FU}{\affiliation{Dahlem Center for Complex Quantum Systems, Freie Universität Berlin, 14195 Berlin, Germany}}

\author{Niklas Ziereis} \TUM \MCQST \FU
\author{Sanjay Moudgalya} \TUM \MCQST 
\author{Michael Knap} \TUM \MCQST 

\date{\today}
             
\begin{abstract}
Mixed states can exhibit two distinct kinds of symmetries, either on the level of the individual states (strong symmetry), or only on the level of the ensemble (weak symmetry).
Strong symmetries can be spontaneously broken down to weak ones, a mechanism referred to as Strong-to-Weak Spontaneous Symmetry Breaking (SW-SSB).
In this work, we first show that maximally mixed symmetric density matrices, which appear, for example, as steady states of symmetric random quantum circuits have SW-SSB when the symmetry is an on-site representation of a compact Lie or finite group. 
We then show that this can be regarded as an isolated point within an entire SW-SSB phase that is stable to more general quantum operations such as measurements followed by weak postselection.
With sufficiently strong postselection, a second-order transition can be driven to a phase where the steady state is strongly symmetric. 
We provide analytical and numerical results for such SW-SSB phases and their transitions for both abelian $\mathbb{Z}_2$ and non-abelian $S_3$ symmetries in the steady state of Brownian random quantum circuits with measurements.
We also show that such continuous SW-SSB transitions are absent in the steady-state of general strongly symmetric, trace-preserving quantum channels (including unital, Brownian, or Lindbladian dynamics) by analyzing the degeneracies of the steady states in the presence of symmetries.
Our results demonstrate robust SW-SSB phases and their transitions in the steady states of noisy quantum operations, and provide a framework for realizing various kinds of mixed-state quantum phases based on their symmetries.
\end{abstract}

\maketitle

\tableofcontents

\section{Introduction}
Gapped quantum phases in thermal equilibrium are defined by their stability under finite-depth local unitary circuits. 
Broadly speaking, such phases are characterized as either trivial or topologically ordered.
Including symmetries can further give rise to symmetry-broken or symmetry-protected phases of matter.
However, quantum systems are generally never truly isolated from their environment.
Instead, they can be thought of as being continuously measured by it.
The information obtained from such measurements is typically inaccessible to the observer, and the system is described by a mixed quantum state.
Yet, the way measurements affect complex quantum states challenges our understanding of quantum phases of matter out of equilibrium and their universal aspects.
It is therefore pertinent to develop a framework for mixed-state quantum phases of matter, to understand their symmetries, and to identify simple examples that illustrate their essential features.
A key aspect in determining the potential of a quantum state for information processing is therefore understanding its symmetries.
Recent research has focused on sharpening and further contrasting the differences between pure- and mixed-state symmetries~\cite{buca_note_2012, albert2014symmetries,  bao_symmetry_2021, ogunnaike_unifying_2023, moudgalya_symmetries_2023, chen_strong--weak_2024, huang_hydrodynamics_2024, kuno_strong--weak_2024, lee_quantum_2023, lessa_strong--weak_2024, ma_average_2023, sala_spontaneous_2024, moharramipour_symmetry_2024, li_highly-entangled_2024, lee_symmetry_2025}.
The latter come in two different forms~\cite{buca_note_2012, albert2014symmetries}.
An ensemble is said to be \textit{strongly symmetric}, when each state in the ensemble is symmetric on its own and carries the same symmetry charge.
By contrast, the ensemble is said to be \textit{weakly symmetric}, when the whole ensemble is symmetric on average, while the individual states are allowed to break the symmetry.
Recent works have illustrated that certain quantum systems can spontaneously break strong symmetries to weak ones, yielding a phase transition between a strongly symmetric and weakly symmetric phase. 
This phenomenon, dubbed as Strong-to-Weak Spontaneous Symmetry Breaking (SW-SSB) is measured by correlation functions $C$ that are non-linear in the density matrix. 
SW-SSB transitions have been explored in the following two different settings.
First, quantum systems with charge conservation in which the charge is continuously measured can undergo a transitions from a fuzzy to a sharp phase~\cite{agrawal_entanglement_2022, barratt_field_2022} in their \textit{steady state}.
Such information-theoretic transitions can be understood as SW-SSB transitions~\cite{Singh2025, Zerba2025} in the following way.
Denoting by $\mathbb{E}$ the ensemble average over the measurement outcomes and circuit realizations, this transition is measured by $\mathbb{E}[C(\ket{\Psi}\bra{\Psi}_{t \to \infty})]$, i.e., it requires one to resolve single trajectories~\cite{potter2022entanglement, fisher_random_2023}.
Second, density matrices that are strongly symmetric can exhibit SW-SSB under the application of certain \textit{finite-depth} quantum channels~\cite{lessa_strong--weak_2024,sala_spontaneous_2024,gu_spontaneous_2024,chen_strong--weak_2024, huang_hydrodynamics_2024,kuno_strong--weak_2024, lee_quantum_2023, ogunnaike_unifying_2023}.
In two spatial dimensions or higher, tuning parameters of these channels can also lead to a transition to a different phase where the density matrix exhibits strong symmetry \cite{lessa_strong--weak_2024, sala_spontaneous_2024}.
This is a transition in the averaged density matrix 
$\mathbb{E}[\rho]$ and is measured by $C(\mathbb{E}[\rho])$.
These two types of transition therefore differ in the objects showing the transition (single trajectory vs. density matrices) and in the way the final state is reached (steady state vs. single application of a channel).

\begin{figure}
    \centering
    \includegraphics[width=1\linewidth]{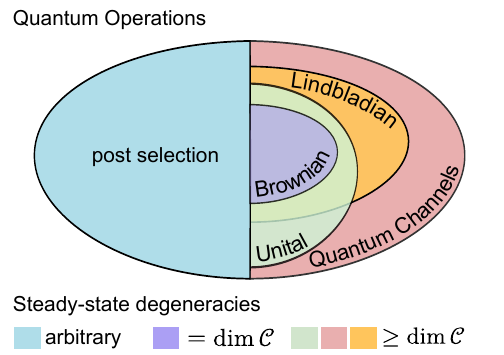}
    \caption{\textbf{Landscape of SW-SSB in the steady state of quantum operations:} Quantum operations obtained from physical, non-trace increasing maps  either consist of quantum channels (i.e., completely positive trace preserving maps), which include Lindbladian, unital, and Brownian circuits, or of non-trace preserving quantum operations obtained by postselection (see Sec.~\ref{sec:DynamicsOfMixedStates} for a review of these operations).
    We show that for strongly symmetric quantum channels the steady-state degeneracy is lower bounded by the dimension of the commutant $\dim \mathcal{C}$ and that their steady states are weakly symmetric. To obtain a strongly symmetric steady-state, the degeneracy has to be lowered, which can only be achieved by postselection.}
    \label{fig: where transition}
\end{figure}

A natural question then arises: When starting from a strongly symmetric initial state, can there be a \textit{transition} between a strongly symmetric phase and one exhibiting SW-SSB in the \textit{steady state} of \textit{quantum operations}?
In this work, we address this challenge and obtain the following key findings (see Fig.~\ref{fig: where transition} for an illustration):
\begin{enumerate}[label=(\roman*)]
    \item When one allows for postselection (i.e., breaking of the trace preservation property of a quantum channel), a continuous phase transition can occur at a finite postselection rate.
    For low postselection rates, the steady state is in the SW-SSB phase, while it is strongly symmetric for large rates.
    \item Postselection is necessary to drive a transition.
    This comes from the fact that SW-SSB of a group $G$ reduces to ordinary symmetry breaking of the form $G\times G\to G_{\mathrm{diag}}$ in the doubled Hilbert space.
    For a Landau-type transition this requires a change in the steady-state degeneracy, where the degeneracy in the symmetric phase is lower than in the symmetry broken phase.
    However, the steady-state degeneracy of a symmetric quantum channel is always lower bounded by $\dim \mathcal{C}$, where $\mathcal{C}$ is the associative algebra of the symmetry operators, also referred to as the \textit{commutant} of the symmetric operators that realize the channel.
    Therefore, it can only be reduced by non-trace preserving operations.
    Since any non-trace preserving quantum operation can be interpreted as postselection on a generalized measurement \cite{nielsen_quantum_2010}, we conclude that such a transition can only be achieved by postselection, see Fig.~\ref{fig: where transition}.
\end{enumerate}
Our work is structured as follows:
Sec.~\ref{sec: 2 SW-SSB} reviews the notions of strong and weak symmetries of mixed states. 
In Sec.~\ref{sec: 3 infinite temperature} we confirm previous expectations that a class of density matrices, which have been referred to as \textit{Maximally Mixed Invariant States (MMIS)}, exhibit SW-SSB \cite{sala_spontaneous_2024, lessa_strong--weak_2024, moharramipour_symmetry_2024} by providing a ready-to-use formula for the Rényi-2 correlator (which is a measure of SW-SSB) of these states for compact Lie group symmetries with exponentially growing symmetry sector, which includes finite groups.
By relating the Rényi-2 correlator to other measures of SW-SSB we conclude that this result is independent of the correlator used to define SW-SSB, thereby establishing MMIS as ideal examples of SW-SSB.
Our results confirm that symmetric random quantum circuits are a natural class of systems that exhibit strong-to-weak symmetry breaking in their steady states, as they thermalize to MMIS on average when initiated with a strongly symmetric state \cite{ogunnaike_unifying_2023, moudgalya_symmetries_2023, Mitsuhashi_2025}. 
In Sec.~\ref{sec:DynamicsOfMixedStates} we recall certain notions of maps between mixed quantum states and  review that the MMIS are the steady states of a class of Lindbladians known as Brownian circuits~\cite{ogunnaike_unifying_2023, moudgalya_symmetries_2023}, which serve as continuous-time versions of random quantum circuits.
In Sec.~\ref{sec: phase transition} we introduce a simple model that combines the dynamics of Brownian circuits with a measurement and postselection protocol.
We find that SW-SSB of the MMIS is stable to on-site measurements and a moderate amount of postselection towards a strongly symmetric target state. 
Further, for strong postselection, we can drive a phase transition in the averaged steady state of random quantum circuits between a weakly symmetric phase (which exhibits SW-SSB) and a strongly symmetric phase.
We demonstrate this mechanism analytically and numerically for both abelian $\mathbb{Z}_2$ and non-abelian $S_3$ symmetries in Brownian circuits. Furthermore, we show that the location of the SW-SSB transition in the steady-state density matrices does not depend on the correlator used.
In Sec.~\ref{sec: Lindbladian}, we determine that SW-SSB phase transitions cannot occur in the steady states of quantum channels by showing that the steady-state degeneracy of quantum channels can never be lower than that of a Brownian circuit with the same symmetry. We also construct an explicit example of Lindbladian dynamics where steering towards the same target state is performed by feedback to illustrate this general argument.
Thereby, the only way of driving a SW-SSB phase transition in steady states of quantum operations is via quantum operations that fail to be trace preserving. 
We provide and outlook and discussion in Sec.~\ref{sec: discussion} and technical details are relegated to the appendices.
\section{Strong-to-Weak Spontaneous Symmetry Breaking (SW-SSB)}
\label{sec: 2 SW-SSB}
In this section, we will review the notions of strong and weak symmetries for mixed quantum states~\cite{buca_note_2012, albert2014symmetries} and motivate the definitions of the Rényi-2 correlation function and the fidelity, which characterize SW-SSB.
\subsection{Pure state symmetries and order parameters}\label{subsec:pureSSB}
To set the stage, we first discuss Spontaneous Symmetry Breaking (SSB) in pure states.
A pure state $\ket{\Psi}$ is said to be \textit{symmetric} under the action of a group $\{U_g\}_{g\in G}$ if 
\begin{equation}    U_g\ket{\Psi}=e^{i\theta_g}\ket{\Psi}\label{eq: pure statate symmetry}
\end{equation}
for all $g\in G$. 
To give a definition of spontaneous symmetry breaking, let us first define an \textit{order parameter}, which is a strictly local operator $O_i$ (which could also have support over multiple sites in the vicinity of $i$) 
that transforms nontrivially under the symmetry group, i.e., $U_g O_i U_g^\dagger \neq O_i$. 
For example, a $\mathbb{Z}_2$ symmetry generated by $U_Z = \prod_i Z_i$, transforms the order parameter $X_i$ as $U_Z X_i U_Z^\dagger = -X_i$.
We then say that a state $\ket{\Psi}$ spontaneously breaks the symmetry (i.e., exhibits  SSB) if there exists an order parameter $O_i$ such that $\ket{\Psi}$ satisfies
\begin{equation}
     \lim_{|i-j| \to\infty}\frac{\bra{\Psi}O_i^\dagger O_j\ket{\Psi}}{\braket{\Psi}{\Psi}}\neq 0. \label{eq: pure state LRO}
\end{equation}
It should be emphasized that long-range correlations characterized by Eq.~\eqref{eq: pure state LRO} does not rule out the fact that $\ket{\Psi}$ obeys Eq.~\eqref{eq: pure statate symmetry}.
%
%\footnote{Strictly speaking, this is only true if we also allow for states that do not satisfy the cluster decomposition.}
%
For example, the cat state $\frac{1}{\sqrt{2}}\left(\ket{+ \ldots + }+\ket{-\ldots -}\right)$, where $X_i\ket{\pm}_i = \pm \ket{\pm}_i$, is symmetric under the action of the $\mathbb{Z}_2$ symmetry $U_Z$, but it has SSB according to Eq.~\eqref{eq: pure statate symmetry} when taking $X_i$ as the local order parameter.
More abstractly, we can interpret the correlation function $\bra{\Psi}O_i^\dagger O_j\ket{\Psi}$ as a way of ``comparing" the states $O_i\ket{\Psi}$ and $O_j\ket{\Psi}$ via their overlap defined as
\begin{equation}
    \frac{\bra{\Psi}O_i^\dagger O_j\ket{\Psi}}{\bra{\Psi}\ket{\Psi}} =\frac{\bra{O_i\Psi}\ket{O_j \Psi}}{\bra{\Psi}\ket{\Psi}}. \label{eq: correlation function inner prod}
\end{equation}
This viewpoint of ``comparing" states will prove useful when the notion of SSB is generalized to mixed states discussed now.
\subsection{Symmetries of mixed states}
Symmetries of mixed states come in two forms \cite{buca_note_2012, albert2014symmetries, gu_spontaneous_2024, sala_spontaneous_2024, lessa_strong--weak_2024}.
A density matrix $\rho$ is said to be \textit{strongly symmetric} under a group $G$ generated by $\{U_g\}$ if 
\begin{equation}
    U_g \rho =e^{i\theta_g}\rho \label{eq: strong symmetry} \qquad \textit{(strong symmetry)}
\end{equation}
for all $g\in G$.
In this case $\rho$ can be diagonalized in terms of symmetric pure states that are eigenstates which share the same eigenvalues under the symmetry, i.e., obey Eq.~\eqref{eq: pure statate symmetry} with the same function $\theta_g$.
On the other hand, a density matrix $\rho$ is called \textit{weakly symmetric} if 
\begin{equation}
    U_g\rho U_g^\dagger =\rho \label{eq: weak symmetry} \qquad \textit{(weak symmetry)}
\end{equation}
for all $g\in G$.
Weakly symmetric states are density matrices of ensembles that are symmetric only on average, and the individual states within the ensemble do not need to be symmetric.
Note that the distinction between strong and weak symmetries is inherent to mixed states, as the density matrix of any symmetric pure state is necessarily strongly symmetric.
Mixed state symmetries are better understood in terms of doubled Hilbert spaces.
\begin{figure}
    \centering
    \includegraphics[width=1.0\linewidth]{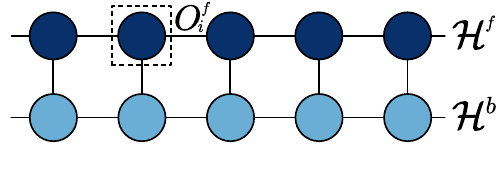}
    \caption{\textbf{The doubled Hilbert space:} The density matrix $\rho$ is described in terms of a pure state $\opket{\rho}$ living on the doubled Hilbert space $\mathcal{H}^f\otimes\mathcal{H}^b$. The new Hilbert space can be imagined as two copies of the original Hilbert space stacked stacked on top of each other. Operators of the form $O_i^{f/b}$ act nontrivially only on site $i$ of the forward/backward copy.}
    \label{fig: doubled Hspace}
\end{figure}
Given an orthonormal basis $\{\ket{b_n}\}$ of $\mathcal{H}$ and an operator $A = \sum_{nm}A_{nm}\ket{b_n}\bra{b_m}$ on $\mathcal{H}$ we write $\opket{A} = \sum_{nm}A_{nm}\ket{b_n}\otimes\ket{b_m} \in \mathcal{H}_f\otimes \mathcal{H}_b$.
The two copies of $\mathcal{H}$ are sometimes referred to as forward and backward copies.
In this way, we can understand the density matrix $\rho$ as an unnormalized pure state on double degrees of freedom $\mathcal{H}_f\otimes \mathcal{H}_b$ (see Fig.~\ref{fig: doubled Hspace}).
On the doubled Hilbert space, left/right multiplication of the density matrix $\rho$ by an operator $O_i$ corresponds to $O_i$ only acting on the forward/backward copy of the vectorized density matrix: 
\begin{align}
    O_i\rho \iff O_i\otimes \mathds{1}\opket{\rho},\quad \rho O_i^\dagger \iff \mathds{1}\otimes O_i^*\opket{\rho}.
\end{align}
We will often abbreviate $O_i^f:=O_i\otimes \mathds{1}$ and $O_i^{*b}:=\mathds{1}\otimes O_i^*$.
As an inner product we take the Hilbert-Schmidt inner product
\begin{equation} \opbraket{A}{B}= \Tr[A^\dagger B].
\end{equation}
By acting with $\dagger$ on both sides of Eq.~\eqref{eq: strong symmetry} we see that the vectorized density matrix $\opket{\rho}$ is separately symmetric under the action of $U_g\otimes \mathbb{1} =U_g^f$ and $\mathbb{1}\otimes U_g^{*} =U_g^{*b}$; and Eq.~\eqref{eq: strong symmetry} translates to $U_g^f\opket{\rho}=U_g^b\opket{\rho}=e^{i\theta_g}\opket{\rho}$ in analogy with Eq.~\eqref{eq: pure statate symmetry}.
Hence, the strong symmetry corresponds to a $G\times G$ symmetry in the doubled Hilbert space. 
By contrast, if $\rho$ is weakly symmetric, $\opket{\rho}$ is symmetric under the action of $U_g\otimes U_g^{*}$, which corresponds to being symmetric with respect to the diagonal subgroup $G_{\mathrm{diag}}$ of the strong symmetry.
We can then say that any density matrix, which in the doubled Hilbert space exhibits spontaneous symmetry breaking of the form $G \times G \to G_{\mathrm{diag}}$ (in a manner analogous to pure states, as we discuss in the next subsection), exhibits Strong-to-Weak SSB, or SW-SSB in short.
This should be contrasted with conventional SSB, where for a thermal transition in the Gibbs ensemble $G_{diag} \to 1$, and for quantum phase transitions in pure states $G\times G \to 1$.

\subsection{Correlation functions for SW-SSB}
Given the symmetry breaking pattern of SW-SSB in the doubled Hilbert space, we can define a correlation function for SW-SSB, generalizing Eq.~\eqref{eq: correlation function inner prod} to mixed states.
For a local order parameter $O_i$, instead of comparing $\ket{O_i\Psi}$ and $ \ket{O_j\Psi}$, we now compare the states $\opket{O_i\rho O_i^\dagger}$ and $\opket{O_j\rho O_j^\dagger}$ on the doubled Hilbert space.
In complete analogy to Eq.~\eqref{eq: correlation function inner prod}, we can define the correlator that measures SW-SSB as
\begin{equation}
    C^{sw}(i,j)[\rho] := \frac{\opbraket{O_i\rho O_i^\dagger}{O_j\rho O_j^\dagger}}{\opbraket{\rho}{\rho}}=\frac{\Tr[O_i\rho O_i^\dagger O_j\rho O_j^\dagger]}{\Tr[\rho^2]}. \nonumber
\end{equation}
This is referred to as the Rényi-2 correlator $C^{sw}$ in the literature \cite{lessa_strong--weak_2024, sala_spontaneous_2024, kuno_strong--weak_2024,sun_scheme_2025}.
We can use this correlator to arrive at a definition of SW-SSB:
\begin{subequations}
\begin{align}
&\lim_{|i-j|\to \infty}C^{sw}(i,j)[\rho]=\lim_{|i-j|\to \infty}\frac{\Tr[O_i\rho O_i^\dagger O_j\rho O_j^\dagger]}{\Tr[\rho^2]} \neq 0 \label{eq: R2 LRO}\\
&\lim_{|i-j|\to \infty}C^{w\emptyset}(i,j)[\rho]=\lim_{|i-j|\to \infty}\Tr[\rho O_i^\dagger O_j]=0. \label{eq: no R1 LRO}
\end{align}
\end{subequations}
Here, Eq.~\eqref{eq: no R1 LRO} ensures the absence of SSB of the weak symmetry, i.e., ensures the absence of conventional SSB. 
Similar to the case of pure states, Eq.~\eqref{eq: R2 LRO} does not imply that $\rho$ no longer obeys Eq.~\eqref{eq: strong symmetry}.
As an example, consider again the $\mathbb{Z}_2$ symmetry generated by $\prod_i Z_i$, the mixed state $\rho \propto \mathbb{1}+\prod_iZ_i$, and take $X_i$ as a local order parameter.
Then $\rho$ has SW-SSB according to Eq.~\eqref{eq: R2 LRO}, but still obeys Eq.~\eqref{eq: strong symmetry}.
The notion of SW-SSB is not uniquely defined. Another commonly used measure for comparing the ``states" $\opket{\rho}$ and $\opket{\sigma}$ is the fidelity $F(\rho,\sigma)=\Tr[\sqrt{\sqrt{\rho}\sigma\sqrt{\rho}}.]$ \footnote{Note that some sources define the fidelity as $F(\rho,\sigma)^2$.}.
Using this measure for comparison, SW-SSB can be defined as \cite{lessa_strong--weak_2024, sala_spontaneous_2024, moharramipour_symmetry_2024}
\begin{align}
    &\lim_{|i-j|\to \infty}F_O(i,j)[\rho] \neq 0,\nonumber\\
    &F_O(i,j)[\rho]  := \Tr[\sqrt{\sqrt{\rho}O_i^\dagger O_j\rho O_j^\dagger O_i\sqrt{\rho}}] \label{eq: fidelity correlator}
\end{align}
in addition to Eq.~\eqref{eq: no R1 LRO}.
It is known that for some density matrices this definition is not equivalent to the one in terms of the Rényi-2 correlator~\cite{lessa_strong--weak_2024}, introduced in Eq.~(\ref{eq: R2 LRO}).
In practice, the fidelity is significantly harder to evaluate than the Rényi-2 correlator as it does not have a simple expression in the doubled Hilbert space.
However it has certain mathematical properties that the Rényi-2 correlator lacks, e.g., long-range order defined in terms of $F_O$ has the following stability property:
If $\rho$ obeys Eq.~\eqref{eq: fidelity correlator}, then so does any $\mathcal{E}[\rho]$, where $\mathcal{E}$ is a symmetric low-depth quantum channel~\cite{lessa_strong--weak_2024}.

In this work, we will consider both the Rényi-2 and the fidelity correlator.
We derive a rigorous relation between the two for the case of thermal states in Sec. \ref{sec: 3 infinite temperature}.
Further, we show in Appendix \ref{sec: Equivalent measures of SW-SSB} Thm.~\ref{theorem: QuantumInfo correlators} that the definition of fidelity SW-SSB can be equivalently reformulated in terms of other quantum information theoretic quantities, thereby answering a conjecture from \cite{lessa_strong--weak_2024}.
In Sec.~\ref{sec: phase transition} we will present a model exhibiting a SW-SSB phase transition in its steady state by analytically analyzing the Rényi-2 correlator and numerically investigating the fidelity.
Strikingly, for this setup our numerical results indicate that the Rényi-2 correlator and the fidelity give rise to the same onset of SW-SSB for the steady-state density matrices.
As we will discuss later, this can be understood from the fact that the different correlators in $D$ dimensions act as different boundary correlators in the \textit{same} effective $D+1$ dimensional classical statistical mechanics model whose bulk is determined by the evolution to the steady state.
Since the bulk determines the properties of the phases, all boundary correlators will give rise to the same transition.
This should be contrasted with the finite-depth circuits studied in Ref.~\cite{lessa_strong--weak_2024}, for which examples can be constructed where the fidelity and the Rényi-2 correlator map onto \textit{different} $D$-dimensional classical statistical mechanics models, and hence yield different results for SW-SSB.
\section{Maximally Mixed Invariant States (MMIS)}
\label{sec: 3 infinite temperature}
In this section, we will  demonstrate that Maximally Mixed Invariant States (MMIS) exhibit SW-SSB for arbitrary symmetry groups in terms of the Rényi-2 correlator.
This idea is not new. For example, it was conjectured that all symmetric thermal states should have SW-SSB~\cite{lessa_strong--weak_2024}. 
Even before that, steady states of Brownian random circuits that explicitly break strong symmetries such as $\mathbb{Z}_2$ or $U(1)$ to their weak counterparts have been constructed~\cite{ogunnaike_unifying_2023, moudgalya_symmetries_2023}.
In addition, the Rényi-2 correlators has been investigated in sectors where all conserved charges are equal to zero for finite groups (i.e., $\theta_g=0$ in Eq.~\eqref{eq: pure statate symmetry})~\cite{sala_spontaneous_2024}.
It has also been argued that the fidelity correlator in such sectors should be non-zero at finite system sizes~\cite{moharramipour_symmetry_2024}, although the possibility that it decays to zero as $N \to \infty$ has not been fully ruled out.\footnote{Concretely, Ref.~\cite{moharramipour_symmetry_2024} argues that for any MMIS, $F_O(i,j)$ is strictly greater than $0$ for every finite system size and is independent of $i$ and $j$. However, these two conditions are insufficient to conclude that $F_O(i,j)[\rho^\infty_N] \xrightarrow[|i-j| \to \infty]{}>0$. For example, for the MMIS of a $U(1)$ symmetry with $\sum_iZ_i=0$ we have that $\Tr[(-\mathbf{S}_i\cdot \mathbf{S}_j) \rho^\infty_N]=\frac{N}{4N(N-1)}>0$, which is independent of $i,j$ but still decays to $0$ as $N \to \infty$.} 
Here, we obtain two novel results on MMIS.
First, we provide a simple formula for the Rényi-2 correlators of the MMIS in the limit $N \to \infty$ for arbitrary compact Lie groups, including finite groups.  
Second, we bound the fidelity correlator of these with their Rényi-2 correlators, and rigorously establish SW-SSB of these states irrespective of the correlator used.
In the following, we will consider systems that possess an on-site symmetry group $G$.
That is, given a unitary representation $\{u_g\}_{g\in G}$ of $G$ on the local Hilbert space $\mathcal{H}_{\mathrm{loc}}$ with dimension $d$, the group element $g$ is represented as $U_g = \bigotimes _{i=1}^Nu_g$ on the Hilbert space $\mathcal{H}_N = \bigotimes_{i=1}^N \mathcal{H}_{\mathrm{loc}}$ of a lattice containing $N$ qudits.
This covers most cases relevant in quantum many-body physics, including for example the conventional representations of $\mathbb{Z}_2$, $ U(1)$ and $SU(2)$. 
We will begin by introducing necessary terminology  before stating the main result of this section: a formula for the Rényi-2 correlator of the MMIS given in Eq.~\eqref{eq: formula R2 maximally mixed}.
\subsection{Scalar symmetry sectors}
A pure state $\ket{\Psi}$ is said to be symmetric under the action of $G$ if $U_g\ket{\Psi} = e^{i\theta_g}\ket{\Psi}$ for all $g \in G$.
The subspace containing all symmetric states that transform the same way under the action of $G$, i.e., that share the same function $\theta_g$, is referred to as a \textit{scalar} symmetry sector\footnote{In the language of representation theory, $\ket{\Psi}$  transforms under a 1d representation of $G$, and $V_\theta(N)$ is the isotypical component of the 1d irreducible representation labled by $\theta_g$.} 
\begin{equation}
V_\theta(N)=\{{\ket{\Psi}\in \mathcal{H}}_N|\; U_g \ket{\Psi}=e^{i\theta_g}\ket{\Psi}\}.\label{eq:scalarsymmetry}
\end{equation}
Physically, different functions $\theta_g$ correspond to different values of a conserved quantity.
For a $U(1)$ symmetry generated by $\sum_iZ_i$, different scalar symmetry sectors correspond to different eigenvalues of the magnetization $\sum_iZ_i$.
For a $\mathbb{Z}_2$ symmetry generated by $\prod_i Z_i$ the two different scalar symmetry sectors correspond to the subspaces of positive or negative global parity and are labeled by the eigenvalues of $\prod_i Z_i$. 
Importantly, for non-abelian symmetries, there are states that have well-defined conserved quantities, but do not lie within a scalar symmetry sector.
For example, for the standard $SU(2)$ symmetry, there is only one scalar symmetry sector corresponding to $S^2=0$.
For the purpose of SW-SSB we restrict ourselves to scalar symmetry sectors since this guarantees that all states \textit{within a certain symmetry sector} are strongly symmetric according to Eq.~\eqref{eq: strong symmetry}.
The \textit{Maximally Mixed Invariant State (MMIS) of $V_\theta(N)$}  is a \textit{symmetric infinite temperature state of $V_\theta(N)$}.
It is given by the normalized projection onto $V_\theta(N)$, i.e., it is the mixed state density matrix composed of all states within $V_\theta(N)$ with equal probability.
We denote this state by $\rho_N^\infty$ and suppress the dependence on $\theta_g$ as we will always work in a fixed scalar symmetry sector.
It follows from representation theory, that $\rho_N^\infty$ can be written as:
\begin{equation}
    \rho^\infty_N=\frac{1}{\mathrm{dim}(V_\theta(N))}\int_G e^{-i\theta_g}U_g dg\;, \label{eq: definition of maximally mixed}
\end{equation}
where integration is with respect to the Haar measure of $G$.
For example, the maximally mixed invariant states of a $\mathbb{Z}_2$-symmetry generated by $U_Z=\prod_{i=1}^N Z_i$ can be written as $\frac{1}{2^N}\left(\mathbb{1}\pm U_Z\right)$, where $+/-$ refers to the sectors of positive/negative parity, respectively, see Fig.~\ref{fig: Overview} (a).
This extends previous definitions from \cite{moharramipour_symmetry_2024}, where the notion of \textit{maximally mixed states in the invariant sector} was used to describe maximal mixtures in scalar symmetry sectors with $\theta_g=0$, corresponding to the positive parity sector for the $\mathbb{Z}_2$-symmetry above or the $\sum_iZ_i=0$ symmetry sector for the $U(1)$ symmetry.

\subsection{Rényi-2 correlators}
We now discuss the form of the Rényi-2 correlators of Eq.~(\ref{eq: R2 LRO}) for the MMIS of scalar symmetry sectors.
When taking the limit $N \to \infty$ we keep the function $\theta_g$ unchanged, e.g., for a $U(1)$ symmetry, we keep the eigenvalue $Q=\sum_i Z_i$ constant across all system sizes, which corresponds to a finite charge density. 
We now state the result which is valid for all compact Lie groups with exponentially growing scalar symmetry sectors, i.e., $\dim V_\theta(N_L)\sim \frac{1}{\mathrm{poly}(N)}(\dim \mathcal{H}_{\mathrm{loc}})^N$.
This is true for all finite groups, all semisimple compact Lie groups \cite{Biane_1993}, as well as for certain representations of compact connected Lie groups~\cite{Tate_2004}.
The class of finite groups covers groups such as $\mathbb{Z}_n$, and $S_n$ and the semisimple compact Lie groups includes $SU(2)$.
The group $U(1)$ is not semisimple but it can be confirmed (see App. \ref{sec: U1 calculations}) that the scalar symmetry sectors are exponentially growing if the $U(1)$ symmetry is generated by $\sum_iZ_i$ and hence our result applies.
The final expression for the Rényi-2 correlator for such groups reads
\begin{align}
&\lim_{|i-j|\to\infty}C^{sw}(i,j)[\rho_N^\infty]=\frac{\norm{O}_2^4}{(\dim \mathcal{H}_{\mathrm{loc}})^2|\mathcal{I}_O|}. \label{eq: formula R2 maximally mixed}
\end{align}
Here, $O$ is the order parameter used in the expression of $C^{sw}$, $|\mathcal{I}_O|$ is the dimension of the irreducible representation under which it transforms, and $\norm{O}_2=\Tr[|O|^2]^{\frac{1}{2}}$ is its Frobenius norm. 
Moreover, in App.~\ref{appendix: SW-SSB maximally mixed} Thm. \ref{theorem: SWSSB of MMIS}, we also allow for multi-site order parameters (recall that the local operator $O_i$ in Sec.~\ref{sec: 2 SW-SSB} can also have support on multiple but finitely many sites) and consider more general correlation functions that go beyond $C^{sw}$.
We also confirm that the MMIS does not break down the weak symmetry further, i.e., that 
\begin{equation}
    \lim_{|i-j|\to\infty}C^{w\emptyset}(i,j)[\rho^\infty_N] =0 .
\end{equation}
This establishes SW-SSB for the maximally mixed invariant states for arbitrary finite or compact Lie groups $G$ with exponentially growing scalar symmetry sectors.
We will now demonstrate our result for two simple examples on a chain of qubits.
\begin{enumerate}
    \item Consider the $\mathbb{Z}_2$ symmetry generated by $U_Z=\prod_iZ_i$. As a local order parameter, we can pick the Pauli $X$ matrix. Since $U_ZX_iU_Z^\dagger =-X_i$, the order parameter transforms under a one-dimensional representation and we thus have $|\mathcal{I}_X|=1$.
    Further, since $\norm{X}_2^2=2$, we find using Eq.~(\ref{eq: formula R2 maximally mixed}) that $C^{sw}=\frac{2^2}{2^2\cdot 1}=1$, which agrees with a direct calculation.
    \item Next, we consider a $U(1)$ symmetry generated by $\sum_iZ_i$.
    Since $e^{i\phi \sum_jZ_j}X_i e^{-i\phi \sum_jZ_j}=\cos(\phi)X_i+\sin(\phi)Y_i$, we can again take $X_i$ as the order parameter.
    The representation under which the order parameter transforms is now two dimensional and we consequently have $|\mathcal{I}_X|=2$ and thus obtain $C^{sw}=\frac{2^2}{2^2\cdot 2}=\frac{1}{2}$.
    In Appendix \ref{sec: U1 calculations} we show how our formula can be extended to the Rényi-2 correlator of quadratic symmetry-invariant operators such as $\frac{\Tr[\rho \mathbf{S}_i\mathbf{S}_j\rho\mathbf{S}_j\mathbf{S}_i]}{\Tr[\rho^2]}$. This correlator has been studied in the context of ordinary symmetry breaking~\cite{hauser_continuous_2023} and also extended to SW-SSB~\cite{lessa_strong--weak_2024}.
\end{enumerate}
In Sec.~\ref{sec: Potts model} we will also apply the formula to a non-abelian symmetry, in particular, the $S_3$ symmetry of the Potts model.
\subsection{Fidelity correlators}
The fidelity correlator can be used as a measure of SW-SSB as well.
Since it is known to be inequivalent to the Rényi-2 correlator in general, it is natural to study its behavior in symmetric thermal states, i.e., $\rho_\beta = \frac{P_\theta e^{-\beta H}}{\Tr[P_\theta e^{-\beta H}]}$ where $P_\theta$ is the projection onto the scalar symmetry sector $V_\theta$ (i.e., the unnormalized maximally mixed invariant state) and $H$ is a $G$-symmetric Hamiltonian.
We show in Appendix \ref{sec: fidelity bounds} (see Lemma \ref{lemma: bounds for fidelity}) that:
\begin{align}
    C^{sw}(i,j)[\rho_\beta]\leq \norm{O}_\infty^2F_O(i,j)[\rho_{\frac{\beta}{2}}] \label{eq: fidelity Rényi},
\end{align}
where $\norm{O}_\infty$ is the operator norm of the order parameter $O$.
Here, we again assume $G$ to be a compact Lie group, and unlike in Eq.~(\ref{eq: formula R2 maximally mixed}), we have no assumptions on $\dim V_\theta$.
Eq.~\eqref{eq: fidelity Rényi} implies that for the MMIS at $\beta = 0$ long-range order of the Rényi-2 correlator implies long-range order of the fidelity correlator.
We also show that this equivalence extends to other measures of SW-SSB such as the trace distance and quantum divergences by showing that they are equivalent to the fidelity for all density matrices $\rho$ in Appendix \ref{sec: Equivalent measures of SW-SSB} (see Theorem \ref{theorem: QuantumInfo correlators} therein).  
Hence we can conclude that SW-SSB of the MMIS is independent of the correlation function used to define SW-SSB.
This further cements their role as ideal examples of SW-SSB.
\section{Dynamics of mixed states \label{sec:DynamicsOfMixedStates}}
In this section, we will review concepts and models for quantum operations and quantum channels, with a focus on  Lindbladian and Brownian dynamics.
\subsection{Quantum operations}
We refer to a \textit{quantum operation} as a completely positive map 
\begin{equation}
    \rho \mapsto \mathcal{E}[\rho]=\sum_iK_i\rho K_i^\dagger    \label{eq: quantum operation}
\end{equation}
with $\sum_i{K_i^\dagger K_i} \leq \mathbb{1}$.\footnote{For operators $A,B$ $A\leq B$ means that $B-A$ is positive semidefinite~\cite{nielsen_quantum_2010, breuer_theory_2007}}
The operators $K_i$ are referred to as Kraus operators.
A completely positive and trace preserving (CPTP) map is called a \textit{quantum channel}, which is equivalent to $\sum_iK_i^\dagger K_i = \mathbb{1}$ \cite{watrous_theory_2018, nielsen_quantum_2010}.
An alternative interpretation of a quantum channel is that the set $\{M_i\}_i$, $M_i=K_i^\dagger K_i$ forms a positive operator valued measure (POVM) and hence Eq.~\eqref{eq: quantum operation} represents the density matrix after a generalized measurement.
Here, we will be interested in general quantum operations which are not necessarily trace-preserving. 
A quantum operation with $\sum_i K_i^\dagger K_i <\mathbb{1}$ can be interpreted as postselection on a generalized measurement. Let us define $A$ by $A^\dagger A =\mathbb{1}-\sum_i K_i^\dagger K_i$.
The channel $\widetilde{\mathcal{E}}[\rho]=\sum_iK_i\rho K_i^\dagger+A\rho A^\dagger$ then represents a generalized measurement of the POVM $\mathcal{P}=\{K_i^\dagger K_i\}_i\cup \{A^\dagger A\}$.
From this vantage point, the original operation $\sum_iK_i \rho K_i^\dagger$ is obtained from postselection on the POVM $\mathcal{P}$, since we drop measurement outcomes corresponding to $A^\dagger A$ .
Another relevant concept in this context is \textit{unitality}, which corresponds to the condition that the Kraus operator satisfies $\sum_iK_iK_i^\dagger=\mathbb{1}$.
As we will discuss below, this condition directly imposes restrictions on structure of the steady states. We illustrate the classification of quantum operations in Fig.~\ref{fig: where transition}.

\subsection{Lindbladian dynamics}
Given a one-parameter family of quantum operations $\mathcal{E}_t$ that fulfills the semigroup properties $\lim_{t\to 0}\mathcal{E}_t[\rho]=\rho$ and $\mathcal{E}_{t_1+t_2}=\mathcal{E}_{t_1}\circ \mathcal{E}_{t_2}$ we can find a generator $\mathcal{L}$ such that $\rho (t)=\mathcal{E}_t[\rho(0)]=e^{-t\mathcal{L} }\rho(0)$.
If $\mathcal{E}_t$ is a quantum channel, i.e., if it is also trace preserving, then $\mathcal{L}$ always has the following form \cite{breuer_theory_2007, lindblad_generators_1976}:
\begin{align}
    \mathcal{L}\rho = i[H,\rho]-\sum_i \gamma_i\left(2L_i \rho L_i^\dagger- \{L_i^\dagger L_i, \rho\}\right), \label{eq: Lindbladian}
\end{align}
where $\gamma_i \geq 0$, $H$ is hermitian and $L_i$ are referred to as Lindblad operators.
For $\rho(t) = e^{-t\mathcal{L} }\rho(0)$, with $\mathcal{L}$ of Eq.~\eqref{eq: Lindbladian}, the dynamics is said to be \textit{Lindbladian} or \textit{Quantum Markovian}.
In the following, we will often specify a quantum operation $\mathcal{E}_{dt}$ with $\lim_{dt\to 0}\mathcal{E}_{dt}[\rho]=\rho$ and associate with it a generator $\mathcal{L}$ via
\begin{align}
    -\mathcal{L}=\lim_{dt \to 0}\frac{\mathbb{1}-\mathcal{E}_{dt}}{dt} \label{eq: channel to superop}.
\end{align}
If $\mathcal{E}_{dt}$ is not trace preserving, $\mathcal{L}$ will deviate from Eq.~\eqref{eq: Lindbladian}. 
In general, Lindbladian dynamics does not need to be unital.
However, when the jump operators $\{L_i\}$ are Hermitian unitality is directly implied from trace preservation and vice-versa.
We will analyze the dynamics of the density matrix $\rho$ in the doubled Hilbert space $\mathcal{H}_f\otimes \mathcal{H}_b$ and write:
\begin{align}
    &\rho(t)=e^{-t\mathcal{L}}\rho(0) \iff \opket{\rho(t)} =e^{-tP}\opket{\rho(0)} \label{eq: Dynamics doubled H space}\\
    &P=iH^f -iH^{*b}-\sum_i\gamma_i \left[2L_i^fL_i^{*b}-L_i^{f\dagger}L_i^f -\left(L_i^{\dagger b}L_i^b\right)^*\right], \nonumber  
\end{align}
where $P$ acts as an effective non-hermitian Hamiltonian in the doubled Hilbert space.
By writing the dynamics in the form of Eq.~\eqref{eq: Dynamics doubled H space} we can understand the steady states at $t \to \infty$ as right groundstates of the effective Hamiltonian $P$, i.e., states with $P\opket{\rho} =0$.
Using these concepts, we translate between dynamical generators $\mathcal{L}$, quantum operations $\mathcal{E}_{dt}$ and effective non-hermitian Hamiltonians $P$. 
\subsection{Symmetries of quantum operations}
\label{sec: symmetries of operations}
Similar to states, notions of strong and weak symmetries exist also for maps between quantum states, both for quantum operations $\mathcal{E}$ and for the infinitesimal generator $\mathcal{L}$ \cite{buca_note_2012}.
$\mathcal{F} \in \{\mathcal{L}, \mathcal{E}\}$ is called \textit{weakly symmetric} if $\mathcal{F}[U_g\rho U_g^\dagger]= U_g\mathcal{F}[\rho ]U_g^\dagger$.
This translates to the condition that the effective Hamiltonian $P$ of Eq.~(\ref{eq: Dynamics doubled H space}) commutes with two copies of the symmetry operator $U_g$, i.e., $[P, U_g \otimes U_g^\ast] = 0$.
By contrast, $\mathcal{F}$ is referred to as \textit{strongly symmetric} if $[L_i,U_g]=[H,U_g]=0$ for a Lindblad time evolution, or $[K_i,U_g]=0$ for a quantum operation in Kraus form.
This translates to the condition that the effective Hamiltonian $P$ commutes with single copies of the symmetry, i.e., $[P, U_g \otimes \mathds{1}] = [P, \mathds{1} \otimes U_g^\ast] = 0$. 
Hence, any strongly symmetric quantum channel is also weakly symmetric, but not necessarily vice versa.
The strong symmetries of a quantum operation can sometimes dictate the structure of the steady-state manifold directly. 
For example, when a quantum operation is unital, i.e., $\sum_i{K_i K^\dagger_i} = \mathds{1}$, any density matrix that commutes with the $\{K_i\}$ is a steady state. 
This also means that the MMIS of Eq.~(\ref{eq: definition of maximally mixed}) are the steady states of strongly symmetric unital channels, which is evident using the fact that $[K_i,U_g]=0$.
The quantum operations considered in the rest of the paper will all be strongly symmetric.
We then say that a semigroup of quantum operations $\mathcal{E}_t$ has SW-SSB, if, starting from a strongly symmetric initial state $\rho_0$, the steady state $\lim_{t\to \infty}\frac{\mathcal{E}_t[\rho_0]}{\Tr[\mathcal{E}_t[\rho_0]]}$ has SW-SSB.
This extends the examples of SW-SSB of density matrices obtained by the application of finite-depth quantum channels $\mathcal{E}$ on simple density matrices, e.g., such as those in \cite{lessa_strong--weak_2024,sala_spontaneous_2024}, to steady states.

\subsection{Brownian circuits}\label{subsec:brownian}
We will now review a class of random quantum circuits called Brownian circuits, which turn out to yield nice examples of \textit{unital} Lindbladian dynamics.
They have been studied previously in multiple contexts, e.g., in the context of SYK models \cite{lashkari_towards_2013, bauer_stochastic_2017, jian_syk_2021, sunderhauf_quantum_2019, agrawal_entanglement_2022, agarwal_emergent_2022}, to study quantum chaos \cite{zhou_operator_2019, bauer_stochastic_2017, sunderhauf_quantum_2019, agarwal_charge_2023, ogunnaike_unifying_2023}, entanglement \cite{Knap_2018, bernard_entanglement_2021, swann_spacetime_2023,vardhan_entanglement_2024} or in the context of symmetries and hydrodynamics
\cite{moudgalya_symmetries_2023, vardhan_entanglement_2024, sahu_entanglement_2022, chen_strong--weak_2024}.
In this setting, we evolve our system over the time interval $dt \ll 1$ with a random Hamiltonian of the form
\begin{equation}
    H(t)=\sum_\alpha J_\alpha (t)B_\alpha,
\end{equation}
where $B_\alpha$ are local Hermitian operators and the coupling constants $J_\alpha$ are i.i.d. Gaussian random variables with 
\begin{equation}
    \mathbb{E}[J_\alpha(t)]=0, \quad \mathbb{E}[J_\alpha(t)J_\beta(t')]=\frac{2J_\alpha \delta_{tt'}\delta_{\alpha \beta}}{dt}. \label{eq: random couplings}
\end{equation}
Here, we will be interested in the properties of the density matrix obtained under such random time evolution, \textit{after} averaging over the random variables $J_\alpha(t)$, 
For this, it is convenient to vectorize the density matrix $\rho \to \opket{\rho}$.
Then the averaged time evolution takes the form \cite{ogunnaike_unifying_2023, moudgalya_symmetries_2023}:
\begin{align}
   &\mathbb{E}\left[U(dt)\otimes U^*(dt)\right]
    =1-dt\cdot P_{\mathrm{B}}+\mathcal{O}(dt^2) \nonumber\\ 
    &P_{\mathrm{B}}=\sum_\alpha J_\alpha \mathcal{}\left(B_{\alpha,f}-B^{T}_{\alpha,b} \right)^2.  \label{eq: Brownian formula}
\end{align}
Comparing with Eq.~\eqref{eq: Dynamics doubled H space}, we see that $P_{\mathrm{B}}$ realizes the dissipative part of Eq.~\eqref{eq: Lindbladian} with the $B_\alpha$ taking the role of the Lindblad operators $L_i$.
Since the $\{B_\alpha\}$ are Hermitian, the associated Lindbladian dynamics is also unital, which already determines the structure of the steady states, as we discuss below.
Hence, if we evolve the system by $t=N_t\cdot dt$ and take the limit $dt \rightarrow 0$, the dynamics in the doubled Hilbert space is given by an imaginary time evolution with the effective Hermitian Hamiltonian $P_{\mathrm{B}}$:
\begin{align}
&\mathbb{E}\left[U(t)\otimes U(t)^*
\right]\opket{\rho_0}=\prod_{i=1}^{N_t} \mathbb{E}\left[U(dt)\otimes U^*(dt)\right] \opket{\rho_0} \nonumber \\
&= e^{-tP_{\mathrm{B}}}\opket{\rho_0}.\label{eq: average dynamics}
\end{align}
The late-time density matrix is then simply a ground state of this effective Hamiltonian $P_B$, which is also Hermitian since all the $\{B_\alpha\}$ are.
The ground state manifold of $P_B$ is related to the symmetries of the evolution $\{B_\alpha\}$ ~\cite{moudgalya_hilbert_2022, moudgalya_symmetries_2023, moudgalya_symmetries_2023-1, ogunnaike_unifying_2023}.
Concretely, let $\mathcal{C}$ be the algebra generated by all operators commuting with all $B_\alpha$, i.e., $\mathcal{C}=\{A \in L(\mathcal{H})|AB_\alpha = B_\alpha A \; \forall \alpha\}$ where $L(\mathcal{H})$ is the set of linear operators on $\mathcal{H}$.
$\mathcal{C}$ is also referred to as the commutant of $\{B_\alpha\}$, and it spans the groundstate manifold of $P_\mathrm{B}$~\cite{moudgalya_symmetries_2023, vardhan_entanglement_2024}.
This means that if $\{C_m\}_{m \in \mathcal{M}}$ is an orthonormal basis for $\mathcal{C}$ (i.e., $\Tr[C_mC_n]=\delta_{nm}$), the steady-state density matrix is given by (up to normalization)
\begin{align}
    \rho^{\mathrm{eq}}= \sum_{m \in \mathcal{M}}C_m\Tr[C_m^\dagger\rho_0]. \label{eq: commutant steady state}
\end{align}
When the initial density matrix $\rho_0$ is strongly symmetric with $U_g\rho_0 = e^{i\theta_g}\rho_0$, then $\rho^{\mathrm{eq}}$ is the MMIS for $V_\theta$ from Eq.~\eqref{eq: definition of maximally mixed}. 
To see this, we need to show that $\rho^{\mathrm{eq}}$ is proportional to the projection onto $V_\theta$, the scalar symmetry sector defined in Eq.~(\ref{eq:scalarsymmetry}).
That is, we want (i) $\rho^{\mathrm{eq}}\ket{\Psi}\in V_\theta$ for all $\ket{\Psi}\in \mathcal{H}$ and (ii) $\rho^{\mathrm{eq}}\ket{\Psi_\theta}\propto \ket{\Psi_\theta}$ for all $\ket{\Psi_\theta}\in V_\theta$.
To show condition (i), note that if $\{C_m\}_m$ is an orthonormal basis for $\mathcal{C}$, then so is $\{U_gC_m\}_m$ for any choice of unitary symmetry operator $U_g$.
Thus for any $\ket{\Psi} \in \mathcal{H}$, we have
\begin{align}
    &U_g(\rho^{\mathrm{eq}}\ket{\Psi})\propto \sum_{m \in \mathcal{M}}\underbrace{U_gC_m}_{C'_m}\Tr[C_m^\dagger\rho_0] \nonumber \\
    &=\sum_{m \in \mathcal{M}}C_m'\Tr[(C_m')^\dagger U_g\rho_0]=e^{i\theta_g}(\rho^{\mathrm{eq}}\ket{\Psi}),
\end{align}
hence $\rho^{\mathrm{eq}}\ket{\Psi} \in V_\theta$.
To show condition (ii), we use that by definition of the commutant $\mathcal{C}$, we can express every $C_m$ as a linear combination of symmetry operators $C_m=\sum_{g \in G}\alpha^{(m)}_gU_g$. 
Hence $C_m \ket{\Psi_\theta}=\sum_{g \in G}\alpha^{(m)}_ge^{i\theta_g}\ket{\Psi_\theta}$ and the second condition follows directly using Eq.~(\ref{eq: commutant steady state}).

\begin{figure*}
    \centering
    \includegraphics[width=1\linewidth]{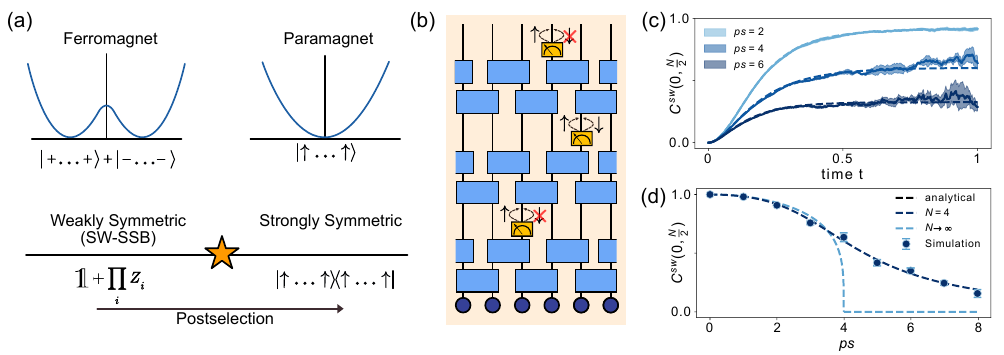}
    \caption{\textbf{Postselection induced SW-SSB transition in a $\mathbb{Z}_2$ symmetric system:} \textbf{(a):} Correspondence principle between the quantum phase transition in the transverse field Ising model (TFI) and the SW-SSB transition of a mixed quantum state. The ferromagnetic cat states correspond to the maximally mixed invariant states. \textbf{(b):} Circuit architecture of the model. Between each layer of random gates a random $Z$ measurement is performed. For a fraction of measurements, the $\downarrow$ result is excluded from the ensemble. \textbf{(c):} Dynamics of the Rényi-2 correlator $C^{sw}(0,\frac{N}{2})$, Eq.~\eqref{eq: Rényi-2 correlaotr Z2}, from simulations for a chain of $N=4$ qubits with periodic boundary conditions. Other parameters were $J=U=1$, $s=\frac{1}{2}$, $dt = 2.5\cdot 10^{-3}$. Solid line is an average over $4\cdot 10^{4}$ iterations, the shaded area represents a Jacknife estimate of the error and dashed line is the analytical result from the TFI. \textbf{(d):} Steady state value of the Rényi-2 correlator. Dashed lines are analytical results for the TFI with $N=4$ and $N\to \infty$, respectively. }
    \label{fig: Overview}
\end{figure*}

Thus, Brownian circuits exhibit SW-SSB in terms of their averaged density matrix.
For abelian symmetries, all irreps of the symmetry are one-dimensional, hence the commutant $\mathcal{C}$ has a basis in terms of projections onto scalar symmetry sectors, which are the MMIS.
Since convex combinations of SW-SSB states again exhibit SW-SSB (see App. \ref{sec: convex combinations}), the entire steady state manifold has SW-SSB.
For non-abelian symmetries, since some irreps are necessarily not one-dimensional, the MMIS states \textit{do not} span the entire ground state manifold of $P_\mathrm{B}$.
Nevertheless, the MMIS are the only ones reachable from strongly symmetric initial states, as shown above.
This leaves open the question on the full structure of the steady state manifold for non-abelian symmetries (which are not accessible from strongly symmetric initial states), and whether they exhibit SW-SSB, which we leave for future work. 
\section{SW-SSB transitions in random quantum circuits with postselection}

\label{sec: phase transition}
The observation that MMIS are the steady states of strongly symmetric unital quantum operations, and the fact that they have SW-SSB establishes such random circuits as a natural class of systems to study  SW-SSB.
We now explore whether there is (i) a \textit{phase} of steady state density matrices that exhibit SW-SSB, and whether (ii) one might drive the system to a steady state that \textit{does not} exhibit SW-SSB and is strongly symmetric instead.
To this end, we need to relax the condition of having a unital and trace-preserving quantum channel.
We will discuss in Sec.~\ref{sec: Lindbladian} that breaking unitality alone is not sufficient; rather we need to break the trace-preservation. 
Interspersing such non-trace-preserving operations in between the random unitary gates of the quantum circuit could then drive a phase transition towards a state that does not have SW-SSB, see Fig.~\ref{fig: where transition}. 
Due to the absence of trace preservation the resulting dynamics is therefore not a quantum channel and the continuous-time limit of such an operation is also not Lindbladian by construction.
Let us now describe a concrete protocol and operations that realize these ideas.

\subsection{Protocol} 
\label{sec: protocol}
Given an on-site symmetry group $G$, we pick a ``target state" $\ket{\Psi_\mathrm{t}}=\ket{\varphi}\otimes \ldots\otimes\ket{\varphi}$, which is a product state of local $G$-symmetric states $\ket{\varphi}$,  e.g.,  a paramagnetic state for $\mathbb{Z}_2$ symmetry.
We also pick a local $G$-symmetric observable $A$ such that $\ket{\varphi}$ corresponds to one of the measurement outcomes.
In the background, we implement $G$-symmetric Brownian dynamics.
After evolving the system by $dt$ with a random Hamiltonian, we apply the following measurement protocol on a lattice with $N$ sites for some $p\in [0,\frac{1}{dt\cdot N}]$ and $s\in[0,1]$ (see Fig.~\ref{fig: Overview}(b) for an illustration):
\begin{enumerate}
    \item  With a probability $p\cdot dt\cdot N$, measure $A_i$ at a uniformly random chosen site.
    \item  Given that a measurement occurred, we perform the following probabilistic postselection step: With a probability of $s$, discard the trajectory if the measurement outcome does not correspond to $\ket{\varphi}$.
\end{enumerate}
The normalization of the measurement probability $p\cdot dt\cdot N$ is chosen such that $p$ is interpreted to be the measurement rate per lattice site. 
As $s$ is the conditional probability of postselecting, provided that a measurement was performed, $p\cdot s$ is the postselection rate per site.
We expect the following behavior of the circuit. 
For $s=0$, the averaged steady state should be a MMIS, which has SW-SSB according to our previous discussion. 
On the other hand, when $p\cdot s$ is large, the steady state is expected to be close to target density matrix $\ket{\Psi_\mathrm{t}}\bra{\Psi_\mathrm{t}}$ which is a strongly symmetric pure state and thus does not have SW-SSB. 
In between these two limits a phase transition should occur. 
The above protocol alters the effective Hamiltonian $P_\mathrm{B}$ of Eq.~(\ref{eq: Brownian formula}).
We check SW-SSB of the steady state by finding the groundstates of the resulting Hamiltonian, and by computing the Rényi-2 and fidelity correlators for the resulting steady state.
For a $\mathbb{Z}_2$  symmetric system, we find an analytic solution in arbitrary dimensions by relating our system to the transverse field Ising model of effective degrees of freedom. 
We find that there is a ``phase" of steady states that exhibits SW-SSB, and a phase that does not. 
These phases are stable to local Brownian perturbations that preserve the symmetry by utilizing the mapping to the ground state of an effective quantum model.
We also numerically evaluate higher Rényi-$m$ correlators from which we extrapolate the fidelity.  
Strikingly,  our numerical results are consistent with a simultaneous onset of SW-SSB in all of these distinct correlators evaluated with the averaged density matrix and diagnose SW-SSB from $G\times G \to G_\text{diag}$.
We also study an $S_3$-symmetric Potts circuit as an example of a non-abelian symmetry, where we are able to demonstrate the SW-SSB phase and a phase transition to a strongly symmetric phase numerically on a one-dimensional chain.
%
%We find the ground state of the effective Hamiltonian using tensor network methods. 
%
%
%

\begin{figure*}
    \centering
    \includegraphics[width=1\linewidth]{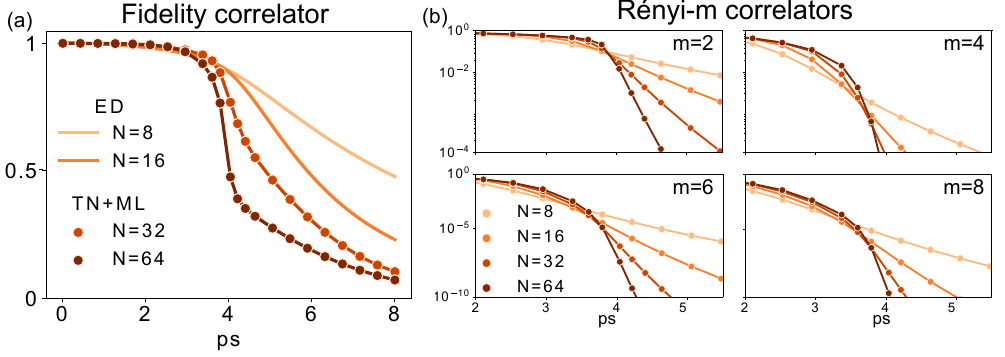}
    \caption{\textbf{Fidelity and Rényi-m correlators for the $\mathbb{Z}_2$-symmetric Brownian circuit with postselection:} {(a)} Steady-state value of the fidelity correlator $F_X(\frac{N}{4},\frac{3N}{4})$ for various system sizes $N$.
    For small systems, the results were obtained using exact diagonalization (ED) and for large systems a combination of tensor network and machine learning methods was employed (TN+ML).
    {(b)} The steady state value of the Rényi-m correlators $C^{(m)}(\frac{N}{4},\frac{3N}{4})$ [{defined in Eq.~(\ref{eq: Rényi-m})}] for $m=2,4,6,8$ for various system sizes close to the analytically predicted critical point at $p \cdot s=4$.
    Results were obtained by finding the ground state of the 2-copy effective Hamiltonian Eq.~\eqref{eq: TFI} using the Density Matrix Renormalization Group algorithm and then contracting the tensor network corresponding to the Rényi-correlator $C^{(m)}$. We used open boundary conditions and parameters $J=U=1$ and $s=1$.}
    \label{fig: fidelity and Rényis}
\end{figure*}

\subsection{Abelian \texorpdfstring{$\mathbb{Z}_2$}{Lg} symmetry}\label{subsec:abelianZ2}
We now construct the Brownian circuit for a system of qubits with strong $\mathbb{Z}_2$ symmetry generated by $\prod_iZ_i$. 
We will start by introducing the circuit and showing that we can understand SW-SSB in its steady state as a $\mathbb{Z}_2$ ferromagnet~\cite{moudgalya_symmetries_2023}. 
We will then extend the model by the postselection protocol described in Sec. \ref{sec: protocol} and show that postselection drives a phase transition of the steady state density matrix into a strongly symmetric phase.
\subsubsection{Model}
We consider the following random Hamiltonian on $N$ sites
\begin{equation}
H(t)=\sum_{\langle i,j\rangle}J_{ij}(t)X_iX_{j}+\sum_{i}U_i(t)Z_i, \label{eq: rando Ising}
\end{equation}
where $\langle i,j\rangle$ labels adjacent lattice sites and $J_{ij}(t),U_i(t)$ are Brownian random variables distributed according to Eq.~\eqref{eq: random couplings} with
$\mathbb{E}[J_{ij}(t)J_{i'j'}(t')]=\frac{2J \delta_{tt'}\delta_{ii'}\delta_{jj'}}{dt}$ and $\mathbb{E}[U_i(t)U_j(t')]=\frac{2U \delta_{tt'}\delta_{i j}}{dt}$.
% with $J_{ij}=J$ and $U_i=U$.
%
We leave the dimensions and boundary conditions unspecified for now, as our analytic calculations are independent of them, while our numerical simulations are for a one-dimensional system.
Following Ref.~\cite{moudgalya_symmetries_2023} the effective Hamiltonian $P_{\mathbb{Z}_2}$ in the doubled Hilbert space is given by: 
\begin{equation}
    P_{\mathbb{Z}_2}=2U\sum_j[1-Z_i^fZ_i^b]+2J\sum_{ij}[1-X_i^fX_i^bX_j^fX_j^b].
\label{eq:PZ2nomeas}
\end{equation}
This effective Hamiltonian has a $\mathbb{Z}_2^{N}$ symmetry as it commutes with $Z_i^fZ_i^b$ for every site $i$. 
When we take the initial state to be a product state in the positive parity sector, for example $\rho_0=\ket{\uparrow\ldots \uparrow}\bra{\uparrow\ldots \uparrow}$, we can restrict our analysis to the reduced Hilbert space $\widetilde{\mathcal{H}}=\mathrm{span}\{\bigotimes_i\opket{\widetilde{p}_i}, \; p_i \in \{\uparrow,\downarrow\}\}$, where we have defined $\opket{\widetilde{\uparrow}}_i=\ket{\uparrow}_{i,f}\otimes\ket{\uparrow}_{i,b}$ and $\opket{\widetilde{\downarrow}}_i=\ket{\downarrow}_{i,f}\otimes\ket{\downarrow}_{i,b}$. 
On this subspace $P_{\mathbb{Z}_2}$ acts as~\cite{moudgalya_symmetries_2023}
\begin{equation}
    P_{\mathbb{Z_2}}=2J\sum_{\langle i,j \rangle}[1-\widetilde{X}_j\widetilde{X}_i], \label{eq: P for Z2 base model}
\end{equation}
which is the Hamiltonian of the classical Ising model. 
The groundstate space is spanned by $\opket{\widetilde{+}\ldots\widetilde{+}}$ and $\opket{\widetilde{-}\ldots\widetilde{-}}$ and hence the steady state is given by $\opket{\rho}=\opket{\widetilde{+}}\opbraket{\widetilde{+}}{\rho_0}+\opket{\widetilde{-}}\opbraket{\widetilde{-}}{\rho_0}$. 
By ``devectorizing" $\opket{\rho}$ and using that $\rho_0$ was taken to be strongly symmetric and in the positive parity sector, we see that the steady-state density matrix is $\rho=\frac{1}{2^N}\mathbb{1}+\frac{1}{2^N}\prod_iZ_i$. 
This is simply the MMIS in the sector of positive parity.
To compute the Rényi-2 correlator, we pick $O_i=X_i$ as the order parameter
\begin{equation}
    C^{sw}(i,j)[\rho]=\frac{\Tr[\rho X_i X_j \rho X_j X_i]}{\Tr[\rho^2]}=\frac{\opbra{\rho}\widetilde{X}_i\widetilde{X}_j\opket{\rho}}{\opbra{\rho}\opket{\rho}}, \label{eq: Rényi-2 correlaotr Z2}
\end{equation}
which acts as the ferromagnetic correlation function on $\widetilde{\mathcal{H}}$.
The groundstates of $P_{\mathbb{Z}_2}$ in Eq.~\eqref{eq: P for Z2 base model} correspond to the averaged steady states of the circuit and spontaneously break the $\mathbb{Z}_2$ symmetry of the Ising model. 
Hence $C^{sw}(i,j)[\rho]=1$ is long-ranged ordered and the steady state has SW-SSB. 
Our direct calculation is in agreement with Eq.~\eqref{eq: formula R2 maximally mixed} where we diagnose SW-SSB for all MMIS.
\subsubsection{SW-SSB transition driven by postselection}
As we have demonstrated,  SW-SSB in the steady state can be reformulated as pure-state spontaneous symmetry breaking in the groundstate of $P_{\mathbb{Z}_2}$, which is in the ferromagnetic phase of the Ising model.
We now seek a quantum operation $\mathcal{E}_{dt}[\rho]$ that can be interspersed between the random gates that break the trace preservation and drives the system into the paramagnetic phase.
This can be achieved by the postselection protocol introduced in Sec.\ref{sec: protocol}.
As the target state, we pick $\ket{\Psi_\mathrm{t}}=\ket{\uparrow\ldots\uparrow}$, i.e., $\ket{\varphi}=\ket{\uparrow}$ in Sec.~\ref{sec: protocol}.
As the local observable to perform probabilistic postselection for, we then take $A_i=Z_i$, which commutes with the $\mathbb{Z}_2$-symmetry generated by $\prod_iZ_i$. 
With the postselection protocol described above, the density matrix evolves under the following quantum map:
\begin{align}
    &\mathcal{E}_{dt}[\rho] = (1-p\ dt\ N)\rho + \nonumber \\
    &\frac{p\ dt\ N}{N} \sum_i \left[(1-s)\left(\pi^+_i\rho\pi^+_i+\pi^-_i\rho\pi^-_i\right) +s\pi^+_i\rho\pi^+_i\right],
    \label{eq: postselection operation}
\end{align}
where $\pi^\pm_i = \frac{1\pm Z_i}{2}$. 
Importantly, this map is \textit{not} trace preserving for $s > 0$.
Thus we need to renormalize $\rho$ in order to obtain a valid density matrix.
Hence the object under consideration is not $\rho(t\to \infty) = \lim_{t \to \infty} \mathbb{E}[\rho(t)]$, but $\rho^{\mathrm{eq}}=\frac{\rho(t \to \infty)}{\Tr[\rho(t\to \infty)]}$ instead.
The difference between the two objects is not visible in the Rényi-2 correlator, as it has two copies of $\rho$ in the numerator and in the denominator.
 In the doubled Hilbert space, Eq.\eqref{eq: postselection operation} reads (see Appendix \ref{app:effH}):\
\begin{align}
    &\mathcal{E}_{dt}\opket{\rho} 
    =\left(1-dt\cdot P^{\mathrm{meas}}_{\mathbb{Z}_2} \right)\opket{\rho}\label{eq: effective Hamiltonian Z2 postselection}\\
    &P^{\mathrm{meas}}_{\mathbb{Z}_2}=\frac{(s-2)p}{4}\sum_i(\frac{s+2}{s-2}+Z_i^fZ_i^b) -
\frac{ps}{4}\sum_i(Z_i^f+Z_i^b).\nonumber
\end{align}
The full effective Hamiltonian for the total operation including the Brownian evolution becomes
$P_{\mathbb{Z}_2}^{\mathrm{total}}=P_{\mathbb{Z}_2}+P^{\mathrm{meas}}_{\mathbb{Z}_2}$. 
As before, we can restrict the analysis to the subspace $\widetilde{\mathcal{H}}$ [see discussion above Eq.~\eqref{eq: P for Z2 base model}].
Here, $P_{\mathbb{Z_2}}^{\mathrm{total}}$ acts as:
\begin{equation}
P_{\mathbb{Z_2}}^{\mathrm{total}}=2J\sum_{\langle i,j \rangle}[1-\widetilde{X}_j\widetilde{X}_i]-\frac{ps}{2}\sum_i\widetilde{Z}_i +\frac{psN}{2}, \label{eq: TFI}
\end{equation}
which is the Hamiltonian of the transverse-field Ising model.
Hence, SW-SSB in the steady state of the Brownian circuit maps to spontaneous symmetry breaking in the ground state of $P_{\mathbb{Z_2}}^{\mathrm{total}}$, with the Rényi-2 correlator acting as the ferromagnetic correlation function.
The transverse field Ising model exhibits a quantum phase transition at $\left|\frac{ps}{4J}\right|=1$~\cite{pfeuty_one-dimensional_1970}.
It follows that the averaged density matrix $\mathbb{E}[\rho]$ exhibits a phase transition \textit{in the steady state} from an SW-SSB phase (ferromagnet) and a strongly symmetric phase (paramagnet). 

\subsubsection{Numerical results on the Rényi-m and fidelity correlators}

We now numerically investigate the effective Brownian dynamics with postselection and focus on a one-dimensional chain with periodic boundary conditions; see Fig.~\ref{fig: Overview} (c,d).
Our simulations show excellent agreement with the predictions from the analytical results on the transverse field Ising model for the Rényi-2 correlator.
The survival probability of a single trajectory scales as $\left({1}/{2}\right)^{psNT}$ where $T$ is the time simulated and $N$ is the size of the linear chain, which limits the accessible system sizes for the real time evolution. 
We also numerically study the fidelity correlator introduced in Eq.~\eqref{eq: fidelity correlator}.
To do so, we introduce a family of correlators $C^{(m)}(i,j)$, {which we refer to as the Rényi-m correlators} (similar {but slightly distinct} quantities have been defined in \cite{zhang_probing_2025, lessa_strong--weak_2024}).
Let $\sigma_{ij}=O_i^\dagger O_j\rho O_j^\dagger O_i$. 
Then the Rényi-m correlator $C^{(m)}(i,j)$ is defined as:
\begin{align}
    C^{(m)}(i,j)[\rho]=\frac{\Tr[\left(\rho \sigma_{ij}\right)^{\frac{m}{2}}]}{\Tr[\rho^m]} \label{eq: Rényi-m}
\end{align}
Note that $C^{(2)}(i,j)=C^{sw}(i,j)$, the Rényi-2 correlator. 
The Rényi-1 correlator is simply the fidelity correlator, i.e., $C^{(1)}(i,j)=F_O(i,j)$, since $\Tr[\sqrt{\sqrt{\rho}\sigma_{ij} \sqrt{\rho}}]=\Tr[\sqrt{\rho \sigma_{ij}}]$ even if $\rho$ and $\sigma_{ij}$ do not commute~\cite{müller_2023_simplified}.\footnote{A brief sketch of the argument in \cite{müller_2023_simplified} goes as follows: $\Tr[\sqrt{\sqrt{\rho}\sigma \sqrt{\rho}}]=\sum_i\sqrt{\lambda_i}$, where $\lambda_i$ are eigenvalues of $\sqrt{\rho}\sigma \sqrt{\rho}$. Since the spectrum of a product of matrices is invariant under cyclic permutations, the result follows.}
For even $m$, $C^{(m)}(i,j)$ can be efficiently computed for the averaged steady state with tensor methods~\cite{hauschild_efficient_2018}.
Specifically, we first use the Density Matrix Renormalization Group  (DMRG) algorithm to find the groundstate of the effective Hamiltonian Eq.~\eqref{eq: TFI} in the doubled Hilbert space, convert it into a density matrix and then perform the contraction of the network corresponding to $C^{(m)}$, where we truncate small singular values.
Strikingly, our results show transitions in $C^{(m)}(i,j)$ for various values of $m$ that are consistent with the analytically obtained phase transition of the Rényi-2 correlator at $ps=4$ [see Fig. \ref{fig: fidelity and Rényis}b].
This suggests that the same transition is also present in higher copy quantities, which is consistent with the fact that we are studying the same underlying averaged steady-state density matrix. 
The fidelity correlator, however, is in general not easy to obtain using tensor network methods.
To estimate the fidelity for large system sizes, we combine a tensor-based calculation of $C^{(m)}$ for even $m$ with tools from machine learning. 
Our procedure is as follows: We first train a neural network to predict the fidelity $F_X(i,j)$ from the higher Rényi correlators $C^{(m)}(i,j)$, the postselection rate $p\cdot s$ and the system size $N$.
Training data is provided by exact diagonalization of the transverse field Ising model Eq. \eqref{eq: TFI} for $N\leq 18$.
Throughout this process, we use $m=2,4,6,8$ and $i=\frac{N}{4}$, $j=\frac{3N}{4}$.
Using data for $C^{(m)}$ obtained from tensor network methods (see Fig. \ref{fig: fidelity and Rényis}b) for $N=32,64$ we use the trained neural network to predict $F_X$ for these larger system sizes.
To get a robust prediction, we repeat this process of training and predicting $10^3$ times and average over the predictions. 
The results are shown in Fig. \ref{fig: fidelity and Rényis}a.
Finally, using this extrapolation to large system sizes $N$, our approach predicts a critical point of the fidelity correlator that is consistent with the one of the Rényi correlators.
In App.~\ref{app:fidelity} we provide an analytic argument for  this behavior by mapping the fidelity to a series of n-point correlation functions in the same statistical mechanics model where the Rényi-2 correlator acts as the 2-point correlation function. From this, we argue that all Rényi correlators and the fidelity capture the SW-SSB from $G\times G \to G_\text{diag}$ of the averaged density matrix.
\subsubsection{Identical transition in distinct correlators}\label{subsubsec:identtrans}
The fact that the transitions in the fidelity correlator matches that in the Rényi-2 correlator is in stark contrast to previous examples in finite-depth quantum channels, where both correlation functions predicted different critical points~\cite{lessa_strong--weak_2024, sala_spontaneous_2024}. 
Our understanding of this difference is the following.
The steady state of a $D$-dimensional quantum system maps to the ground state of a $D$-dimensional quantum model (in our case the $1d$ transverse-field Ising model). Equivalently, it can be expressed in terms of the partition function of a $(D+1)$-dimensional classical statistical mechanics model (in our case the 2d classical Ising model).
The Rényi-m correlation functions are expectation values of local operators in this ground state, which, in the classical statistical mechanics model, are operator insertions located at the boundary of the $(D+1)$-dimensional model.
Hence, in this language, the choice of correlation function {evaluated using boundary operators} should not influence the critical behavior, which is determined by the bulk {statistical mechanics model}.
By contrast, when a fixed quantum channel $\mathcal{E}$ to a quantum state $\rho$ in $D$ dimensions is applied, the correlation functions in the state $\mathcal{E}[\rho]$ can be understood using a $D$-dimensional classical statistical mechanics model~\cite{sala_spontaneous_2024, lessa_strong--weak_2024} that depends on the particular kind of correlation function being evaluated.
Hence in that case, the bulk model itself changes depending on the correlation function, and it is natural to expect that these different bulk models can lead to different locations of the  critical point.
\subsection{Non-Abelian \texorpdfstring{$S_3$ symmetry}{Lg}}
\label{sec: Potts model}
We will now present an example of a similar phase transitions in circuits with a non-abelian symmetry.
The group we will consider is the symmetric group in three elements $S_3 \cong \mathbb{Z}_3 \ltimes \mathbb{Z}_2$.
A well-known $S_3$ symmetric model is the three-state Potts model \cite{Baxter_exactly, RevModPhys.54.235,friedman_localization-protected_2018}.
To this end, consider a 3-level local Hilbert space $\mathcal{H}_{\mathrm{loc}}=\mathbb{C}^3$ and the operators
 \begin{equation}
     \sigma=\begin{pmatrix}
         0 &1&0\\
         0&0&1\\
         1&0&0
     \end{pmatrix},\;\tau = \begin{pmatrix}
         1&0&0\\
         0&\omega &0\\
         0&0&\omega^2
     \end{pmatrix}, \;
     \chi = \begin{pmatrix} 
1 & 0 & 0\\
0 & 0 & 1\\
0 & 1 & 0
\end{pmatrix},
 \end{equation}
where $\omega = e^{\frac{2i\pi}{3}}$.
The $S_3$ symmetry is then generated by the operators $\prod_i \chi_i$ and $\prod_i \tau_i = \omega^{\sum_iQ_i}$, where the $\mathbb{Z}_3$-charge Q is given by
\begin{equation}
    Q_i=\frac{i}{\sqrt{3}}\left(\tau^\dagger_i-\tau_i\right).
\label{eq:Qidefn}
\end{equation}
$Q_i$ has three eigenvalues, $\{0,+1, -1\}$, which are thus conserved modulo 3. 
We label the basis states of the local Hilbert space by the respective eigenvalues of $Q_i$, i.e., $\mathcal{H}_{\mathrm{loc}}= \mathrm{span}\{\ket{0},\ket{+},\ket{-}\}$.
Then, $\prod_i\chi_i$ serves as a charge conjugation symmetry exchanging $\ket{+}\leftrightarrow \ket{-}$ or alternatively, the operators transform as $\sigma_i\leftrightarrow \sigma_i^\dagger$ and $\tau_i \leftrightarrow \tau^\dagger_i$ under conjugation by it. 
We introduce the Brownian three-state $S_3$-symmetric Potts model 
\begin{equation}
H(t)=\sum_iU_i(t)\left(\tau_i+\tau^\dagger_i\right)+\sum_{\langle i,j\rangle}J_{ij}(t)\left(\sigma_i^\dagger\sigma_j+\sigma_i\sigma_j^\dagger\right), \label{eq: Potts}
\end{equation}
where the random variables $J_{ij}$ and $U_i$ are distributed as in Eq.~\eqref{eq: random couplings}:
\begin{align}
    \mathbb{E}[U_i(t)]=0\quad &\mathbb{E}[U_{i}(t)U_{i'}(t')]=\frac{2U\delta_{tt'}\delta_{i i'}}{dt} \\ 
    \mathbb{E}[J_{ij}(t)] =0 \quad &\mathbb{E}[J_{ij}(t)J_{i'j'}(t')]=\frac{2J \delta_{tt'}\delta_{i i'}\delta_{jj'}}{dt}. \nonumber
\end{align}
\begin{figure}
    \centering
    \includegraphics[width=1.0\linewidth]{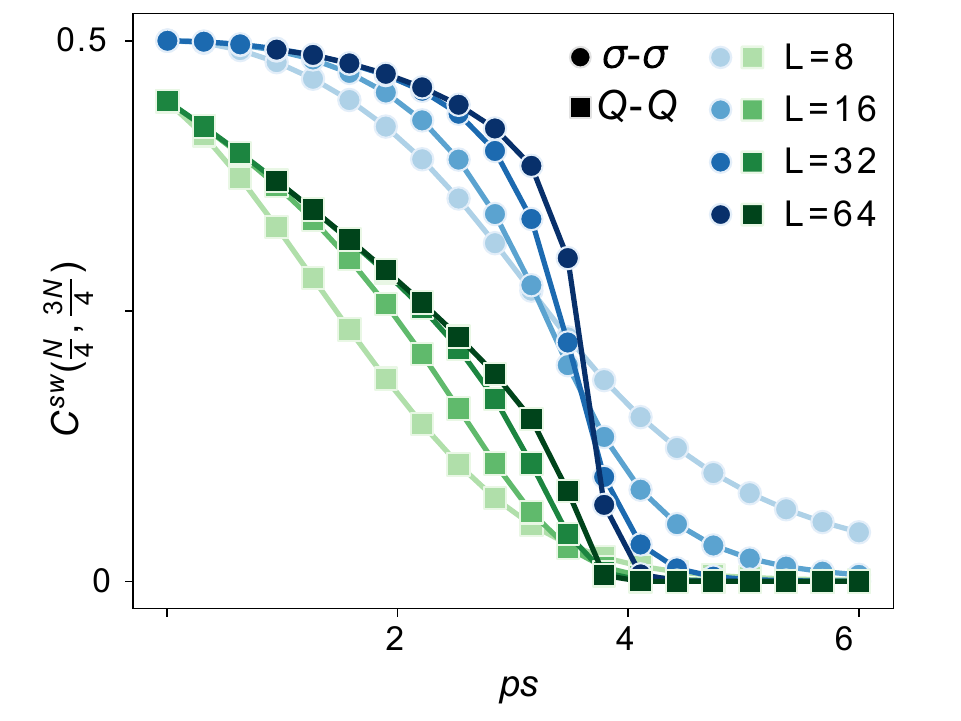}
    \caption{\textbf{Steady-state Rényi-2 correlators for the non-abelian $S_3$ Potts circuit.} Results were obtained using the Density Matrix Renormalization Group (DMRG) algorithm to find the groundstate of Eq.~\eqref{eq: Potts-2copy} with open boundary conditions. We show both $\sigma-\sigma$ correlations, Eq.~\eqref{eq: sigma sigma}, and $Q-Q$ correlations, Eq.~\eqref{eq: Q Q}, with $i=\frac{N}{4}$ and $j=\frac{3N}{4}$. Parameters are $J=U=1$ and $s=0.5$.}
    \label{fig: S3 DMRG}
\end{figure}

We follow the protocol described in Sec.~\ref{sec: protocol} to drive a phase transition as a function of the postselection rate.
As a target state we pick $\ket{\Psi_\mathrm{t}}=\ket{0\ldots 0}$ and measure the local operator $A_i = Q_i^2$.
This operator has eigenvalues $0$ and $1$ with projections $\pi^1 = Q_i^2$ and $\pi^0_i=1-Q_i^2$.
We then postselect onto the ``$0$" result, which drives the system towards the target state.  
The quantum operation describing the above protocol is given by: 
\begin{align}
    \mathcal{E}_{dt}[\rho] &= (1-p\ dtN)\rho \nonumber \\
    &+\frac{p\ dt N}{N} \sum_i \left[(1-s)\left(\pi^0_i\rho\pi^0_i+\pi^1_i\rho\pi^1_i\right) +s\pi^0_i\rho\pi^0_i\right].
    \label{eq: postselection operation S3}
\end{align}
Then, the entire effective Hamiltonian including the Brownian part in the doubled Hilbert space is (up to an overall constant, see Appendix \ref{app:effH}):
\begin{align}
    P_{S_3}^{\mathrm{tot}}&= H^f_{\mathrm{Potts}}(U',J)+H^b_{\mathrm{Potts}}(U',J)\label{eq: Potts-2copy}\\
    &-2J\sum_{\langle i,j\rangle}\left(\sigma_{i}^{f \dagger}\sigma_{j}^f+\sigma_{i}^f \sigma_{j}^{f \dagger}\right)\left(\sigma_{i}^{b \dagger}\sigma_{j}^b + \sigma_{i}^b \sigma_{j}^{b \dagger}\right) \nonumber\\
      &-2\left(U+\frac{p(s-2)}{9}\right)\sum_i\left(\tau_{i
    }^f+\tau^{f \dagger}_{i}\right)\left(\tau_{i}^b + \tau^{b \dagger}_{i}\right)\nonumber, 
\end{align}
where $H^{f/b}_{\mathrm{Potts}}(U',J)$ are Potts Hamiltonians of the form of Eq.~\eqref{eq: Potts}, but with constant coefficients $U_i(t)=U'=U+\frac{p(1-2s)}{9}$ and $J_{ij}(t)=J$ acting on the forward/backward copy only.
To find a local order parameter, notice that under the $\prod_j\chi_j$ generator $\sigma_i \to \sigma_i^\dagger$ and $Q_i \to -Q_i$ while under the $\prod_j\tau_j$ generator $\sigma_i \to \omega^2 \sigma_i$ and $Q_i$ remains unchanged. Hence $\sigma_i$ transforms nontrivially under the action of the entire symmetry group, while $Q_i$ only transforms nontrivially under the $\mathbb{Z}_2$ subgroup generated by $\prod_j\chi_j$. For completeness we look at both order parameters and define the corresponding Rényi-2 correlators:
\begin{align}
    &{\sigma-\sigma\;\mathrm{correlations:}\quad C^{sw}(i,j)[\rho]=\frac{\Tr[\rho \sigma_i^\dagger\sigma_j \rho \sigma_j^\dagger\sigma_i]}{\Tr[\rho^2]}} \label{eq: sigma sigma}\\
    &{Q-Q\;\mathrm{correlations:}\quad C^{sw}(i,j)[\rho]=\frac{\Tr[\rho Q_i^\dagger Q_j \rho Q_j^\dagger Q_i]}{\Tr[\rho^2]}}\label{eq: Q Q}.
\end{align}
In contrast to the $\mathbb{Z}_2$-symmetry studied before, we do not find an analytic solution for the ground state and the transition. 
We instead resort to determine the ground state of Eq.~\eqref{eq: Potts-2copy} numerically using tensor network simulations~\cite{hauschild_efficient_2018}.
%
%\subsubsection{Numerical results}
Increasing the postselection rate $ps$, we find a phase transition from the SW-SSB phase to the strongly symmetric phase with a critical point at $ps \approx 3.75$; see Fig.~\ref{fig: S3 DMRG}. 
Both order correlation functions appear to indicate the same phase transition. 

Let us now compare our formula for the Rényi-2 correlator Eq.~\eqref{eq: formula R2 maximally mixed} with the numerical data at $s=0$. 
The local Hilbert space dimension is $d=3$ and
$\sigma$ transforms under a two-dimensional representation of $S_3$ since under charge conjugation $\sigma \to \sigma^\dagger$ which yields $|\mathcal{I}_\sigma|=2$.
Because $\norm{\sigma}_2^2 = 3$ we find from Eq.~\eqref{eq: formula R2 maximally mixed} that for the $\sigma-\sigma$ correlations $C^{sw}=\frac{\norm{\sigma}_2^4}{d^2|\mathcal{I}_\sigma|}=\frac{3^2}{3^2\cdot 2}=\frac{1}{2}$. 
When we take the charge operator $Q$ as the local order parameter, we have $\norm{Q}_2^2=2$ and $|\mathcal{I}_Q|=1$ since $Q$ transforms under the 1d sign representation of $S_3$.
Hence in this case we find $C^{sw}=\frac{\norm{Q}_2^4}{d^2|\mathcal{I}_Q|}=\frac{2^2}{3^2\cdot 1}=\frac{4}{9}$.
Both results are consistent with the numerical results.
\section{Absence of phase transitions in steady states of quantum channels}
\label{sec: Lindbladian}
\subsection{General argument}
We will now argue that a transition to a strongly symmetric phase is absent in the steady state of general quantum channels, which are quantum operations that are trace-preserving.
To illustrate the logic, consider the $\mathbb{Z}_2$-symmetric circuit with postselection introduced in Sec.\ref{subsec:abelianZ2}. 
In SW-SSB phase in the absence of postselection at $ps=0$ the steady state was twofold degenerate with the two steady states given by the MMIS $\mathbb{1}\pm \prod_IZ_i$, while in the strongly symmetric phase at $J=0$  the steady state was unique and is given by $\ket{\uparrow \ldots \uparrow}\bra{\uparrow \ldots \uparrow}$.
This is the typical situation for spontaneous symmetry breaking, where the symmetry-broken phase has a higher degeneracy than the symmetric phase (Landau-type transition).

Our argument is based on the expectation that a phase transition from a SW-SSB phase to a strongly symmetric phase requires a change in the steady-state degeneracy from a degenerate phase (SW-SSB) to a phase with lower degeneracy (stronlgy symmetric phase).
However, in the following, we show that the degeneracy cannot be lowered unless trace preservation is violated. 
To this end, we use that the steady state of a quantum operation $\mathcal{E}$ is a state $\rho$ such that $\mathcal{E}[\rho]=\rho$. 
The structure of the argument for the degeneracy then closely follows \cite{zhang_stationary_2020,moharramipour_symmetry_2024}.
Given a unitary representation of the symmetry group $G$, the Hilbert space decomposes into irreducible representations $V_\lambda$ (irreps) of $G$ as follows:
\begin{equation}
    \mathcal{H}=\bigoplus_\lambda m_\lambda V_\lambda, \label{eq: decomposition}
\end{equation}
where $m_\lambda$ is the multiplicity of $V_\lambda$ and $d_\lambda = \dim V_\lambda$.
$\mathcal{H}$ has an orthonormal basis in terms of the decomposition in Eq.~\eqref{eq: decomposition} given by states of the form $\ket{\lambda,s_\lambda,\alpha_\lambda}$, where $\lambda$ labels the irrep, $s_\lambda$ a state within $V_\lambda$ and $1\leq \alpha_\lambda \leq m_\lambda$ labels the copy of $V_\lambda$ within $m_\lambda V_\lambda$.

Let $\mathcal{E}[\rho]=\sum_iK_i\rho K_i^\dagger$ be a strongly symmetric quantum channel, i.e., a trace-preserving quantum operation. Since $\sum_iK_iK_i^\dagger=\mathbb{1}$, the adjoint map $\mathcal{E}^*[\rho]=\sum_iK_i^\dagger \rho K_i$ is unital.
Since it is strongly symmetric, the operators $\rho_{\lambda,s, p}=\sum_{\alpha_\lambda}\ket{\lambda,s_\lambda,\alpha_\lambda}\bra{\lambda,p_\lambda,\alpha_\lambda}$ commute with all Kraus operators $K_i$ and since $\mathcal{E}^*$ is unital it follows that 
\begin{equation}
    \mathcal{E}^*[\rho_{\lambda, s, p}]=\sum_iK_i^\dagger \; \rho_{\lambda, s,p}K_i=\rho_{\lambda, s,p}\sum_iK_i^\dagger K_i=\rho_{\lambda, s,p}. \label{eq: steady state unital}
\end{equation}
Hence, $\mathcal{E}^*$ has at least $\sum_\lambda (d_\lambda)^2$ steady states.\footnote{Note that $\rho_{\lambda,s,p}$ is not a density matrix for $p\neq s$, since it is traceless. However, adding $\rho_{\lambda,s,p}$ to any linearly independent steady state density matrix results in a new linearly independent steady state and hence the unphysical nature of $\rho_{\lambda,s,p}$ for $s\neq p$ does not impact the steady-state degeneracy.} Since $\mathcal{E}^*$ and $\mathcal{E}$ have the same steady state degeneracy, it follows that the steady-state degeneracy of $\mathcal{E}$ is also lower bounded by $\sum_\lambda (d_\lambda)^2$. 
The lower bound is, for example, realized in the steady state of strongly-symmetric Brownian circuits since the dimension of the commutant $\mathcal{C}$ defined above Eq.~\eqref{eq: commutant steady state} is precisely given by $\sum_\lambda (d_\lambda)^2$ \cite{moudgalya_symmetries_2023,li_highly-entangled_2024}.  
Hence, any strongly $G$-symmetric trace-preserving quantum operation has at least the same steady-state degeneracy as a $G$-symmetric Brownian circuit, which shows SW-SSB, see Fig.~\ref{fig: where transition} for an illustration.
Since a phase transition from a symmetry broken to a symmetric phase is associated with a \textit{decrease} in the steady-state degeneracy, a Landau-type phase transition from the SW-SSB phase to a strongly symmetric phase is absent in the steady state of a trace-preserving quantum channel. 
This, of course, does not rule out the existence of thermal phase transitions in the steady state of Lindbladians. In that case, the weak symmetry is further broken down and the degeneracy of the steady state \textit{increases}.

\subsection{Concrete example}\label{subsec:Lindbladexample}
To illustrate the above argument we will now consider an example of a $\mathbb{Z}_2$-symmetric circuit subjected to a measurement and feedback protocol which seemingly follows the same principles as the probabilistic postselection protocol described in Sec.~\ref{sec: protocol}, but does so by preserving the trace of the density matrix.
We will then show numerically, that this model does not exhibit a robust strongly-symmetric phase in terms of the Rényi-2 correlator.
\begin{figure}
    \centering
    \includegraphics[width=1\linewidth]{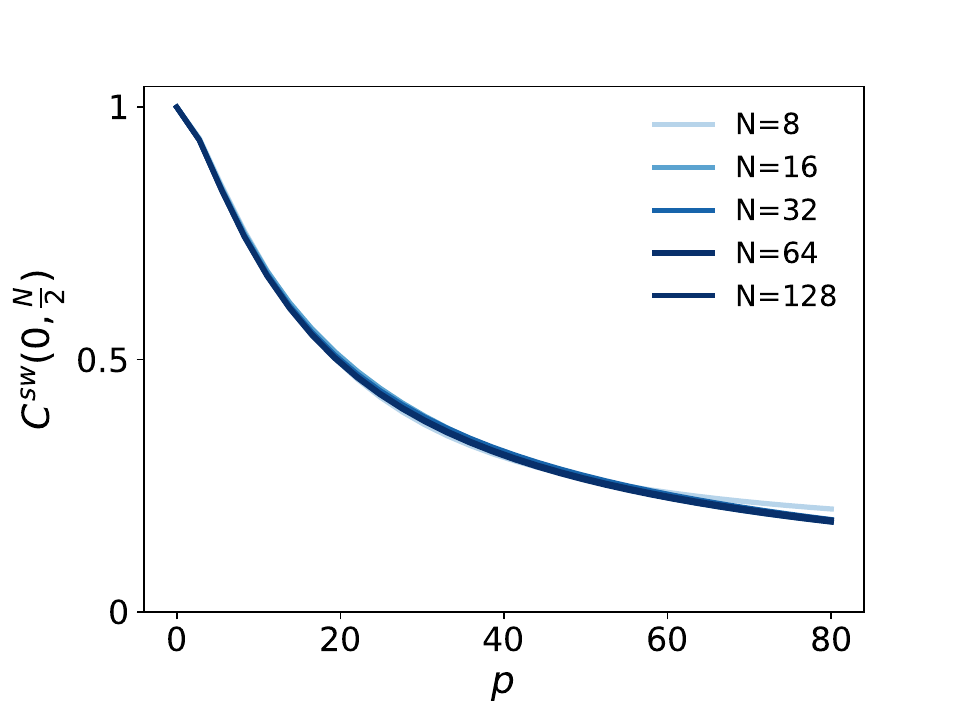}
    \caption{\textbf{Absence of a SW-SSB phase transition in the steady state of an adaptive, trace-preserving model.} The  Rényi-2 correlator smoothly decays with measurement probability $p$, which we obtained using the effective Hamiltonian expressed in terms of the total spin Eq.~\eqref{eq: total spin hamiltonian}. The curves for $N=16,32,64,128$ are on top of each other. Parameters are $J=U=1$.}
    \label{fig: adaptive Rényi}
\end{figure}
Concretely, we consider the $\mathbb{Z}_2$ circuit, characterized by the random Hamiltonian in Eq.~\eqref{eq: rando Ising}, and choose $\ket{\uparrow\ldots \uparrow}$ as the target state. 
We will now argue that any strongly $\mathbb{Z}_2$ symmetric trace-preserving protocol that targets \textit{only} this state necessarily leads to all-to-all interaction in the effective Hamiltonian acting in the doubled Hilbert space.
We showed in Sec.~\ref{sec: phase transition} that under $\mathbb{Z}_2$ symmetric Brownian dynamics, the vectorized density matrix $\opket{\rho}$ is confined to a subspace $\widetilde{\mathcal{H}} \subset \mathcal{H}_f\otimes \mathcal{H}_b$ (see discussion above Eq.~\eqref{eq: P for Z2 base model}). 
Recall that $\widetilde{\mathcal{H}}=\mathrm{span}\{\bigotimes_i\opket{\widetilde{p}_i}, \; p_i \in \{\uparrow,\downarrow\}\}$, where we have defined $\opket{\widetilde{\uparrow}}_i=\ket{\uparrow}_{i,f}\otimes\ket{\uparrow}_{i,b}$ and $\opket{\widetilde{\downarrow}}_i=\ket{\downarrow}_{i,f}\otimes\ket{\downarrow}_{i,b}$.
In terms of the unvectorized density matrix $\rho$, this is equivalent to saying that $\rho$ is diagonal in the $Z$-basis, i.e., it can be expressed as a mixture of states of the form $\ket{b_1 \ldots b_N}$ with $b_i \in \{\uparrow, \downarrow\}$.
Since the initial state lies in the sector of positive parity, the state remains in this sector and we are guaranteed an even number of $\downarrow$s.
Targeting the state $\ket{\uparrow \ldots \uparrow}$ now corresponds to locally flipping $\downarrow$ to $\uparrow$.
In order to respect the strong $\mathbb{Z}_2$ symmetry, we are restricted to correct pairs of $\downarrow\downarrow$ to $\uparrow \uparrow$.
Targeting the state $\ket{\uparrow \ldots \uparrow}$ necessarily involves correcting configurations of the form $\ket{\ldots \downarrow_i \ldots \downarrow_j \ldots} \to \ket{\ldots \uparrow_i \ldots \uparrow_j \ldots}$ for all pairs of sites $(i,j)$, leading to an all-to-all coupling.
We also introduce all-to-all couplings for the Brownian dynamics in our concrete example, which helps us to solve the model explicitly.
As before, we take the initial state to be $\rho_0 = \ket{\uparrow \ldots \uparrow}\bra{\uparrow \ldots \uparrow}$.
In addition, consider the following measurement protocol on a system of $N$ qubits:
\begin{enumerate}
    \item With probability $p\frac{dt\ N(N-1)}{2}$ measure the operators $Z_i$ and $Z_j$ for a randomly uniform chosen pair $i\neq j$.
    \item If the measurement outcomes correspond to $\downarrow$ for both measurements, apply the gate $X_i X_j$, thereby flipping both spins. This increases the probability of flowing towards the target steady state $\ket{\uparrow\ \cdots\ \uparrow}\bra{\uparrow\ \cdots\ \uparrow}$.
    For any other measurement  outcome, do nothing.
\end{enumerate}
The normalization of the probability is such that $p$ has the interpretation of a measurement rate per pair.
Since the initial state lies in the positive parity sector, we are always guaranteed an even number of $\uparrow$'s. Then the protocol drives the system towards the target state $\ket{\uparrow \ldots \uparrow}$ (similar to the postselection model) by correcting configurations of the form $\downarrow \downarrow$ to $\uparrow \uparrow$.
Since this protocol described above is trace-preserving, the dynamics become Lindbladian for $dt \to 0$.
The jump operators are given by $\pi^+_i\pi^+_j$, $\pi^+_i\pi^-_j$, $\pi^-_i\pi^+_j$ and $X_iX_j \pi^-_i\pi^-_j$, where $\pi_i^\pm = \frac{1\pm Z_i}{2}$ (see Appendix \ref{app:effH}).
By following the same steps as to derive Eq.~\eqref{eq: TFI}, we can express the time evolution of the vectorized density matrix $\opket{\rho}$ in the doubled Hilbert space in terms of an effective non-hermitian Hamiltonian $P^{\mathrm{tot}}_{\mathbb{Z}_2}$ (see Appendix \ref{app:effH}).
As before we can restrict the analysis to the reduce Hilbert space $\widetilde{\mathcal{H}}$, defined above Eq.~\eqref{eq: P for Z2 base model}. Here $P^{\mathrm{tot}}_{\mathbb{Z}_2}$ acts as:

\begin{align}  &P^{\mathrm{tot}}_{\mathbb{Z}_2}=2J\left(N^2-4\widetilde{S}^x\widetilde{S}^x\right)\nonumber\\
    &-\frac{p}{2}\left(N\widetilde{S}^z-\widetilde{S}^z\widetilde{S}^z+\widetilde{S}^+\widetilde{S}^+-\frac{N^2}{4}\right) -\frac{pN}{4}+\frac{p}{2}\widetilde{S}^z.  \label{eq: total spin hamiltonian}
\end{align}
Here, $\widetilde{S}^{x}=\frac{1}{2}\sum_i \widetilde X_i$ (similarly for $\widetilde{S}^{y,z}$) are the total spin operators and $\widetilde{S}^+=\widetilde{S}^x +i\widetilde{S}^y$ is the spin raising operator.
The initial state $\opket{\rho_0}=\ket{\uparrow \ldots \uparrow}_f\otimes \ket{\uparrow \ldots \uparrow}_b = |{\widetilde{\uparrow} \ldots \widetilde\uparrow{\widetilde\uparrow}}\rangle$ lies in the sector of maximal total spin $\widetilde{S}=\frac{N}{2}$, and $P^{\mathrm{tot}}_{\mathbb{Z}_2}$ commutes with $\widetilde{\mathbf{S}}\cdot \widetilde{\mathbf{S}}$.
Hence the dynamics remains confined to the maximal spin sector with $\widetilde{S}=\frac{N}{2}$ for all times which greatly simplifies numerically determining the Rényi-2 correlator of the steady state $\lim_{t\to \infty}e^{-P_{\mathbb{Z}_2}^{\mathrm{tot}} t}\opket{\rho_0}$.
The result is plotted in Fig.~\ref{fig: adaptive Rényi}.
Contrary to the previous examples with postselection, we do not find a critical point, and  $C^{sw}$ decays smoothly to zero for large measurement rates. 
We observe similar behavior  for the fidelity, see App.~\ref{app:fidelity}.
Hence, an extended strongly-symmetric phase is absent in this model with active control, consistent with our general considerations that a Landau type phase transition is absent in trace-preserving dynamics. 
To see how our previous discussion applies to this model, we analyze the groundstate degeneracy $P_{\mathbb{Z}_2}^{\mathrm{tot}}$ in the two limits $p=0, J\neq 0$ and $J=0,p\neq 0$.
In the first case, the steady state is twofold degenerate with the two steady states being given by $\mathbb{1}\pm \prod_iZ_i$, which is the lower bound for the degeneracy of any $\mathbb{Z}_2$ symmetric Lindbladian.
In the other limit, the target state $\ket{\uparrow \ldots  \uparrow}$ is not the only steady state. All pure states of the form $\ket{\uparrow  \ldots \downarrow \ldots \uparrow}$ that feature a single downward pointing spin are also left unchanged by the measurement and feedback protocol, leading to a $N+1$ fold degeneracy.
Hence, even in the case of very strong measurements, the system is not able to break the steady-state degeneracy of the SW-SSB phase, required for Landau type SW-SSB transitions.
\section{Summary and Outlook}
\label{sec: discussion}
In this work, we have demonstrated the existence of Strong-to-Weak Spontaneous Symmetry Breaking (SW-SSB) phases and phase transitions in steady-states of quantum operations.
While SW-SSB is exhibited by the so-called Maximally Mixed Invariant States (MMIS), which are realized as averaged steady states of noisy symmetric unitary quantum evolutions, we showed that SW-SSB is stable to the addition of quantum operations such as measurements and postselection.
Further, we found that breaking trace preservation is a necessary condition for Landau-type SW-SSB transitions to occur, which can be achieved by postselection only.
With this knowledge, we constructed concrete protocols to realize the SW-SSB phase and transitions out of that phase for both abelian and non-abelian on-site symmetries by mapping the averaged steady state of certain noisy quantum operations onto ground states of effective quantum Hamiltonians.
On a technical note, we also derived a simple formula for the Rényi-2 correlator of a large class of MMIS of symmetry sectors of compact Lie groups and finite groups, and we proved rigorously that these states have SW-SSB. 
This opens up a number of avenues for future explorations. 
First, these MMIS can be interpreted as states with infinite chemical potential.
A pertinent next step in exploring the landscape of SW-SSB states is hence to consider states at different chemical potentials.
For abelian symmetries, these would amount to some convex combinations of the MMIS, and hence such states have SW-SSB (see Appendix \ref{sec: convex combinations} for a proof).
However, in the case of non-Abelian symmetries, general steady states could also be those that are not convex combinations of MMIS, i.e., those that are not reachable from strongly symmetric initial states. 
We have not studied such states in this work, and understanding their structure would be interesting future work. 
Morevoer, while we have only explored MMIS with on-site symmetries in this work, recent works have revealed numerous novel types of symmetries that are not of that form, including non-invertible symmetries~\cite{shao2023noninvertible}, Hilbert space fragmentation~\cite{moudgalya_hilbert_2022, Sala2020, khemani2020hsf, Rakovszky20, li_hsf_2023, moudgalya2021review}, and quantum many-body scars~\cite{moudgalya_scars_2024, serbyn2020review, moudgalya2021review, papic2021review}, and it would be interesting to understand MMIS in such settings as well, where steady states have been demonstrated to show interesting properties~\cite{li_hsf_2023, li_highly-entangled_2024, sahu2025entanglement}. 
Characterizing some of these examples would also involve understanding the SW-SSB of MMIS in symmetry sectors of sizes that are not growing exponentially with system size (e.g., these are ubiquitous in the study of Hilbert space fragmentation~\cite{Sala2020, khemani2020hsf, Rakovszky20, moudgalya2019krylov, yang2020hsf,  moudgalya2021review}), which are not covered by our results. 
There are also a number of directions to pursue with regard to phase transitions.
First, while in this work we have focused on SW-SSB phases and transitions with discrete symmetries, it would be interesting to understand them in the presence of continuous symmetries such as $U(1)$ and $SU(2)$.
SW-SSB for continuous symmetries is closely related to hydrodynamics~\cite{ogunnaike_unifying_2023,moudgalya_symmetries_2023, huang_hydrodynamics_2024}. Hence a transition out of the SW-SSB phase should then be associated with an interesting breakdown of hydrodynamics.
Second, although we showed in this work that continuous Landau-type transitions cannot arise without postselection, other kinds of transitions might exist, including dynamical ``freezing" transition in dipole-conserving systems~\cite{Morningstar2020, Feldmeier2021}. It would be interesting to study the fate of SW-SSB across such a transition.
In this context too, it would be pertinent to characterize MMIS in scalar symmetry sectors that grow sub-exponentially with system size.
Finally, on the practical front, in the future it will be interesting to explore measurement protocols for the fidelity and Rényi correlators that can experimentally detect such transitions. 
There have been proposals to obtain the Rényi-2 correlator via randomized measurements \cite{sun_scheme_2025}.
Even though the sample complexity of these methods usually scales exponentially with system size, they have been successfully used in experiments for non-linear functions of the density matrix \cite{elben_mixed-state_2020, brydges_probing_2019, satzinger_2021}, because the exponents are favorable to those of full-state tomography.

\begin{acknowledgments}
We thank Pablo Sala for useful discussions on SW-SSB, and Gustav Berth for discussions on representation theory. S.M. thanks Lesik Motrunich for a previous collaboration on \cite{moudgalya_symmetries_2023}. We acknowledge support from the Deutsche Forschungsgemeinschaft (DFG, German Research Foundation) under Germany’s Excellence Strategy--EXC--2111--390814868, TRR 360 – 492547816 and DFG grants No. KN1254/1-2, KN1254/2-1, the European Research Council (ERC) under the European Union’s Horizon 2020 research and innovation programme (grant agreement No. 851161), the European Union (grant agreement No. 101169765), as well as the Munich Quantum Valley (MQV), which is supported by the Bavarian state government with funds from the Hightech Agenda Bayern Plus. 
\end{acknowledgments}

{\textbf{Data availability.---}}Data, data analysis, and simulation codes are available upon reasonable request on Zenodo~\cite{zenodo}.

\bibliography{SWSSB_lib.bib}

\begin{thebibliography}{79}%
\makeatletter
\providecommand \@ifxundefined [1]{%
 \@ifx{#1\undefined}
}%
\providecommand \@ifnum [1]{%
 \ifnum #1\expandafter \@firstoftwo
 \else \expandafter \@secondoftwo
 \fi
}%
\providecommand \@ifx [1]{%
 \ifx #1\expandafter \@firstoftwo
 \else \expandafter \@secondoftwo
 \fi
}%
\providecommand \natexlab [1]{#1}%
\providecommand \enquote  [1]{``#1''}%
\providecommand \bibnamefont  [1]{#1}%
\providecommand \bibfnamefont [1]{#1}%
\providecommand \citenamefont [1]{#1}%
\providecommand \href@noop [0]{\@secondoftwo}%
\providecommand \href [0]{\begingroup \@sanitize@url \@href}%
\providecommand \@href[1]{\@@startlink{#1}\@@href}%
\providecommand \@@href[1]{\endgroup#1\@@endlink}%
\providecommand \@sanitize@url [0]{\catcode `\\12\catcode `\$12\catcode `\&12\catcode `\#12\catcode `\^12\catcode `\_12\catcode `\%12\relax}%
\providecommand \@@startlink[1]{}%
\providecommand \@@endlink[0]{}%
\providecommand \url  [0]{\begingroup\@sanitize@url \@url }%
\providecommand \@url [1]{\endgroup\@href {#1}{\urlprefix }}%
\providecommand \urlprefix  [0]{URL }%
\providecommand \Eprint [0]{\href }%
\providecommand \doibase [0]{https://doi.org/}%
\providecommand \selectlanguage [0]{\@gobble}%
\providecommand \bibinfo  [0]{\@secondoftwo}%
\providecommand \bibfield  [0]{\@secondoftwo}%
\providecommand \translation [1]{[#1]}%
\providecommand \BibitemOpen [0]{}%
\providecommand \bibitemStop [0]{}%
\providecommand \bibitemNoStop [0]{.\EOS\space}%
\providecommand \EOS [0]{\spacefactor3000\relax}%
\providecommand \BibitemShut  [1]{\csname bibitem#1\endcsname}%
\let\auto@bib@innerbib\@empty
%</preamble>
\bibitem [{\citenamefont {Buča}\ and\ \citenamefont {Prosen}(2012)}]{buca_note_2012}%
  \BibitemOpen
  \bibfield  {author} {\bibinfo {author} {\bibfnamefont {B.}~\bibnamefont {Buča}}\ and\ \bibinfo {author} {\bibfnamefont {T.}~\bibnamefont {Prosen}},\ }\bibfield  {title} {\bibinfo {title} {A note on symmetry reductions of the {Lindblad} equation: transport in constrained open spin chains},\ }\href {https://doi.org/10.1088/1367-2630/14/7/073007} {\bibfield  {journal} {\bibinfo  {journal} {New Journal of Physics}\ ,\ \bibinfo {pages} {073007}} (\bibinfo {year} {2012})}\BibitemShut {NoStop}%
\bibitem [{\citenamefont {Albert}\ and\ \citenamefont {Jiang}(2014)}]{albert2014symmetries}%
  \BibitemOpen
  \bibfield  {author} {\bibinfo {author} {\bibfnamefont {V.~V.}\ \bibnamefont {Albert}}\ and\ \bibinfo {author} {\bibfnamefont {L.}~\bibnamefont {Jiang}},\ }\bibfield  {title} {\bibinfo {title} {Symmetries and conserved quantities in lindblad master equations},\ }\href {https://doi.org/10.1103/PhysRevA.89.022118} {\bibfield  {journal} {\bibinfo  {journal} {Phys. Rev. A}\ }\textbf {\bibinfo {volume} {89}},\ \bibinfo {pages} {022118} (\bibinfo {year} {2014})}\BibitemShut {NoStop}%
\bibitem [{\citenamefont {Bao}\ \emph {et~al.}(2021)\citenamefont {Bao}, \citenamefont {Choi},\ and\ \citenamefont {Altman}}]{bao_symmetry_2021}%
  \BibitemOpen
  \bibfield  {author} {\bibinfo {author} {\bibfnamefont {Y.}~\bibnamefont {Bao}}, \bibinfo {author} {\bibfnamefont {S.}~\bibnamefont {Choi}},\ and\ \bibinfo {author} {\bibfnamefont {E.}~\bibnamefont {Altman}},\ }\bibfield  {title} {\bibinfo {title} {Symmetry enriched phases of quantum circuits},\ }\href {https://doi.org/10.1016/j.aop.2021.168618} {\bibfield  {journal} {\bibinfo  {journal} {Annals of Physics}\ }\textbf {\bibinfo {volume} {435}},\ \bibinfo {pages} {168618} (\bibinfo {year} {2021})}\BibitemShut {NoStop}%
\bibitem [{\citenamefont {Ogunnaike}\ \emph {et~al.}(2023)\citenamefont {Ogunnaike}, \citenamefont {Feldmeier},\ and\ \citenamefont {Lee}}]{ogunnaike_unifying_2023}%
  \BibitemOpen
  \bibfield  {author} {\bibinfo {author} {\bibfnamefont {O.}~\bibnamefont {Ogunnaike}}, \bibinfo {author} {\bibfnamefont {J.}~\bibnamefont {Feldmeier}},\ and\ \bibinfo {author} {\bibfnamefont {J.~Y.}\ \bibnamefont {Lee}},\ }\bibfield  {title} {\bibinfo {title} {Unifying {Emergent} {Hydrodynamics} and {Lindbladian} {Low}-{Energy} {Spectra} across {Symmetries}, {Constraints}, and {Long}-{Range} {Interactions}},\ }\href {https://doi.org/10.1103/PhysRevLett.131.220403} {\bibfield  {journal} {\bibinfo  {journal} {Physical Review Letters}\ }\textbf {\bibinfo {volume} {131}},\ \bibinfo {pages} {220403} (\bibinfo {year} {2023})}\BibitemShut {NoStop}%
\bibitem [{\citenamefont {Moudgalya}\ and\ \citenamefont {Motrunich}(2024{\natexlab{a}})}]{moudgalya_symmetries_2023}%
  \BibitemOpen
  \bibfield  {author} {\bibinfo {author} {\bibfnamefont {S.}~\bibnamefont {Moudgalya}}\ and\ \bibinfo {author} {\bibfnamefont {O.~I.}\ \bibnamefont {Motrunich}},\ }\bibfield  {title} {\bibinfo {title} {Symmetries as ground states of local superoperators: Hydrodynamic implications},\ }\href {https://doi.org/10.1103/PRXQuantum.5.040330} {\bibfield  {journal} {\bibinfo  {journal} {PRX Quantum}\ }\textbf {\bibinfo {volume} {5}},\ \bibinfo {pages} {040330} (\bibinfo {year} {2024}{\natexlab{a}})}\BibitemShut {NoStop}%
\bibitem [{\citenamefont {Chen}\ \emph {et~al.}(2025)\citenamefont {Chen}, \citenamefont {Sun},\ and\ \citenamefont {Zhang}}]{chen_strong--weak_2024}%
  \BibitemOpen
  \bibfield  {author} {\bibinfo {author} {\bibfnamefont {L.}~\bibnamefont {Chen}}, \bibinfo {author} {\bibfnamefont {N.}~\bibnamefont {Sun}},\ and\ \bibinfo {author} {\bibfnamefont {P.}~\bibnamefont {Zhang}},\ }\bibfield  {title} {\bibinfo {title} {Strong-to-weak symmetry breaking and entanglement transitions},\ }\href {https://doi.org/10.1103/PhysRevB.111.L060304} {\bibfield  {journal} {\bibinfo  {journal} {Phys. Rev. B}\ }\textbf {\bibinfo {volume} {111}},\ \bibinfo {pages} {L060304} (\bibinfo {year} {2025})}\BibitemShut {NoStop}%
\bibitem [{\citenamefont {Huang}\ \emph {et~al.}(2024)\citenamefont {Huang}, \citenamefont {Qi}, \citenamefont {Zhang},\ and\ \citenamefont {Lucas}}]{huang_hydrodynamics_2024}%
  \BibitemOpen
  \bibfield  {author} {\bibinfo {author} {\bibfnamefont {X.}~\bibnamefont {Huang}}, \bibinfo {author} {\bibfnamefont {M.}~\bibnamefont {Qi}}, \bibinfo {author} {\bibfnamefont {J.-H.}\ \bibnamefont {Zhang}},\ and\ \bibinfo {author} {\bibfnamefont {A.}~\bibnamefont {Lucas}},\ }\href {https://arxiv.org/abs/2407.08760v1} {\bibinfo {title} {Hydrodynamics as the effective field theory of strong-to-weak spontaneous symmetry breaking}} (\bibinfo {year} {2024})\BibitemShut {NoStop}%
\bibitem [{\citenamefont {Kuno}\ \emph {et~al.}(2024)\citenamefont {Kuno}, \citenamefont {Orito},\ and\ \citenamefont {Ichinose}}]{kuno_strong--weak_2024}%
  \BibitemOpen
  \bibfield  {author} {\bibinfo {author} {\bibfnamefont {Y.}~\bibnamefont {Kuno}}, \bibinfo {author} {\bibfnamefont {T.}~\bibnamefont {Orito}},\ and\ \bibinfo {author} {\bibfnamefont {I.}~\bibnamefont {Ichinose}},\ }\href {https://doi.org/10.1103/PhysRevB.110.094106} {\bibinfo {title} {Strong-to-weak symmetry breaking states in stochastic dephasing stabilizer circuits}} (\bibinfo {year} {2024})\BibitemShut {NoStop}%
\bibitem [{\citenamefont {Lee}\ \emph {et~al.}(2023)\citenamefont {Lee}, \citenamefont {Jian},\ and\ \citenamefont {Xu}}]{lee_quantum_2023}%
  \BibitemOpen
  \bibfield  {author} {\bibinfo {author} {\bibfnamefont {J.~Y.}\ \bibnamefont {Lee}}, \bibinfo {author} {\bibfnamefont {C.-M.}\ \bibnamefont {Jian}},\ and\ \bibinfo {author} {\bibfnamefont {C.}~\bibnamefont {Xu}},\ }\bibfield  {title} {\bibinfo {title} {Quantum criticality under decoherence or weak measurement},\ }\href {https://doi.org/10.1103/PRXQuantum.4.030317} {\bibfield  {journal} {\bibinfo  {journal} {PRX Quantum}\ }\textbf {\bibinfo {volume} {4}},\ \bibinfo {pages} {030317} (\bibinfo {year} {2023})}\BibitemShut {NoStop}%
\bibitem [{\citenamefont {Lessa}\ \emph {et~al.}(2025)\citenamefont {Lessa}, \citenamefont {Ma}, \citenamefont {Zhang}, \citenamefont {Bi}, \citenamefont {Cheng},\ and\ \citenamefont {Wang}}]{lessa_strong--weak_2024}%
  \BibitemOpen
  \bibfield  {author} {\bibinfo {author} {\bibfnamefont {L.~A.}\ \bibnamefont {Lessa}}, \bibinfo {author} {\bibfnamefont {R.}~\bibnamefont {Ma}}, \bibinfo {author} {\bibfnamefont {J.-H.}\ \bibnamefont {Zhang}}, \bibinfo {author} {\bibfnamefont {Z.}~\bibnamefont {Bi}}, \bibinfo {author} {\bibfnamefont {M.}~\bibnamefont {Cheng}},\ and\ \bibinfo {author} {\bibfnamefont {C.}~\bibnamefont {Wang}},\ }\bibfield  {title} {\bibinfo {title} {Strong-to-weak spontaneous symmetry breaking in mixed quantum states},\ }\href {https://doi.org/10.1103/PRXQuantum.6.010344} {\bibfield  {journal} {\bibinfo  {journal} {PRX Quantum}\ }\textbf {\bibinfo {volume} {6}},\ \bibinfo {pages} {010344} (\bibinfo {year} {2025})}\BibitemShut {NoStop}%
\bibitem [{\citenamefont {Ma}\ and\ \citenamefont {Wang}(2023)}]{ma_average_2023}%
  \BibitemOpen
  \bibfield  {author} {\bibinfo {author} {\bibfnamefont {R.}~\bibnamefont {Ma}}\ and\ \bibinfo {author} {\bibfnamefont {C.}~\bibnamefont {Wang}},\ }\bibfield  {title} {\bibinfo {title} {Average {Symmetry}-{Protected} {Topological} {Phases}},\ }\href {https://doi.org/10.1103/PhysRevX.13.031016} {\bibfield  {journal} {\bibinfo  {journal} {Physical Review X}\ }\textbf {\bibinfo {volume} {13}},\ \bibinfo {pages} {031016} (\bibinfo {year} {2023})}\BibitemShut {NoStop}%
\bibitem [{\citenamefont {Sala}\ \emph {et~al.}(2024)\citenamefont {Sala}, \citenamefont {Gopalakrishnan}, \citenamefont {Oshikawa},\ and\ \citenamefont {You}}]{sala_spontaneous_2024}%
  \BibitemOpen
  \bibfield  {author} {\bibinfo {author} {\bibfnamefont {P.}~\bibnamefont {Sala}}, \bibinfo {author} {\bibfnamefont {S.}~\bibnamefont {Gopalakrishnan}}, \bibinfo {author} {\bibfnamefont {M.}~\bibnamefont {Oshikawa}},\ and\ \bibinfo {author} {\bibfnamefont {Y.}~\bibnamefont {You}},\ }\bibfield  {title} {\bibinfo {title} {Spontaneous strong symmetry breaking in open systems: Purification perspective},\ }\href {https://doi.org/10.1103/PhysRevB.110.155150} {\bibfield  {journal} {\bibinfo  {journal} {Phys. Rev. B}\ }\textbf {\bibinfo {volume} {110}},\ \bibinfo {pages} {155150} (\bibinfo {year} {2024})}\BibitemShut {NoStop}%
\bibitem [{\citenamefont {Moharramipour}\ \emph {et~al.}(2024)\citenamefont {Moharramipour}, \citenamefont {Lessa}, \citenamefont {Wang}, \citenamefont {Hsieh},\ and\ \citenamefont {Sahu}}]{moharramipour_symmetry_2024}%
  \BibitemOpen
  \bibfield  {author} {\bibinfo {author} {\bibfnamefont {A.}~\bibnamefont {Moharramipour}}, \bibinfo {author} {\bibfnamefont {L.~A.}\ \bibnamefont {Lessa}}, \bibinfo {author} {\bibfnamefont {C.}~\bibnamefont {Wang}}, \bibinfo {author} {\bibfnamefont {T.~H.}\ \bibnamefont {Hsieh}},\ and\ \bibinfo {author} {\bibfnamefont {S.}~\bibnamefont {Sahu}},\ }\bibfield  {title} {\bibinfo {title} {Symmetry enforced entanglement in maximally mixed states},\ }\href {https://doi.org/10.1103/PRXQuantum.5.040336} {\bibfield  {journal} {\bibinfo  {journal} {PRX Quantum}\ }\textbf {\bibinfo {volume} {5}},\ \bibinfo {pages} {040336} (\bibinfo {year} {2024})}\BibitemShut {NoStop}%
\bibitem [{\citenamefont {Li}\ \emph {et~al.}(2024)\citenamefont {Li}, \citenamefont {Pollmann}, \citenamefont {Read},\ and\ \citenamefont {Sala}}]{li_highly-entangled_2024}%
  \BibitemOpen
  \bibfield  {author} {\bibinfo {author} {\bibfnamefont {Y.}~\bibnamefont {Li}}, \bibinfo {author} {\bibfnamefont {F.}~\bibnamefont {Pollmann}}, \bibinfo {author} {\bibfnamefont {N.}~\bibnamefont {Read}},\ and\ \bibinfo {author} {\bibfnamefont {P.}~\bibnamefont {Sala}},\ }\href {https://doi.org/10.48550/arXiv.2406.08567} {\bibinfo {title} {Highly-entangled stationary states from strong symmetries}} (\bibinfo {year} {2024}),\ \bibinfo {note} {arXiv:2406.08567 [quant-ph]}\BibitemShut {NoStop}%
\bibitem [{\citenamefont {Lee}\ \emph {et~al.}(2025)\citenamefont {Lee}, \citenamefont {You},\ and\ \citenamefont {Xu}}]{lee_symmetry_2025}%
  \BibitemOpen
  \bibfield  {author} {\bibinfo {author} {\bibfnamefont {J.~Y.}\ \bibnamefont {Lee}}, \bibinfo {author} {\bibfnamefont {Y.-Z.}\ \bibnamefont {You}},\ and\ \bibinfo {author} {\bibfnamefont {C.}~\bibnamefont {Xu}},\ }\bibfield  {title} {\bibinfo {title} {Symmetry protected topological phases under decoherence},\ }\href {https://doi.org/10.22331/q-2025-01-23-1607} {\bibfield  {journal} {\bibinfo  {journal} {Quantum}\ }\textbf {\bibinfo {volume} {9}},\ \bibinfo {pages} {1607} (\bibinfo {year} {2025})}\BibitemShut {NoStop}%
\bibitem [{\citenamefont {Agrawal}\ \emph {et~al.}(2022)\citenamefont {Agrawal}, \citenamefont {Zabalo}, \citenamefont {Chen}, \citenamefont {Wilson}, \citenamefont {Potter}, \citenamefont {Pixley}, \citenamefont {Gopalakrishnan},\ and\ \citenamefont {Vasseur}}]{agrawal_entanglement_2022}%
  \BibitemOpen
  \bibfield  {author} {\bibinfo {author} {\bibfnamefont {U.}~\bibnamefont {Agrawal}}, \bibinfo {author} {\bibfnamefont {A.}~\bibnamefont {Zabalo}}, \bibinfo {author} {\bibfnamefont {K.}~\bibnamefont {Chen}}, \bibinfo {author} {\bibfnamefont {J.~H.}\ \bibnamefont {Wilson}}, \bibinfo {author} {\bibfnamefont {A.~C.}\ \bibnamefont {Potter}}, \bibinfo {author} {\bibfnamefont {J.~H.}\ \bibnamefont {Pixley}}, \bibinfo {author} {\bibfnamefont {S.}~\bibnamefont {Gopalakrishnan}},\ and\ \bibinfo {author} {\bibfnamefont {R.}~\bibnamefont {Vasseur}},\ }\bibfield  {title} {\bibinfo {title} {Entanglement and charge-sharpening transitions in {U}(1) symmetric monitored quantum circuits},\ }\href {https://doi.org/10.1103/PhysRevX.12.041002} {\bibfield  {journal} {\bibinfo  {journal} {Physical Review X}\ }\textbf {\bibinfo {volume} {12}},\ \bibinfo {pages} {041002} (\bibinfo {year} {2022})}\BibitemShut {NoStop}%
\bibitem [{\citenamefont {Barratt}\ \emph {et~al.}(2022)\citenamefont {Barratt}, \citenamefont {Agrawal}, \citenamefont {Gopalakrishnan}, \citenamefont {Huse}, \citenamefont {Vasseur},\ and\ \citenamefont {Potter}}]{barratt_field_2022}%
  \BibitemOpen
  \bibfield  {author} {\bibinfo {author} {\bibfnamefont {F.}~\bibnamefont {Barratt}}, \bibinfo {author} {\bibfnamefont {U.}~\bibnamefont {Agrawal}}, \bibinfo {author} {\bibfnamefont {S.}~\bibnamefont {Gopalakrishnan}}, \bibinfo {author} {\bibfnamefont {D.~A.}\ \bibnamefont {Huse}}, \bibinfo {author} {\bibfnamefont {R.}~\bibnamefont {Vasseur}},\ and\ \bibinfo {author} {\bibfnamefont {A.~C.}\ \bibnamefont {Potter}},\ }\bibfield  {title} {\bibinfo {title} {Field theory of charge sharpening in symmetric monitored quantum circuits},\ }\href {https://doi.org/10.1103/PhysRevLett.129.120604} {\bibfield  {journal} {\bibinfo  {journal} {Physical Review Letters}\ }\textbf {\bibinfo {volume} {129}},\ \bibinfo {pages} {120604} (\bibinfo {year} {2022})}\BibitemShut {NoStop}%
\bibitem [{\citenamefont {Singh}\ \emph {et~al.}(2025)\citenamefont {Singh}, \citenamefont {Vasseur}, \citenamefont {Potter},\ and\ \citenamefont {Gopalakrishnan}}]{Singh2025}%
  \BibitemOpen
  \bibfield  {author} {\bibinfo {author} {\bibfnamefont {H.}~\bibnamefont {Singh}}, \bibinfo {author} {\bibfnamefont {R.}~\bibnamefont {Vasseur}}, \bibinfo {author} {\bibfnamefont {A.~C.}\ \bibnamefont {Potter}},\ and\ \bibinfo {author} {\bibfnamefont {S.}~\bibnamefont {Gopalakrishnan}},\ }\href {https://arxiv.org/abs/2503.10308} {\bibinfo {title} {Mixed-state learnability transitions in monitored noisy quantum dynamics}} (\bibinfo {year} {2025}),\ \Eprint {https://arxiv.org/abs/2503.10308} {arXiv:2503.10308 [quant-ph]} \BibitemShut {NoStop}%
\bibitem [{\citenamefont {Zerba}\ \emph {et~al.}(2025)\citenamefont {Zerba} \emph {et~al.}}]{Zerba2025}%
  \BibitemOpen
  \bibfield  {author} {\bibinfo {author} {\bibfnamefont {C.}~\bibnamefont {Zerba}} \emph {et~al.},\ }\href@noop {} {} (\bibinfo {year} {2025}),\ \bibinfo {note} {in preparation}\BibitemShut {NoStop}%
\bibitem [{\citenamefont {Potter}\ and\ \citenamefont {Vasseur}(2022)}]{potter2022entanglement}%
  \BibitemOpen
  \bibfield  {author} {\bibinfo {author} {\bibfnamefont {A.~C.}\ \bibnamefont {Potter}}\ and\ \bibinfo {author} {\bibfnamefont {R.}~\bibnamefont {Vasseur}},\ }\bibinfo {title} {Entanglement dynamics in hybrid quantum circuits},\ in\ \href {https://doi.org/10.1007/978-3-031-03998-0_9} {\emph {\bibinfo {booktitle} {Entanglement in Spin Chains: From Theory to Quantum Technology Applications}}},\ \bibinfo {editor} {edited by\ \bibinfo {editor} {\bibfnamefont {A.}~\bibnamefont {Bayat}}, \bibinfo {editor} {\bibfnamefont {S.}~\bibnamefont {Bose}},\ and\ \bibinfo {editor} {\bibfnamefont {H.}~\bibnamefont {Johannesson}}}\ (\bibinfo  {publisher} {Springer International Publishing},\ \bibinfo {address} {Cham},\ \bibinfo {year} {2022})\ pp.\ \bibinfo {pages} {211--249}\BibitemShut {NoStop}%
\bibitem [{\citenamefont {Fisher}\ \emph {et~al.}(2023)\citenamefont {Fisher}, \citenamefont {Khemani}, \citenamefont {Nahum},\ and\ \citenamefont {Vijay}}]{fisher_random_2023}%
  \BibitemOpen
  \bibfield  {author} {\bibinfo {author} {\bibfnamefont {M.~P.~A.}\ \bibnamefont {Fisher}}, \bibinfo {author} {\bibfnamefont {V.}~\bibnamefont {Khemani}}, \bibinfo {author} {\bibfnamefont {A.}~\bibnamefont {Nahum}},\ and\ \bibinfo {author} {\bibfnamefont {S.}~\bibnamefont {Vijay}},\ }\bibfield  {title} {\bibinfo {title} {Random {Quantum} {Circuits}},\ }\href {https://doi.org/10.1146/annurev-conmatphys-031720-030658} {\bibfield  {journal} {\bibinfo  {journal} {Annual Review of Condensed Matter Physics}\ }\textbf {\bibinfo {volume} {14}},\ \bibinfo {pages} {335} (\bibinfo {year} {2023})}\BibitemShut {NoStop}%
\bibitem [{\citenamefont {Gu}\ \emph {et~al.}(2024)\citenamefont {Gu}, \citenamefont {Wang},\ and\ \citenamefont {Wang}}]{gu_spontaneous_2024}%
  \BibitemOpen
  \bibfield  {author} {\bibinfo {author} {\bibfnamefont {D.}~\bibnamefont {Gu}}, \bibinfo {author} {\bibfnamefont {Z.}~\bibnamefont {Wang}},\ and\ \bibinfo {author} {\bibfnamefont {Z.}~\bibnamefont {Wang}},\ }\href {https://arxiv.org/abs/2406.19381} {\bibinfo {title} {Spontaneous symmetry breaking in open quantum systems: strong, weak, and strong-to-weak}} (\bibinfo {year} {2024}),\ \Eprint {https://arxiv.org/abs/2406.19381} {arXiv:2406.19381 [quant-ph]} \BibitemShut {NoStop}%
\bibitem [{\citenamefont {Nielsen}\ and\ \citenamefont {Chuang}(2010)}]{nielsen_quantum_2010}%
  \BibitemOpen
  \bibfield  {author} {\bibinfo {author} {\bibfnamefont {M.~A.}\ \bibnamefont {Nielsen}}\ and\ \bibinfo {author} {\bibfnamefont {I.~L.}\ \bibnamefont {Chuang}},\ }\href {https://doi.org/10.1017/CBO9780511976667} {\bibinfo {title} {Quantum {Computation} and {Quantum} {Information}: 10th {Anniversary} {Edition}}} (\bibinfo {year} {2010})\BibitemShut {NoStop}%
\bibitem [{\citenamefont {Mitsuhashi}\ \emph {et~al.}(2025)\citenamefont {Mitsuhashi}, \citenamefont {Suzuki}, \citenamefont {Soejima},\ and\ \citenamefont {Yoshioka}}]{Mitsuhashi_2025}%
  \BibitemOpen
  \bibfield  {author} {\bibinfo {author} {\bibfnamefont {Y.}~\bibnamefont {Mitsuhashi}}, \bibinfo {author} {\bibfnamefont {R.}~\bibnamefont {Suzuki}}, \bibinfo {author} {\bibfnamefont {T.}~\bibnamefont {Soejima}},\ and\ \bibinfo {author} {\bibfnamefont {N.}~\bibnamefont {Yoshioka}},\ }\bibfield  {title} {\bibinfo {title} {Unitary designs of symmetric local random circuits},\ }\href {https://doi.org/10.1103/PhysRevLett.134.180404} {\bibfield  {journal} {\bibinfo  {journal} {Phys. Rev. Lett.}\ }\textbf {\bibinfo {volume} {134}},\ \bibinfo {pages} {180404} (\bibinfo {year} {2025})}\BibitemShut {NoStop}%
\bibitem [{\citenamefont {Sun}\ \emph {et~al.}(2025)\citenamefont {Sun}, \citenamefont {Zhang},\ and\ \citenamefont {Feng}}]{sun_scheme_2025}%
  \BibitemOpen
  \bibfield  {author} {\bibinfo {author} {\bibfnamefont {N.}~\bibnamefont {Sun}}, \bibinfo {author} {\bibfnamefont {P.}~\bibnamefont {Zhang}},\ and\ \bibinfo {author} {\bibfnamefont {L.}~\bibnamefont {Feng}},\ }\bibfield  {title} {\bibinfo {title} {Scheme to detect the strong-to-weak symmetry breaking via randomized measurements},\ }\href {https://doi.org/10.1103/7p5x-7yqb} {\bibfield  {journal} {\bibinfo  {journal} {Phys. Rev. Lett.}\ }\textbf {\bibinfo {volume} {135}},\ \bibinfo {pages} {090403} (\bibinfo {year} {2025})}\BibitemShut {NoStop}%
\bibitem [{\citenamefont {Biane}(1993)}]{Biane_1993}%
  \BibitemOpen
  \bibfield  {author} {\bibinfo {author} {\bibfnamefont {P.}~\bibnamefont {Biane}},\ }\bibfield  {title} {\bibinfo {title} {Asymptotic estimates for the multiplicities in tensor powers of a g-module},\ }\href@noop {} {\bibfield  {journal} {\bibinfo  {journal} {C. R. Acad. Sci., Paris, S{\'e}r. I}\ }\textbf {\bibinfo {volume} {316}},\ \bibinfo {pages} {849} (\bibinfo {year} {1993})}\BibitemShut {NoStop}%
\bibitem [{\citenamefont {Tate}\ and\ \citenamefont {Zelditch}(2004)}]{Tate_2004}%
  \BibitemOpen
  \bibfield  {author} {\bibinfo {author} {\bibfnamefont {T.}~\bibnamefont {Tate}}\ and\ \bibinfo {author} {\bibfnamefont {S.}~\bibnamefont {Zelditch}},\ }\bibfield  {title} {\bibinfo {title} {Lattice path combinatorics and asymptotics of multiplicities of weights in tensor powers},\ }\href {https://doi.org/10.1016/j.jfa.2004.01.004} {\bibfield  {journal} {\bibinfo  {journal} {J. Funct. Anal.}\ }\textbf {\bibinfo {volume} {217}},\ \bibinfo {pages} {402} (\bibinfo {year} {2004})}\BibitemShut {NoStop}%
\bibitem [{\citenamefont {Hauser}\ \emph {et~al.}(2024)\citenamefont {Hauser}, \citenamefont {Li}, \citenamefont {Vijay},\ and\ \citenamefont {Fisher}}]{hauser_continuous_2023}%
  \BibitemOpen
  \bibfield  {author} {\bibinfo {author} {\bibfnamefont {J.}~\bibnamefont {Hauser}}, \bibinfo {author} {\bibfnamefont {Y.}~\bibnamefont {Li}}, \bibinfo {author} {\bibfnamefont {S.}~\bibnamefont {Vijay}},\ and\ \bibinfo {author} {\bibfnamefont {M.~P.~A.}\ \bibnamefont {Fisher}},\ }\bibfield  {title} {\bibinfo {title} {Continuous symmetry breaking in adaptive quantum dynamics},\ }\href {https://doi.org/10.1103/PhysRevB.109.214305} {\bibfield  {journal} {\bibinfo  {journal} {Phys. Rev. B}\ }\textbf {\bibinfo {volume} {109}},\ \bibinfo {pages} {214305} (\bibinfo {year} {2024})}\BibitemShut {NoStop}%
\bibitem [{\citenamefont {Breuer}\ \emph {et~al.}(2007)\citenamefont {Breuer}, \citenamefont {Petruccione}, \citenamefont {Breuer},\ and\ \citenamefont {Petruccione}}]{breuer_theory_2007}%
  \BibitemOpen
  \bibfield  {author} {\bibinfo {author} {\bibfnamefont {H.-P.}\ \bibnamefont {Breuer}}, \bibinfo {author} {\bibfnamefont {F.}~\bibnamefont {Petruccione}}, \bibinfo {author} {\bibfnamefont {H.-P.}\ \bibnamefont {Breuer}},\ and\ \bibinfo {author} {\bibfnamefont {F.}~\bibnamefont {Petruccione}},\ }\href@noop {} {\emph {\bibinfo {title} {The {Theory} of {Open} {Quantum} {Systems}}}}\ (\bibinfo  {publisher} {Oxford University Press},\ \bibinfo {address} {Oxford, New York},\ \bibinfo {year} {2007})\BibitemShut {NoStop}%
\bibitem [{\citenamefont {Watrous}(2018)}]{watrous_theory_2018}%
  \BibitemOpen
  \bibfield  {author} {\bibinfo {author} {\bibfnamefont {J.}~\bibnamefont {Watrous}},\ }\href {https://doi.org/10.1017/9781316848142} {\emph {\bibinfo {title} {The {Theory} of {Quantum} {Information}}}}\ (\bibinfo  {publisher} {Cambridge University Press},\ \bibinfo {address} {Cambridge},\ \bibinfo {year} {2018})\BibitemShut {NoStop}%
\bibitem [{\citenamefont {Lindblad}(1976)}]{lindblad_generators_1976}%
  \BibitemOpen
  \bibfield  {author} {\bibinfo {author} {\bibfnamefont {G.}~\bibnamefont {Lindblad}},\ }\bibfield  {title} {\bibinfo {title} {On the generators of quantum dynamical semigroups},\ }\href {https://doi.org/10.1007/BF01608499} {\bibfield  {journal} {\bibinfo  {journal} {Communications in Mathematical Physics}\ }\textbf {\bibinfo {volume} {48}},\ \bibinfo {pages} {119} (\bibinfo {year} {1976})}\BibitemShut {NoStop}%
\bibitem [{\citenamefont {Lashkari}\ \emph {et~al.}(2013)\citenamefont {Lashkari}, \citenamefont {Stanford}, \citenamefont {Hastings}, \citenamefont {Osborne},\ and\ \citenamefont {Hayden}}]{lashkari_towards_2013}%
  \BibitemOpen
  \bibfield  {author} {\bibinfo {author} {\bibfnamefont {N.}~\bibnamefont {Lashkari}}, \bibinfo {author} {\bibfnamefont {D.}~\bibnamefont {Stanford}}, \bibinfo {author} {\bibfnamefont {M.}~\bibnamefont {Hastings}}, \bibinfo {author} {\bibfnamefont {T.}~\bibnamefont {Osborne}},\ and\ \bibinfo {author} {\bibfnamefont {P.}~\bibnamefont {Hayden}},\ }\bibfield  {title} {\bibinfo {title} {Towards the fast scrambling conjecture},\ }\href {https://doi.org/10.1007/JHEP04(2013)022} {\bibfield  {journal} {\bibinfo  {journal} {Journal of High Energy Physics}\ }\textbf {\bibinfo {volume} {2013}},\ \bibinfo {pages} {22} (\bibinfo {year} {2013})}\BibitemShut {NoStop}%
\bibitem [{\citenamefont {Bauer}\ \emph {et~al.}(2017)\citenamefont {Bauer}, \citenamefont {Bernard},\ and\ \citenamefont {Jin}}]{bauer_stochastic_2017}%
  \BibitemOpen
  \bibfield  {author} {\bibinfo {author} {\bibfnamefont {M.}~\bibnamefont {Bauer}}, \bibinfo {author} {\bibfnamefont {D.}~\bibnamefont {Bernard}},\ and\ \bibinfo {author} {\bibfnamefont {T.}~\bibnamefont {Jin}},\ }\bibfield  {title} {\bibinfo {title} {Stochastic dissipative quantum spin chains ({I}) : {Quantum} fluctuating discrete hydrodynamics},\ }\href {https://doi.org/10.21468/SciPostPhys.3.5.033} {\bibfield  {journal} {\bibinfo  {journal} {SciPost Physics}\ }\textbf {\bibinfo {volume} {3}},\ \bibinfo {pages} {033} (\bibinfo {year} {2017})}\BibitemShut {NoStop}%
\bibitem [{\citenamefont {Jian}\ \emph {et~al.}(2021)\citenamefont {Jian}, \citenamefont {Liu}, \citenamefont {Chen}, \citenamefont {Swingle},\ and\ \citenamefont {Zhang}}]{jian_syk_2021}%
  \BibitemOpen
  \bibfield  {author} {\bibinfo {author} {\bibfnamefont {S.-K.}\ \bibnamefont {Jian}}, \bibinfo {author} {\bibfnamefont {C.}~\bibnamefont {Liu}}, \bibinfo {author} {\bibfnamefont {X.}~\bibnamefont {Chen}}, \bibinfo {author} {\bibfnamefont {B.}~\bibnamefont {Swingle}},\ and\ \bibinfo {author} {\bibfnamefont {P.}~\bibnamefont {Zhang}},\ }\bibfield  {title} {\bibinfo {title} {{SYK} meets non-{Hermiticity} {II}: measurement-induced phase transition},\ }\href {https://doi.org/10.1103/PhysRevLett.127.140601} {\bibfield  {journal} {\bibinfo  {journal} {Physical Review Letters}\ }\textbf {\bibinfo {volume} {127}},\ \bibinfo {pages} {140601} (\bibinfo {year} {2021})}\BibitemShut {NoStop}%
\bibitem [{\citenamefont {Sünderhauf}\ \emph {et~al.}(2019)\citenamefont {Sünderhauf}, \citenamefont {Piroli}, \citenamefont {Qi}, \citenamefont {Schuch},\ and\ \citenamefont {Cirac}}]{sunderhauf_quantum_2019}%
  \BibitemOpen
  \bibfield  {author} {\bibinfo {author} {\bibfnamefont {C.}~\bibnamefont {Sünderhauf}}, \bibinfo {author} {\bibfnamefont {L.}~\bibnamefont {Piroli}}, \bibinfo {author} {\bibfnamefont {X.-L.}\ \bibnamefont {Qi}}, \bibinfo {author} {\bibfnamefont {N.}~\bibnamefont {Schuch}},\ and\ \bibinfo {author} {\bibfnamefont {J.~I.}\ \bibnamefont {Cirac}},\ }\bibfield  {title} {\bibinfo {title} {Quantum chaos in the {Brownian} {SYK} model with large finite n: {OTOCs} and tripartite information},\ }\href {https://doi.org/10.1007/JHEP11(2019)038} {\bibfield  {journal} {\bibinfo  {journal} {Journal of High Energy Physics}\ }\textbf {\bibinfo {volume} {2019}},\ \bibinfo {pages} {38} (\bibinfo {year} {2019})}\BibitemShut {NoStop}%
\bibitem [{\citenamefont {Agarwal}\ and\ \citenamefont {Xu}(2022)}]{agarwal_emergent_2022}%
  \BibitemOpen
  \bibfield  {author} {\bibinfo {author} {\bibfnamefont {L.}~\bibnamefont {Agarwal}}\ and\ \bibinfo {author} {\bibfnamefont {S.}~\bibnamefont {Xu}},\ }\bibfield  {title} {\bibinfo {title} {Emergent symmetry in {Brownian} {SYK} models and charge dependent scrambling},\ }\href {https://doi.org/10.1007/JHEP02(2022)045} {\bibfield  {journal} {\bibinfo  {journal} {Journal of High Energy Physics}\ }\textbf {\bibinfo {volume} {2022}},\ \bibinfo {pages} {45} (\bibinfo {year} {2022})}\BibitemShut {NoStop}%
\bibitem [{\citenamefont {Zhou}\ and\ \citenamefont {Chen}(2019)}]{zhou_operator_2019}%
  \BibitemOpen
  \bibfield  {author} {\bibinfo {author} {\bibfnamefont {T.}~\bibnamefont {Zhou}}\ and\ \bibinfo {author} {\bibfnamefont {X.}~\bibnamefont {Chen}},\ }\bibfield  {title} {\bibinfo {title} {Operator dynamics in a {Brownian} quantum circuit},\ }\href {https://doi.org/10.1103/PhysRevE.99.052212} {\bibfield  {journal} {\bibinfo  {journal} {Physical Review E}\ }\textbf {\bibinfo {volume} {99}},\ \bibinfo {pages} {052212} (\bibinfo {year} {2019})}\BibitemShut {NoStop}%
\bibitem [{\citenamefont {Agarwal}\ \emph {et~al.}(2023)\citenamefont {Agarwal}, \citenamefont {Sahu},\ and\ \citenamefont {Xu}}]{agarwal_charge_2023}%
  \BibitemOpen
  \bibfield  {author} {\bibinfo {author} {\bibfnamefont {L.}~\bibnamefont {Agarwal}}, \bibinfo {author} {\bibfnamefont {S.}~\bibnamefont {Sahu}},\ and\ \bibinfo {author} {\bibfnamefont {S.}~\bibnamefont {Xu}},\ }\bibfield  {title} {\bibinfo {title} {Charge transport, information scrambling and quantum operator-coherence in a many-body system with {U}(1) symmetry},\ }\href {https://doi.org/10.1007/JHEP05(2023)037} {\bibfield  {journal} {\bibinfo  {journal} {Journal of High Energy Physics}\ }\textbf {\bibinfo {volume} {2023}},\ \bibinfo {pages} {37} (\bibinfo {year} {2023})}\BibitemShut {NoStop}%
\bibitem [{\citenamefont {Knap}(2018)}]{Knap_2018}%
  \BibitemOpen
  \bibfield  {author} {\bibinfo {author} {\bibfnamefont {M.}~\bibnamefont {Knap}},\ }\bibfield  {title} {\bibinfo {title} {Entanglement production and information scrambling in a noisy spin system},\ }\bibfield  {journal} {\bibinfo  {journal} {Physical Review B}\ }\textbf {\bibinfo {volume} {98}},\ \href {https://doi.org/10.1103/physrevb.98.184416} {10.1103/physrevb.98.184416} (\bibinfo {year} {2018})\BibitemShut {NoStop}%
\bibitem [{\citenamefont {Bernard}\ and\ \citenamefont {Piroli}(2021)}]{bernard_entanglement_2021}%
  \BibitemOpen
  \bibfield  {author} {\bibinfo {author} {\bibfnamefont {D.}~\bibnamefont {Bernard}}\ and\ \bibinfo {author} {\bibfnamefont {L.}~\bibnamefont {Piroli}},\ }\bibfield  {title} {\bibinfo {title} {Entanglement distribution in the quantum symmetric simple exclusion process},\ }\href {https://doi.org/10.1103/PhysRevE.104.014146} {\bibfield  {journal} {\bibinfo  {journal} {Physical Review E}\ }\textbf {\bibinfo {volume} {104}},\ \bibinfo {pages} {014146} (\bibinfo {year} {2021})}\BibitemShut {NoStop}%
\bibitem [{\citenamefont {Swann}\ \emph {et~al.}(2025)\citenamefont {Swann}, \citenamefont {Bernard},\ and\ \citenamefont {Nahum}}]{swann_spacetime_2023}%
  \BibitemOpen
  \bibfield  {author} {\bibinfo {author} {\bibfnamefont {T.}~\bibnamefont {Swann}}, \bibinfo {author} {\bibfnamefont {D.}~\bibnamefont {Bernard}},\ and\ \bibinfo {author} {\bibfnamefont {A.}~\bibnamefont {Nahum}},\ }\bibfield  {title} {\bibinfo {title} {Spacetime picture for entanglement generation in noisy fermion chains},\ }\href {https://doi.org/10.1103/PhysRevB.112.064301} {\bibfield  {journal} {\bibinfo  {journal} {Phys. Rev. B}\ }\textbf {\bibinfo {volume} {112}},\ \bibinfo {pages} {064301} (\bibinfo {year} {2025})}\BibitemShut {NoStop}%
\bibitem [{\citenamefont {Vardhan}\ and\ \citenamefont {Moudgalya}(2024)}]{vardhan_entanglement_2024}%
  \BibitemOpen
  \bibfield  {author} {\bibinfo {author} {\bibfnamefont {S.}~\bibnamefont {Vardhan}}\ and\ \bibinfo {author} {\bibfnamefont {S.}~\bibnamefont {Moudgalya}},\ }\href {https://doi.org/10.48550/arXiv.2407.16763} {\bibinfo {title} {Entanglement dynamics from universal low-lying modes}} (\bibinfo {year} {2024}),\ \bibinfo {note} {arXiv:2407.16763}\BibitemShut {NoStop}%
\bibitem [{\citenamefont {Sahu}\ \emph {et~al.}(2022)\citenamefont {Sahu}, \citenamefont {Jian}, \citenamefont {Bentsen},\ and\ \citenamefont {Swingle}}]{sahu_entanglement_2022}%
  \BibitemOpen
  \bibfield  {author} {\bibinfo {author} {\bibfnamefont {S.}~\bibnamefont {Sahu}}, \bibinfo {author} {\bibfnamefont {S.-K.}\ \bibnamefont {Jian}}, \bibinfo {author} {\bibfnamefont {G.}~\bibnamefont {Bentsen}},\ and\ \bibinfo {author} {\bibfnamefont {B.}~\bibnamefont {Swingle}},\ }\bibfield  {title} {\bibinfo {title} {Entanglement phases in large- {N} hybrid {Brownian} circuits with long-range couplings},\ }\href {https://doi.org/10.1103/PhysRevB.106.224305} {\bibfield  {journal} {\bibinfo  {journal} {Physical Review B}\ }\textbf {\bibinfo {volume} {106}},\ \bibinfo {pages} {224305} (\bibinfo {year} {2022})}\BibitemShut {NoStop}%
\bibitem [{\citenamefont {Moudgalya}\ and\ \citenamefont {Motrunich}(2022)}]{moudgalya_hilbert_2022}%
  \BibitemOpen
  \bibfield  {author} {\bibinfo {author} {\bibfnamefont {S.}~\bibnamefont {Moudgalya}}\ and\ \bibinfo {author} {\bibfnamefont {O.~I.}\ \bibnamefont {Motrunich}},\ }\bibfield  {title} {\bibinfo {title} {Hilbert {Space} {Fragmentation} and {Commutant} {Algebras}},\ }\href {https://doi.org/10.1103/PhysRevX.12.011050} {\bibfield  {journal} {\bibinfo  {journal} {Physical Review X}\ }\textbf {\bibinfo {volume} {12}},\ \bibinfo {pages} {011050} (\bibinfo {year} {2022})}\BibitemShut {NoStop}%
\bibitem [{\citenamefont {Moudgalya}\ and\ \citenamefont {Motrunich}(2023)}]{moudgalya_symmetries_2023-1}%
  \BibitemOpen
  \bibfield  {author} {\bibinfo {author} {\bibfnamefont {S.}~\bibnamefont {Moudgalya}}\ and\ \bibinfo {author} {\bibfnamefont {O.~I.}\ \bibnamefont {Motrunich}},\ }\bibfield  {title} {\bibinfo {title} {From symmetries to commutant algebras in standard {Hamiltonians}},\ }\href {https://doi.org/10.1016/j.aop.2023.169384} {\bibfield  {journal} {\bibinfo  {journal} {Annals of Physics}\ }\textbf {\bibinfo {volume} {455}},\ \bibinfo {pages} {169384} (\bibinfo {year} {2023})}\BibitemShut {NoStop}%
\bibitem [{\citenamefont {Pfeuty}(1970)}]{pfeuty_one-dimensional_1970}%
  \BibitemOpen
  \bibfield  {author} {\bibinfo {author} {\bibfnamefont {P.}~\bibnamefont {Pfeuty}},\ }\bibfield  {title} {\bibinfo {title} {The one-dimensional {Ising} model with a transverse field},\ }\href {https://doi.org/10.1016/0003-4916(70)90270-8} {\bibfield  {journal} {\bibinfo  {journal} {Annals of Physics}\ }\textbf {\bibinfo {volume} {57}},\ \bibinfo {pages} {79} (\bibinfo {year} {1970})}\BibitemShut {NoStop}%
\bibitem [{\citenamefont {Zhang}\ \emph {et~al.}(2025)\citenamefont {Zhang}, \citenamefont {Hsieh}, \citenamefont {Kim},\ and\ \citenamefont {Zou}}]{zhang_probing_2025}%
  \BibitemOpen
  \bibfield  {author} {\bibinfo {author} {\bibfnamefont {Y.}~\bibnamefont {Zhang}}, \bibinfo {author} {\bibfnamefont {T.~H.}\ \bibnamefont {Hsieh}}, \bibinfo {author} {\bibfnamefont {Y.~B.}\ \bibnamefont {Kim}},\ and\ \bibinfo {author} {\bibfnamefont {Y.}~\bibnamefont {Zou}},\ }\href {https://doi.org/10.48550/arXiv.2505.02900} {\bibinfo {title} {Probing mixed-state phases on a quantum computer via {Renyi} correlators and variational decoding}} (\bibinfo {year} {2025}),\ \bibinfo {note} {arXiv:2505.02900 [quant-ph]}\BibitemShut {NoStop}%
\bibitem [{\citenamefont {Müller}(2023)}]{müller_2023_simplified}%
  \BibitemOpen
  \bibfield  {author} {\bibinfo {author} {\bibfnamefont {A.}~\bibnamefont {Müller}},\ }\href {https://arxiv.org/abs/2309.10565} {\bibinfo {title} {A simplified expression for quantum fidelity}} (\bibinfo {year} {2023}),\ \Eprint {https://arxiv.org/abs/2309.10565} {arXiv:2309.10565 [quant-ph]} \BibitemShut {NoStop}%
\bibitem [{\citenamefont {Hauschild}\ and\ \citenamefont {Pollmann}(2018)}]{hauschild_efficient_2018}%
  \BibitemOpen
  \bibfield  {author} {\bibinfo {author} {\bibfnamefont {J.}~\bibnamefont {Hauschild}}\ and\ \bibinfo {author} {\bibfnamefont {F.}~\bibnamefont {Pollmann}},\ }\bibfield  {title} {\bibinfo {title} {Efficient numerical simulations with {Tensor} {Networks}: {Tensor} {Network} {Python} ({TeNPy})},\ }\href {https://doi.org/10.21468/SciPostPhysLectNotes.5} {\bibfield  {journal} {\bibinfo  {journal} {SciPost Physics Lecture Notes}\ ,\ \bibinfo {pages} {5}} (\bibinfo {year} {2018})}\BibitemShut {NoStop}%
\bibitem [{\citenamefont {Baxter}(1982)}]{Baxter_exactly}%
  \BibitemOpen
  \bibfield  {author} {\bibinfo {author} {\bibfnamefont {R.~J.}\ \bibnamefont {Baxter}},\ }\bibfield  {title} {\bibinfo {title} {Exactly solved models in statistical mechanics}\ }(\bibinfo {year} {1982})\BibitemShut {NoStop}%
\bibitem [{\citenamefont {Wu}(1982)}]{RevModPhys.54.235}%
  \BibitemOpen
  \bibfield  {author} {\bibinfo {author} {\bibfnamefont {F.~Y.}\ \bibnamefont {Wu}},\ }\bibfield  {title} {\bibinfo {title} {The potts model},\ }\href {https://doi.org/10.1103/RevModPhys.54.235} {\bibfield  {journal} {\bibinfo  {journal} {Rev. Mod. Phys.}\ }\textbf {\bibinfo {volume} {54}},\ \bibinfo {pages} {235} (\bibinfo {year} {1982})}\BibitemShut {NoStop}%
\bibitem [{\citenamefont {Friedman}\ \emph {et~al.}(2018)\citenamefont {Friedman}, \citenamefont {Vasseur}, \citenamefont {Potter},\ and\ \citenamefont {Parameswaran}}]{friedman_localization-protected_2018}%
  \BibitemOpen
  \bibfield  {author} {\bibinfo {author} {\bibfnamefont {A.~J.}\ \bibnamefont {Friedman}}, \bibinfo {author} {\bibfnamefont {R.}~\bibnamefont {Vasseur}}, \bibinfo {author} {\bibfnamefont {A.~C.}\ \bibnamefont {Potter}},\ and\ \bibinfo {author} {\bibfnamefont {S.~A.}\ \bibnamefont {Parameswaran}},\ }\bibfield  {title} {\bibinfo {title} {Localization-protected order in spin chains with non-{Abelian} discrete symmetries},\ }\href {https://doi.org/10.1103/PhysRevB.98.064203} {\bibfield  {journal} {\bibinfo  {journal} {Physical Review B}\ }\textbf {\bibinfo {volume} {98}},\ \bibinfo {pages} {064203} (\bibinfo {year} {2018})}\BibitemShut {NoStop}%
\bibitem [{\citenamefont {Zhang}\ \emph {et~al.}(2020)\citenamefont {Zhang}, \citenamefont {Tindall}, \citenamefont {Mur-Petit}, \citenamefont {Jaksch},\ and\ \citenamefont {Buča}}]{zhang_stationary_2020}%
  \BibitemOpen
  \bibfield  {author} {\bibinfo {author} {\bibfnamefont {Z.}~\bibnamefont {Zhang}}, \bibinfo {author} {\bibfnamefont {J.}~\bibnamefont {Tindall}}, \bibinfo {author} {\bibfnamefont {J.}~\bibnamefont {Mur-Petit}}, \bibinfo {author} {\bibfnamefont {D.}~\bibnamefont {Jaksch}},\ and\ \bibinfo {author} {\bibfnamefont {B.}~\bibnamefont {Buča}},\ }\bibfield  {title} {\bibinfo {title} {Stationary state degeneracy of open quantum systems with non-abelian symmetries},\ }\href {https://doi.org/10.1088/1751-8121/ab88e3} {\bibfield  {journal} {\bibinfo  {journal} {Journal of Physics A: Mathematical and Theoretical}\ }\textbf {\bibinfo {volume} {53}},\ \bibinfo {pages} {215304} (\bibinfo {year} {2020})}\BibitemShut {NoStop}%
\bibitem [{\citenamefont {{Shao}}(2023)}]{shao2023noninvertible}%
  \BibitemOpen
  \bibfield  {author} {\bibinfo {author} {\bibfnamefont {S.-H.}\ \bibnamefont {{Shao}}},\ }\bibfield  {title} {\bibinfo {title} {{What's Done Cannot Be Undone: TASI Lectures on Non-Invertible Symmetries}},\ }\bibfield  {journal} {\bibinfo  {journal} {arXiv e-prints}\ }\href {https://doi.org/10.48550/arXiv.2308.00747} {10.48550/arXiv.2308.00747} (\bibinfo {year} {2023}),\ \Eprint {https://arxiv.org/abs/2308.00747} {arXiv:2308.00747 [hep-th]} \BibitemShut {NoStop}%
\bibitem [{\citenamefont {Sala}\ \emph {et~al.}(2020)\citenamefont {Sala}, \citenamefont {Rakovszky}, \citenamefont {Verresen}, \citenamefont {Knap},\ and\ \citenamefont {Pollmann}}]{Sala2020}%
  \BibitemOpen
  \bibfield  {author} {\bibinfo {author} {\bibfnamefont {P.}~\bibnamefont {Sala}}, \bibinfo {author} {\bibfnamefont {T.}~\bibnamefont {Rakovszky}}, \bibinfo {author} {\bibfnamefont {R.}~\bibnamefont {Verresen}}, \bibinfo {author} {\bibfnamefont {M.}~\bibnamefont {Knap}},\ and\ \bibinfo {author} {\bibfnamefont {F.}~\bibnamefont {Pollmann}},\ }\bibfield  {title} {\bibinfo {title} {Ergodicity breaking arising from hilbert space fragmentation in dipole-conserving hamiltonians},\ }\href {https://doi.org/10.1103/PhysRevX.10.011047} {\bibfield  {journal} {\bibinfo  {journal} {Phys. Rev. X}\ }\textbf {\bibinfo {volume} {10}},\ \bibinfo {pages} {011047} (\bibinfo {year} {2020})}\BibitemShut {NoStop}%
\bibitem [{\citenamefont {Khemani}\ \emph {et~al.}(2020)\citenamefont {Khemani}, \citenamefont {Hermele},\ and\ \citenamefont {Nandkishore}}]{khemani2020hsf}%
  \BibitemOpen
  \bibfield  {author} {\bibinfo {author} {\bibfnamefont {V.}~\bibnamefont {Khemani}}, \bibinfo {author} {\bibfnamefont {M.}~\bibnamefont {Hermele}},\ and\ \bibinfo {author} {\bibfnamefont {R.}~\bibnamefont {Nandkishore}},\ }\bibfield  {title} {\bibinfo {title} {Localization from hilbert space shattering: From theory to physical realizations},\ }\href {https://doi.org/10.1103/PhysRevB.101.174204} {\bibfield  {journal} {\bibinfo  {journal} {Phys. Rev. B}\ }\textbf {\bibinfo {volume} {101}},\ \bibinfo {pages} {174204} (\bibinfo {year} {2020})}\BibitemShut {NoStop}%
\bibitem [{\citenamefont {Rakovszky}\ \emph {et~al.}(2020)\citenamefont {Rakovszky}, \citenamefont {Sala}, \citenamefont {Verresen}, \citenamefont {Knap},\ and\ \citenamefont {Pollmann}}]{Rakovszky20}%
  \BibitemOpen
  \bibfield  {author} {\bibinfo {author} {\bibfnamefont {T.}~\bibnamefont {Rakovszky}}, \bibinfo {author} {\bibfnamefont {P.}~\bibnamefont {Sala}}, \bibinfo {author} {\bibfnamefont {R.}~\bibnamefont {Verresen}}, \bibinfo {author} {\bibfnamefont {M.}~\bibnamefont {Knap}},\ and\ \bibinfo {author} {\bibfnamefont {F.}~\bibnamefont {Pollmann}},\ }\bibfield  {title} {\bibinfo {title} {{Statistical localization: From strong fragmentation to strong edge modes}},\ }\href {https://doi.org/10.1103/PhysRevB.101.125126} {\bibfield  {journal} {\bibinfo  {journal} {Phys. Rev. B}\ }\textbf {\bibinfo {volume} {101}},\ \bibinfo {pages} {125126} (\bibinfo {year} {2020})}\BibitemShut {NoStop}%
\bibitem [{\citenamefont {Li}\ \emph {et~al.}(2023)\citenamefont {Li}, \citenamefont {Sala},\ and\ \citenamefont {Pollmann}}]{li_hsf_2023}%
  \BibitemOpen
  \bibfield  {author} {\bibinfo {author} {\bibfnamefont {Y.}~\bibnamefont {Li}}, \bibinfo {author} {\bibfnamefont {P.}~\bibnamefont {Sala}},\ and\ \bibinfo {author} {\bibfnamefont {F.}~\bibnamefont {Pollmann}},\ }\bibfield  {title} {\bibinfo {title} {Hilbert space fragmentation in open quantum systems},\ }\href {https://doi.org/10.1103/PhysRevResearch.5.043239} {\bibfield  {journal} {\bibinfo  {journal} {Phys. Rev. Res.}\ }\textbf {\bibinfo {volume} {5}},\ \bibinfo {pages} {043239} (\bibinfo {year} {2023})}\BibitemShut {NoStop}%
\bibitem [{\citenamefont {Moudgalya}\ \emph {et~al.}(2022)\citenamefont {Moudgalya}, \citenamefont {Bernevig},\ and\ \citenamefont {Regnault}}]{moudgalya2021review}%
  \BibitemOpen
  \bibfield  {author} {\bibinfo {author} {\bibfnamefont {S.}~\bibnamefont {Moudgalya}}, \bibinfo {author} {\bibfnamefont {B.~A.}\ \bibnamefont {Bernevig}},\ and\ \bibinfo {author} {\bibfnamefont {N.}~\bibnamefont {Regnault}},\ }\bibfield  {title} {\bibinfo {title} {{Quantum many-body scars and Hilbert space fragmentation: a review of exact results}},\ }\href {https://doi.org/10.1088/1361-6633/ac73a0} {\bibfield  {journal} {\bibinfo  {journal} {Reports on Progress in Physics}\ }\textbf {\bibinfo {volume} {85}},\ \bibinfo {pages} {086501} (\bibinfo {year} {2022})}\BibitemShut {NoStop}%
\bibitem [{\citenamefont {Moudgalya}\ and\ \citenamefont {Motrunich}(2024{\natexlab{b}})}]{moudgalya_scars_2024}%
  \BibitemOpen
  \bibfield  {author} {\bibinfo {author} {\bibfnamefont {S.}~\bibnamefont {Moudgalya}}\ and\ \bibinfo {author} {\bibfnamefont {O.~I.}\ \bibnamefont {Motrunich}},\ }\bibfield  {title} {\bibinfo {title} {Exhaustive characterization of quantum many-body scars using commutant algebras},\ }\href {https://doi.org/10.1103/PhysRevX.14.041069} {\bibfield  {journal} {\bibinfo  {journal} {Phys. Rev. X}\ }\textbf {\bibinfo {volume} {14}},\ \bibinfo {pages} {041069} (\bibinfo {year} {2024}{\natexlab{b}})}\BibitemShut {NoStop}%
\bibitem [{\citenamefont {Serbyn}\ \emph {et~al.}(2021)\citenamefont {Serbyn}, \citenamefont {Abanin},\ and\ \citenamefont {Papi{\'{c}}}}]{serbyn2020review}%
  \BibitemOpen
  \bibfield  {author} {\bibinfo {author} {\bibfnamefont {M.}~\bibnamefont {Serbyn}}, \bibinfo {author} {\bibfnamefont {D.~A.}\ \bibnamefont {Abanin}},\ and\ \bibinfo {author} {\bibfnamefont {Z.}~\bibnamefont {Papi{\'{c}}}},\ }\bibfield  {title} {\bibinfo {title} {{Quantum many-body scars and weak breaking of ergodicity}},\ }\href {https://doi.org/10.1038/s41567-021-01230-2} {\bibfield  {journal} {\bibinfo  {journal} {Nature Physics}\ }\textbf {\bibinfo {volume} {17}},\ \bibinfo {pages} {675} (\bibinfo {year} {2021})}\BibitemShut {NoStop}%
\bibitem [{\citenamefont {Papi{\'{c}}}(2022)}]{papic2021review}%
  \BibitemOpen
  \bibfield  {author} {\bibinfo {author} {\bibfnamefont {Z.}~\bibnamefont {Papi{\'{c}}}},\ }\bibinfo {title} {{Weak Ergodicity Breaking Through the Lens of Quantum Entanglement}},\ in\ \href {https://doi.org/10.1007/978-3-031-03998-0_13} {\emph {\bibinfo {booktitle} {{Entanglement in Spin Chains: From Theory to Quantum Technology Applications}}}},\ \bibinfo {editor} {edited by\ \bibinfo {editor} {\bibfnamefont {A.}~\bibnamefont {Bayat}}, \bibinfo {editor} {\bibfnamefont {S.}~\bibnamefont {Bose}},\ and\ \bibinfo {editor} {\bibfnamefont {H.}~\bibnamefont {Johannesson}}}\ (\bibinfo  {publisher} {Springer International Publishing},\ \bibinfo {address} {Cham},\ \bibinfo {year} {2022})\ pp.\ \bibinfo {pages} {341--395}\BibitemShut {NoStop}%
\bibitem [{\citenamefont {{Sahu}}\ \emph {et~al.}(2025)\citenamefont {{Sahu}}, \citenamefont {{Li}},\ and\ \citenamefont {{Sala}}}]{sahu2025entanglement}%
  \BibitemOpen
  \bibfield  {author} {\bibinfo {author} {\bibfnamefont {S.}~\bibnamefont {{Sahu}}}, \bibinfo {author} {\bibfnamefont {Y.}~\bibnamefont {{Li}}},\ and\ \bibinfo {author} {\bibfnamefont {P.}~\bibnamefont {{Sala}}},\ }\bibfield  {title} {\bibinfo {title} {{Entanglement cost hierarchies in quantum fragmented mixed states}},\ }\href {https://doi.org/10.48550/arXiv.2506.04637} {\bibfield  {journal} {\bibinfo  {journal} {arXiv e-prints}\ ,\ \bibinfo {eid} {arXiv:2506.04637}} (\bibinfo {year} {2025})},\ \Eprint {https://arxiv.org/abs/2506.04637} {arXiv:2506.04637 [quant-ph]} \BibitemShut {NoStop}%
\bibitem [{\citenamefont {Moudgalya}\ \emph {et~al.}()\citenamefont {Moudgalya}, \citenamefont {Prem}, \citenamefont {Nandkishore}, \citenamefont {Regnault},\ and\ \citenamefont {Bernevig}}]{moudgalya2019krylov}%
  \BibitemOpen
  \bibfield  {author} {\bibinfo {author} {\bibfnamefont {S.}~\bibnamefont {Moudgalya}}, \bibinfo {author} {\bibfnamefont {A.}~\bibnamefont {Prem}}, \bibinfo {author} {\bibfnamefont {R.}~\bibnamefont {Nandkishore}}, \bibinfo {author} {\bibfnamefont {N.}~\bibnamefont {Regnault}},\ and\ \bibinfo {author} {\bibfnamefont {B.~A.}\ \bibnamefont {Bernevig}},\ }\bibinfo {title} {{Thermalization and Its Absence within Krylov Subspaces of a Constrained Hamiltonian}},\ in\ \href {https://doi.org/10.1142/9789811231711_0009} {\emph {\bibinfo {booktitle} {Memorial Volume for Shoucheng Zhang}}},\ Chap.~\bibinfo {chapter} {7}, pp.\ \bibinfo {pages} {147--209}\BibitemShut {NoStop}%
\bibitem [{\citenamefont {Yang}\ \emph {et~al.}(2020)\citenamefont {Yang}, \citenamefont {Liu}, \citenamefont {Gorshkov},\ and\ \citenamefont {Iadecola}}]{yang2020hsf}%
  \BibitemOpen
  \bibfield  {author} {\bibinfo {author} {\bibfnamefont {Z.-C.}\ \bibnamefont {Yang}}, \bibinfo {author} {\bibfnamefont {F.}~\bibnamefont {Liu}}, \bibinfo {author} {\bibfnamefont {A.~V.}\ \bibnamefont {Gorshkov}},\ and\ \bibinfo {author} {\bibfnamefont {T.}~\bibnamefont {Iadecola}},\ }\bibfield  {title} {\bibinfo {title} {Hilbert-space fragmentation from strict confinement},\ }\href {https://doi.org/10.1103/PhysRevLett.124.207602} {\bibfield  {journal} {\bibinfo  {journal} {Phys. Rev. Lett.}\ }\textbf {\bibinfo {volume} {124}},\ \bibinfo {pages} {207602} (\bibinfo {year} {2020})}\BibitemShut {NoStop}%
\bibitem [{\citenamefont {Morningstar}\ \emph {et~al.}(2020)\citenamefont {Morningstar}, \citenamefont {Khemani},\ and\ \citenamefont {Huse}}]{Morningstar2020}%
  \BibitemOpen
  \bibfield  {author} {\bibinfo {author} {\bibfnamefont {A.}~\bibnamefont {Morningstar}}, \bibinfo {author} {\bibfnamefont {V.}~\bibnamefont {Khemani}},\ and\ \bibinfo {author} {\bibfnamefont {D.~A.}\ \bibnamefont {Huse}},\ }\bibfield  {title} {\bibinfo {title} {Kinetically constrained freezing transition in a dipole-conserving system},\ }\href {https://doi.org/10.1103/PhysRevB.101.214205} {\bibfield  {journal} {\bibinfo  {journal} {Phys. Rev. B}\ }\textbf {\bibinfo {volume} {101}},\ \bibinfo {pages} {214205} (\bibinfo {year} {2020})}\BibitemShut {NoStop}%
\bibitem [{\citenamefont {Feldmeier}\ and\ \citenamefont {Knap}(2021)}]{Feldmeier2021}%
  \BibitemOpen
  \bibfield  {author} {\bibinfo {author} {\bibfnamefont {J.}~\bibnamefont {Feldmeier}}\ and\ \bibinfo {author} {\bibfnamefont {M.}~\bibnamefont {Knap}},\ }\bibfield  {title} {\bibinfo {title} {Critically slow operator dynamics in constrained many-body systems},\ }\href {https://doi.org/10.1103/physrevlett.127.235301} {\bibfield  {journal} {\bibinfo  {journal} {Phys. Rev. Lett.}\ }\textbf {\bibinfo {volume} {127}},\ \bibinfo {pages} {235301} (\bibinfo {year} {2021})}\BibitemShut {NoStop}%
\bibitem [{\citenamefont {Elben}\ \emph {et~al.}(2020)\citenamefont {Elben}, \citenamefont {Kueng}, \citenamefont {Huang}, \citenamefont {van Bijnen}, \citenamefont {Kokail}, \citenamefont {Dalmonte}, \citenamefont {Calabrese}, \citenamefont {Kraus}, \citenamefont {Preskill}, \citenamefont {Zoller},\ and\ \citenamefont {Vermersch}}]{elben_mixed-state_2020}%
  \BibitemOpen
  \bibfield  {author} {\bibinfo {author} {\bibfnamefont {A.}~\bibnamefont {Elben}}, \bibinfo {author} {\bibfnamefont {R.}~\bibnamefont {Kueng}}, \bibinfo {author} {\bibfnamefont {H.-Y.~R.}\ \bibnamefont {Huang}}, \bibinfo {author} {\bibfnamefont {R.}~\bibnamefont {van Bijnen}}, \bibinfo {author} {\bibfnamefont {C.}~\bibnamefont {Kokail}}, \bibinfo {author} {\bibfnamefont {M.}~\bibnamefont {Dalmonte}}, \bibinfo {author} {\bibfnamefont {P.}~\bibnamefont {Calabrese}}, \bibinfo {author} {\bibfnamefont {B.}~\bibnamefont {Kraus}}, \bibinfo {author} {\bibfnamefont {J.}~\bibnamefont {Preskill}}, \bibinfo {author} {\bibfnamefont {P.}~\bibnamefont {Zoller}},\ and\ \bibinfo {author} {\bibfnamefont {B.}~\bibnamefont {Vermersch}},\ }\bibfield  {title} {\bibinfo {title} {Mixed-{State} {Entanglement} from {Local} {Randomized} {Measurements}},\ }\href {https://doi.org/10.1103/PhysRevLett.125.200501} {\bibfield  {journal} {\bibinfo  {journal} {Physical Review Letters}\ }\textbf {\bibinfo {volume} {125}},\ \bibinfo {pages}
  {200501} (\bibinfo {year} {2020})},\ \bibinfo {note} {publisher: American Physical Society}\BibitemShut {NoStop}%
\bibitem [{\citenamefont {Brydges}\ \emph {et~al.}(2019)\citenamefont {Brydges}, \citenamefont {Elben}, \citenamefont {Jurcevic}, \citenamefont {Vermersch}, \citenamefont {Maier}, \citenamefont {Lanyon}, \citenamefont {Zoller}, \citenamefont {Blatt},\ and\ \citenamefont {Roos}}]{brydges_probing_2019}%
  \BibitemOpen
  \bibfield  {author} {\bibinfo {author} {\bibfnamefont {T.}~\bibnamefont {Brydges}}, \bibinfo {author} {\bibfnamefont {A.}~\bibnamefont {Elben}}, \bibinfo {author} {\bibfnamefont {P.}~\bibnamefont {Jurcevic}}, \bibinfo {author} {\bibfnamefont {B.}~\bibnamefont {Vermersch}}, \bibinfo {author} {\bibfnamefont {C.}~\bibnamefont {Maier}}, \bibinfo {author} {\bibfnamefont {B.~P.}\ \bibnamefont {Lanyon}}, \bibinfo {author} {\bibfnamefont {P.}~\bibnamefont {Zoller}}, \bibinfo {author} {\bibfnamefont {R.}~\bibnamefont {Blatt}},\ and\ \bibinfo {author} {\bibfnamefont {C.~F.}\ \bibnamefont {Roos}},\ }\bibfield  {title} {\bibinfo {title} {Probing {Rényi} entanglement entropy via randomized measurements},\ }\href {https://www.science.org/doi/10.1126/science.aau4963} {\bibfield  {journal} {\bibinfo  {journal} {Science}\ }\textbf {\bibinfo {volume} {364}},\ \bibinfo {pages} {260} (\bibinfo {year} {2019})}\BibitemShut {NoStop}%
\bibitem [{\citenamefont {Satzinger}\ \emph {et~al.}(2021)\citenamefont {Satzinger}, \citenamefont {Liu}, \citenamefont {Smith}, \citenamefont {Knapp}, \citenamefont {Newman}, \citenamefont {Jones}, \citenamefont {Chen}, \citenamefont {Quintana}, \citenamefont {Mi}, \citenamefont {Dunsworth} \emph {et~al.}}]{satzinger_2021}%
  \BibitemOpen
  \bibfield  {author} {\bibinfo {author} {\bibfnamefont {K.~J.}\ \bibnamefont {Satzinger}}, \bibinfo {author} {\bibfnamefont {Y.-J.}\ \bibnamefont {Liu}}, \bibinfo {author} {\bibfnamefont {A.}~\bibnamefont {Smith}}, \bibinfo {author} {\bibfnamefont {C.}~\bibnamefont {Knapp}}, \bibinfo {author} {\bibfnamefont {M.}~\bibnamefont {Newman}}, \bibinfo {author} {\bibfnamefont {C.}~\bibnamefont {Jones}}, \bibinfo {author} {\bibfnamefont {Z.}~\bibnamefont {Chen}}, \bibinfo {author} {\bibfnamefont {C.}~\bibnamefont {Quintana}}, \bibinfo {author} {\bibfnamefont {X.}~\bibnamefont {Mi}}, \bibinfo {author} {\bibfnamefont {A.}~\bibnamefont {Dunsworth}}, \emph {et~al.},\ }\bibfield  {title} {\bibinfo {title} {Realizing topologically ordered states on a quantum processor},\ }\href {https://doi.org/10.1126/science.abi8378} {\bibfield  {journal} {\bibinfo  {journal} {Science}\ }\textbf {\bibinfo {volume} {374}},\ \bibinfo {pages} {1237} (\bibinfo {year} {2021})}\BibitemShut {NoStop}%
\bibitem [{\citenamefont {Ziereis}\ \emph {et~al.}(2025)\citenamefont {Ziereis}, \citenamefont {Moudgalya},\ and\ \citenamefont {Knap}}]{zenodo}%
  \BibitemOpen
  \bibfield  {author} {\bibinfo {author} {\bibfnamefont {N.}~\bibnamefont {Ziereis}}, \bibinfo {author} {\bibfnamefont {S.}~\bibnamefont {Moudgalya}},\ and\ \bibinfo {author} {\bibfnamefont {M.}~\bibnamefont {Knap}},\ }\href {https://doi.org/10.5281/zenodo.17100499} {\bibinfo {title} {Zenodo entry for: {S}trong-to-{W}eak {S}ymmetry {B}reaking {P}hases in {S}teady {S}tates of {Q}uantum {O}perations}} (\bibinfo {year} {2025})\BibitemShut {NoStop}%
\bibitem [{\citenamefont {Appel}\ and\ \citenamefont {Kowalski}(2007)}]{appel_mathematics_2007}%
  \BibitemOpen
  \bibfield  {author} {\bibinfo {author} {\bibfnamefont {W.}~\bibnamefont {Appel}}\ and\ \bibinfo {author} {\bibfnamefont {E.}~\bibnamefont {Kowalski}},\ }\href {https://books.google.de/books?id=NaTvAAAAMAAJ} {\emph {\bibinfo {title} {Mathematics for {Physics} and {Physicists}}}},\ Mathematical notes\ (\bibinfo  {publisher} {Princeton University Press},\ \bibinfo {year} {2007})\BibitemShut {NoStop}%
\bibitem [{\citenamefont {Loomis}(1953)}]{loomis_introduction_1953}%
  \BibitemOpen
  \bibfield  {author} {\bibinfo {author} {\bibfnamefont {L.~H.}\ \bibnamefont {Loomis}},\ }\href {http://archive.org/details/introductiontoab031610mbp} {\emph {\bibinfo {title} {An {Introduction} {To} {Abstract} {Harmonic} {Analysis}}}}\ (\bibinfo  {publisher} {D.Van Nostrand Company Inc.},\ \bibinfo {year} {1953})\BibitemShut {NoStop}%
\bibitem [{\citenamefont {Bröcker}\ and\ \citenamefont {Tom~Dieck}(1985)}]{brocker_representations_1985}%
  \BibitemOpen
  \bibfield  {author} {\bibinfo {author} {\bibfnamefont {T.}~\bibnamefont {Bröcker}}\ and\ \bibinfo {author} {\bibfnamefont {T.}~\bibnamefont {Tom~Dieck}},\ }\href {https://doi.org/10.1007/978-3-662-12918-0} {\emph {\bibinfo {title} {Representations of {Compact} {Lie} {Groups}}}},\ \bibinfo {series} {Graduate {Texts} in {Mathematics}}, Vol.~\bibinfo {volume} {98}\ (\bibinfo  {publisher} {Springer},\ \bibinfo {address} {Berlin, Heidelberg},\ \bibinfo {year} {1985})\BibitemShut {NoStop}%
\bibitem [{\citenamefont {Fulton}\ and\ \citenamefont {Harris}(2004)}]{fulton_representation_2004}%
  \BibitemOpen
  \bibfield  {author} {\bibinfo {author} {\bibfnamefont {W.}~\bibnamefont {Fulton}}\ and\ \bibinfo {author} {\bibfnamefont {J.}~\bibnamefont {Harris}},\ }\href {https://doi.org/10.1007/978-1-4612-0979-9} {\emph {\bibinfo {title} {Representation {Theory}}}},\ \bibinfo {series} {Graduate {Texts} in {Mathematics}}, Vol.\ \bibinfo {volume} {129}\ (\bibinfo  {publisher} {Springer},\ \bibinfo {address} {New York, NY},\ \bibinfo {year} {2004})\BibitemShut {NoStop}%
\bibitem [{\citenamefont {Sepanski}(2007)}]{sepanski_compact_2007}%
  \BibitemOpen
  \bibinfo {editor} {\bibfnamefont {M.~R.}\ \bibnamefont {Sepanski}},\ ed.,\ \href {https://doi.org/10.1007/978-0-387-49158-5} {\emph {\bibinfo {title} {Compact {Lie} {Groups}}}},\ \bibinfo {series} {Graduate {Texts} in {Mathematics}}, Vol.\ \bibinfo {volume} {235}\ (\bibinfo  {publisher} {Springer},\ \bibinfo {address} {New York, NY},\ \bibinfo {year} {2007})\BibitemShut {NoStop}%
\bibitem [{\citenamefont {Frank}\ and\ \citenamefont {Lieb}(2013)}]{frank_monotonicity2013}%
  \BibitemOpen
  \bibfield  {author} {\bibinfo {author} {\bibfnamefont {R.~L.}\ \bibnamefont {Frank}}\ and\ \bibinfo {author} {\bibfnamefont {E.~H.}\ \bibnamefont {Lieb}},\ }\bibfield  {title} {\bibinfo {title} {Monotonicity of a relative rényi entropy},\ }\href {https://doi.org/10.1063/1.4838835} {\bibfield  {journal} {\bibinfo  {journal} {Journal of Mathematical Physics}\ }\textbf {\bibinfo {volume} {54}},\ \bibinfo {pages} {122201} (\bibinfo {year} {2013})}\BibitemShut {NoStop}%
\bibitem [{\citenamefont {Müller-Lennert}\ \emph {et~al.}(2013)\citenamefont {Müller-Lennert}, \citenamefont {Dupuis}, \citenamefont {Szehr}, \citenamefont {Fehr},\ and\ \citenamefont {Tomamichel}}]{muller-lennert_quantum_2014}%
  \BibitemOpen
  \bibfield  {author} {\bibinfo {author} {\bibfnamefont {M.}~\bibnamefont {Müller-Lennert}}, \bibinfo {author} {\bibfnamefont {F.}~\bibnamefont {Dupuis}}, \bibinfo {author} {\bibfnamefont {O.}~\bibnamefont {Szehr}}, \bibinfo {author} {\bibfnamefont {S.}~\bibnamefont {Fehr}},\ and\ \bibinfo {author} {\bibfnamefont {M.}~\bibnamefont {Tomamichel}},\ }\bibfield  {title} {\bibinfo {title} {On quantum rényi entropies: A new generalization and some properties},\ }\href {https://doi.org/10.1063/1.4838856} {\bibfield  {journal} {\bibinfo  {journal} {Journal of Mathematical Physics}\ }\textbf {\bibinfo {volume} {54}},\ \bibinfo {pages} {122203} (\bibinfo {year} {2013})}\BibitemShut {NoStop}%
\bibitem [{\citenamefont {Iten}\ \emph {et~al.}(2016)\citenamefont {Iten}, \citenamefont {Renes},\ and\ \citenamefont {Sutter}}]{Iten2016PrettyGM}%
  \BibitemOpen
  \bibfield  {author} {\bibinfo {author} {\bibfnamefont {R.}~\bibnamefont {Iten}}, \bibinfo {author} {\bibfnamefont {J.~M.}\ \bibnamefont {Renes}},\ and\ \bibinfo {author} {\bibfnamefont {D.}~\bibnamefont {Sutter}},\ }\bibfield  {title} {\bibinfo {title} {Pretty good measures in quantum information theory},\ }\href {https://api.semanticscholar.org/CorpusID:6209207} {\bibfield  {journal} {\bibinfo  {journal} {2017 IEEE International Symposium on Information Theory (ISIT)}\ ,\ \bibinfo {pages} {3195}} (\bibinfo {year} {2016})}\BibitemShut {NoStop}%
\end{thebibliography}%

\appendix\label{appendix}

\onecolumngrid
\section{Proof of SW-SSB of MMIS}
\label{appendix: SW-SSB maximally mixed}
\subsection{Main Theorem}
Before we proceed to the main definition, we need to address a small technical point about taking the $N \rightarrow \infty$ limit correctly.
We first need to know for which $N$ the MMIS $\rho^\infty_N(\theta)$ is well defined (i.e., the scalar symmetry sector $V_\theta(N)$ is not empty).
For example for qubit systems with a $U(1)$ symmetry generated by $\sum_iZ_i$, states with $\sum_iZ_i=0$ only exist for even $N$.
More generally, one needs to know how $\mathcal{H}_N$ decomposes into irreps for every $N$, and we will assume this is known for the groups we consider.
It is convenient to define the following notion.
\begin{definition}[n-regular scalar symmetry sector]
    A scalar symmetry sector $V_\theta$ is called $n$-regular, if there are $N_0,n \in \mathbb{N}$ such that $\mathrm{dim}V_\theta(N_0+L\cdot n)>0$ for all $N \in \mathbb{N}$.
\end{definition}

An $n$-regular scalar symmetry sector is thus one that has a well-defined thermodynamic limit by increasing the system size in steps of $n$.
The positive and negative parity sectors of $G = \mathbb{Z}_2$ for qubits are $1$-regular (and therefore also $2$-regular) and the sectors of a $G = U(1)$ generated by $\sum_i Z_i$ for qubits which have charge $Q$ are $2$-regular with $N_0=2Q$.
W.l.o.g. we can assume $n$ to be even since any $n$-regular scalar symmetry sector is also $2n$-regular.
Although we will in the end consider arbitrary sequences $\{\Lambda\}_N$ of system sizes where $\dim V_\theta(N_L)>0$, we will use sequences of the form $N_L=N_0+L\cdot n$ as a type of ``anchor" that help us in proving Thm. \ref{theorem: SWSSB of MMIS}.
Further, note that in the main text we only considered order parameters that acted non-trivially only on a single lattice site, however we can directly extend this notion and consider a set of order parameters $\{O^\alpha_\Omega\}_{\alpha \in \mathcal{I_O}}$ where every $O^\alpha_\Omega$ is supported on a fixed region $\Omega$ of the lattice.
We also assume that the order parameters form an irreducible representation of dimension $|\mathcal{I}_O|$ and are orthogonal w.r.t. the Hilbert-Schmidt inner product.
Taking the limit $|i-j|\to \infty$ then means that we take sequences $\Omega_{i_{N_L}},\Omega_{j_{N_L}}$ of translations of $\Omega$ parametrized by the system size $N$ with $\mathrm{dist}(\Omega_{i_{N_L}},\Omega_{j_{N_L}})\rightarrow \infty$.
We can now state the final result.\\
\begin{theorem}[SW-SSB of MMIS]\label{theorem: SWSSB of MMIS}
    Let $V_\theta$ be an $n$-regular scalar symmetry sector of an on-site representation of a compact Lie group $G$ (including finite group $G$) and $\{N_L\}_N$ be a sequence of increasing system sizes with $\dim V_\theta(N_L)>0$ for all $L$.
    Let $\{O^\alpha_\Omega\}_{\alpha \in \mathcal{I}_O}$ be a set of local order parameters that transform under an irreducible representation of dimension $|\mathcal{I}_O|$ and $P_\Omega$ any operator supported on $\Omega$.
    If we assume that $\dim V_\theta(N_L)\sim \frac{1}{\mathrm{poly}(N_L)}\dim (\mathcal{H}_{\mathrm{loc}})^{N_L}$ as $L\to \infty$ then:
    \begin{align}
        \lim_{\mathrm{dist}(\Omega_i,\Omega_j)\to\infty}\frac{\Tr\left[O_{\Omega_i}^{\alpha}\rho^\infty_{N_L}P_{\Omega_i}^{\dagger}P_{\Omega_j}\rho^\infty_{N_L}O_{\Omega_j}^{\alpha\dagger}\right]}{\Tr\left[\big(\rho^\infty_{N_L}\big)^2\right]}&=\frac{1}{|\mathcal{I}_O|(\dim \mathcal{H}_{\mathrm{loc}})^{2|\Omega|}}\sum_{\beta \in \mathcal{I}_O}\left|\Tr[P^\dagger O^\beta ]\right|^2 \label{eq: generalized Rényi-2}
        \\
        \lim_{\mathrm{dist}(\Omega_i,\Omega_j)\to\infty}C^{sw}(\Omega_{i}, \Omega_{j})[\rho^\infty_{N_L}]&=\frac{\norm{O^\alpha}_2^4}{|\mathcal{I}_O|(\dim \mathcal{H}_{\mathrm{loc}})^{2|\Omega|}} \label{eq:Rényi2expression},
\end{align}
\end{theorem}

Before moving on to the proof of Theorem \ref{theorem: SWSSB of MMIS} in App. \ref{sec: proof of Thm 1}, a few comments are in order:
\begin{itemize}
{
\item The sum in Eq.~\eqref{eq: generalized Rényi-2} is over all operators $O^\beta$ that span the irreducible representation that $O^\alpha$ transforms under.
\item 
Eq.~\eqref{eq:Rényi2expression} follows directly from Eq. \eqref{eq: generalized Rényi-2} by taking $P=O^\alpha$ and taking the order parameters to be orthogonal w.r.t. the Hilbert-Schmidt inner product.
We include the more general expression in Eq.~\eqref{eq: generalized Rényi-2} in order to calculate more general correlation functions that are sometimes encountered in the literature~\cite{hauser_continuous_2023, lessa_strong--weak_2024} which we will demonstrate in App.~\ref{sec: U1 calculations} for the $U(1)$ symmetry generated by $\sum_iZ_i$.
\item In order to diagnose SW-SSB we also need to show that the ordinary correlation function $C^{w\emptyset}$ vanishes under the assumptions of Thm.~\ref{theorem: SWSSB of MMIS}. This is done in Lem.~\ref{lemma: linear correlations}.
Intuitively, this should be expected as the MMIS are states at infinite temperature where ordinary correlations should vanish.
\item It was shown in \cite{Biane_1993, Tate_2004} that the assumption $\dim V_\theta(N_L)\sim \frac{1}{\mathrm{poly}(N_L)}\dim(\mathcal{H}_{\mathrm{loc}})^{N_L}$ is fulfilled at least for all finite and semisimple compact Lie groups as well as certain representations of connected compact Lie groups.
}
\end{itemize}

\subsection{Preliminary Lemmas}
Before we can prove Thm.~\ref{theorem: SWSSB of MMIS}, we need to collect a few tools that will aid in evaluating the expressions for $C^{sw}$ and $C^{w\emptyset}$ (see Eq.~(\ref{eq: R2 LRO}) and (\ref{eq: no R1 LRO}) for their definitions) in the limit $N_L\to \infty$.
Crucially, we will be using the fact that the maximally mixed invariant state is expressed in terms of its integral representation Eq.~\eqref{eq: definition of maximally mixed}.
Denote by $C(G)$ the set of continuous complex valued functions on $G$. The first Lemma gives sufficient conditions for a sequence of functions $f_L$ on $G$ to act as a Dirac delta function centered at the identity $e$ on the group manifold as $L\to \infty$.
\begin{lemma}[Delta distributions on compact Lie groups\label{lemma: delta}]
   Let G be a compact Lie group and $(f_L)_{L\in \mathbb{N}}\in C(G)^\mathbb{N}$ such that $\lim_{L \to \infty} \frac{|f_L(g)|}{\int_Gf_L(g)dg}=0$ for all $g\neq e$. Then for all $\phi \in C(G)$
   \begin{equation}
       \lim_{L\to \infty}\int_G\phi(g)\frac{f_L(g)}{\int_Gf_L(g)dg}=\phi(e)
   \end{equation}
\end{lemma}
\begin{proof}
    Our goal is to show that for every $\varepsilon>0$ we have that
    \begin{align}
        &\left|\int_G\phi(g)\frac{f_L(g)}{\int_Gf_L(g)dg} -\phi(e)\right|=\left|\int_G(\phi(g)-\phi(e))\frac{f_L(g)}{\int_Gf_L(g)dg}\right| \leq \varepsilon \label{eq: delta proof objective}
    \end{align}
    for $L$ large enough. To do so, we follow a common strategy (see e.g. \cite{appel_mathematics_2007} p. 601) and split the integral into a contribution close to $e$ and one that is bounded away from $e$. Let $B_{\frac{\varepsilon}{2}}(\phi(e))$ be the $\frac{\varepsilon}{2}$ ball around $\phi(e)$. Then $U_{\frac{\varepsilon}{2}}=\phi^{-1}(B_{\frac{\varepsilon}{2}}(\phi(e)))$ is open by continuity. We have that:
\begin{align*}
    &\left|\int_G(\phi(g)-\phi(e))\frac{f_L(g)}{\int_Gf_L(g)dg}\right| \leq
    \left|\int_{U_{\frac{\varepsilon}{2}}}\left(\phi(g)-\phi(e)\right)\frac{f_L(g)}{\int_Gf_L(g)dg}\right|+\int_{G\setminus U_{\varepsilon}}\left|\phi(g)-\phi(e)\right|\frac{|f_L(g)|}{\int_Gf_L(g)dg} \label{eq: delta proof 1}
\end{align*}
The first contribution can be bounded by:
\begin{align}
   &\left|\int_{U_{\frac{\varepsilon}{2}}}\left(\phi(g)-\phi(e)\right)\frac{f_L(g)}{\int_Gf_L(g)dg}\; dg\right| \leq \sup_{g \in U_{\frac{\varepsilon}{2}}}|\phi(g)-\phi(e)|\frac{\left|\int_{U_\frac{\varepsilon}{2}}f_L(g)\; dg\right|}{\int_Gf_L(g)dg}\\
   &=\sup_{g \in U_{\frac{\varepsilon}{2}}}|\phi(g)-\phi(e)|\frac{\left| \int_Gf_L(g)dg -\int_{G\setminus U_\frac{\varepsilon}{2}}f_L(g)\; dg\right|}{\int_Gf_L(g)dg}\leq \frac{\varepsilon}{2}\left(1+\frac{\max_{g\in G\setminus U_{\frac{\varepsilon}{2}}}|f_L(g)|}{\int_G f_L(g)\; dg}\right),
\end{align}
where in the last step we used the triangle inequality, the fact that $\int_G dg=1$ and that $|\phi(g)-\phi(e)|\leq \frac{\varepsilon}{2}$ for all $g \in U_{\frac{\varepsilon}{2}}$ by definition of $U_{\frac{\varepsilon}{2}}$. Note that $\max_{g\in G\setminus U_{\frac{\varepsilon}{2}}}|f_L(g)|$ is attained in $G\setminus U_{\frac{\varepsilon}{2}}$ by continuity of $f_L(g)$ since $G\setminus U_{\frac{\varepsilon}{2}}$ is closed and thus compact since $G$ is compact.
The second term in Eq.~\eqref{eq: delta proof 1} is bounded as follows:
\begin{align}
    \int_{G\setminus U_{\varepsilon}}\left|\phi(g)-\phi(e)\right|\frac{|f_L(g)|}{\int_Gf_L(g)dg} \leq \frac{\max_{g\in G\setminus U_{\frac{\varepsilon}{2}}}|f_L(g)|}{\int_G f_L(g)\; dg}\int_{G\setminus U_{\varepsilon}}\left|\phi(g)-\phi(e)\right|\;dg\leq \frac{\max_{g\in G\setminus U_{\frac{\varepsilon}{2}}}|f_L(g)|}{\int_G f_L(g)\; dg}\left(\norm{\phi}_{L^1}+\phi(e)\right)\label{eq: delta proof 2}
\end{align}
In the last step we again used the triangle inequality and the definition of the $L^1$ norm $\norm{\phi}_{L^1}=\int_G|\phi(g)|dg$ which finite since $\phi$ is continuous and $G$ is compact. Combining Eqs.~\eqref{eq: delta proof 1} and \eqref{eq: delta proof 2} we arrive at
\begin{align}
    \left|\int_G(\phi(g)-\phi(e))\frac{f_L(g)}{\int_Gf_L(g)dg}\right| \leq \frac{\varepsilon}{2}+\frac{\max_{g\in G\setminus U_{\frac{\varepsilon}{2}}}|f_L(g)|}{\int_G f_L(g)\; dg}\left(\norm{\phi}_{L^1}+\phi(e) +\frac{\epsilon}{2}\right) \label{eq: delta proof 3}
\end{align}
Since by assumption $\lim_{L \to \infty} \frac{|f_L(g)|}{\int_Gf_L(g)dg}=0$ for all $g \neq e$, we can pick $L$ large enough such that the second contribution in Eq.~\eqref{eq: delta proof 3} becomes smaller than $\frac{\varepsilon}{2}$, which shows Eq.~\eqref{eq: delta proof objective}.
\end{proof} 
We also need the following statement, where we denote the set of integrable functions on G by $L^1(G)$:
\begin{lemma}[Fubini's theorem for Lie groups]\label{lem:fubini}
    Let H be a closed and normal subgroup of G.
    Then $G/H$ again admits a Lie group structure and for a function $\phi \in L^1(G)$ that is constant on the cosets $gH$ it holds true that:
\begin{equation}
    \int_G\phi(g)dg=\int_{G/H}\phi(gH)dgH
\end{equation}
Where $dgH$ is the Haar measure on $G/H$.

\end{lemma}
\begin{proof}
    This is a standard result from abstract harmonic analysis, see for example \cite{loomis_introduction_1953,brocker_representations_1985}.
\end{proof}
Let $R_V,R_W$ be representations of $G$ on the vector spaces $V$ and $W$.
An \textit{interwiner} $\tau$ is a linear function $V \to W$ such that $\tau \circ R_V(g)=R_W(g)\circ \tau$ for all $g \in G$.
Two representations are said to be isomorphic (denoted by $R_V \cong R_W$) if there exists an invertible interwiner between them. We will need the following fundamental results from representation theory:

\begin{lemma}[Schur's Lemma \cite{fulton_representation_2004, sepanski_compact_2007}]\label{lemma: Schur}
    Let $R_V$, $R_W$ be finite dimensional irreducible representations of a compact Lie group G and $\tau$ an interwiner. Then
    \begin{equation}
        \tau = c\cdot \mathbb{1} \; \mathrm{for}\;\mathrm{some} \; c\in S^1\; \mathrm{if}\; R_V\cong R_W\;\mathrm{and}\; \tau =0\;\mathrm{else}
    \end{equation}
\end{lemma}
\begin{lemma}[Othogonality relations]
\label{eq: Schur orthogonality}
Let $R_V$ and $R_W$ be finite dimensional unitary representations of a compact Lie group $G$ with matrix elements $M^V_{ij}$ and $M^W_{nm}$ respectively. In the following $\overline{\cdots}$ stands for complex conjugation. Then
    \begin{equation}
        \int_GM_{ij}^V \; \overline{M^W_{nm}} dg = \begin{cases}
            0 \quad R_V \ncong R_W\\
            \frac{1}{\dim V} \delta_{in}\delta_{nm}
        \end{cases}
    \end{equation}
\end{lemma}

\subsection{Proof Strategy}
\label{sec: proof strategy}
\begin{figure}
    \centering
    \includegraphics[width=0.5\linewidth]{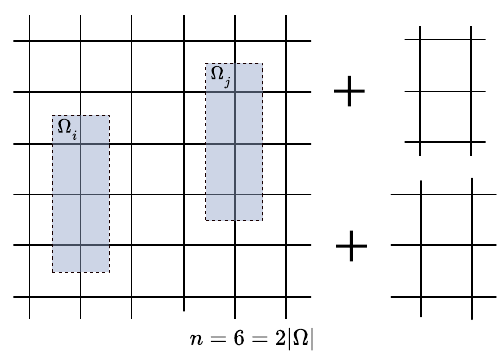}
    \caption{\textbf{Visualization of the proof strategy for the formula for $C^{sw}(i,j)$:} The order parameter $O_{\Omega}$ is supported on a region $\Omega$ containing three lattice sites, i.e., $|\Omega|=3$. $\Omega_i$ and $\Omega_j$ are translations of $\Omega$ and $C^{sw}(i,j)$ measures correlations between $\Omega_i$ and $\Omega_j$. When taking the limit $N\to\infty$, we can increase the system size in steps of $n=2|\Omega|=6$. }
    \label{fig:visualization L infinity}
\end{figure}

\paragraph{Notation:} We will frequently make use of the notation
\begin{equation}
    Ad^{M}(g)O=u_g^{\otimes M}Ou_g^{ \dagger\otimes M}
\end{equation}
for an operator $O$ on $\mathcal{H}_{\mathrm{loc}}^{\otimes M}$. This is the adjoint representation. We will also use $d=\dim \mathcal{H}_{\mathrm{loc}}$ interchangeably.

\paragraph{Proof strategy:} To show Thm.~\ref{theorem: SWSSB of MMIS} we will first consider sequences of the form $N_L=N_0+L\cdot n$ for an n-regular scalar symmetry sector and later lift the result to arbitrary sequences $N_L$ with $\dim V_\theta(N_L)>0$.
Although this could be done in one step, we split the argument for pedagogical reasons.
We also assume $2|\Omega|=n$, where $|\Omega|$ is the size of the support of the order parameter (see Fig. \ref{fig:visualization L infinity}).
There is no loss of generality here: If $2|\Omega|<n$ we can increase the definition of $\Omega$ by extending the order parameter trivially to the neighboring sites.
It is easy to see that this redefinition of $\Omega$ does not change the value of $C^{sw}$ in Thm.~\ref{theorem: SWSSB of MMIS}.
If $2|\Omega|>n$ we can pick an integer $k$ such that $2|\Omega|<kn$ and treat the symmetry sector $V_\theta$ as a $kn$-regular symmetry sector. 
\paragraph{Implications of $\dim V_\theta(N_L) \sim \frac{1}{\mathrm{poly}(N_L)}d^{N_L}$:} We will frequently make use of the following two consequences of this assumption:
\begin{align}
    \lim_{L \to \infty}\frac{V_\theta (N_0+(L-1)\cdot n)}{V_\theta (N_0+L\cdot n)}=d^{-n}=d^{-2|\Omega|} \quad &\mathrm{or}\,\mathrm{more}\,\mathrm{generally}\lim_{L \to \infty}\frac{V_\theta( N_L +M)}{V_\theta ( N_L)}=d^M\label{eq: constant}\\
    \lim_{L\to \infty}\frac{q^{N_0 + L\cdot n}}{\dim V_\theta(N_0+L\cdot n)}=0 \quad\mathrm{for}\;\mathrm{all}\; q<d \quad &\mathrm{or}\,\mathrm{more}\,\mathrm{generally} \lim_{L\to \infty}\frac{q^{N_L}}{\dim V_\theta(N_L)}=0 \label{eq: exponential growth}
\end{align}
In Eq. \eqref{eq: constant} we used $n=2|\Omega|$
\paragraph{How to apply Lemma \ref{lemma: delta}:} 
Our proof consists of expressing the MMIS as an integral over the group $G$ as in Eq.~\eqref{eq: definition of maximally mixed}.
The evaluation of the Rényi-2 correlator involves two integrals (as we will find in App.~\ref{sec: proof of Thm 1}), one of which will be eliminated by Lem.~\ref{eq: Schur orthogonality} and the other one by Lem.~\ref{lemma: delta} for an appropriate sequence of functions $f_L$.
In our case we will pick the function
\begin{equation}
f_L(g) = e^{-i\theta_g}\Tr[u_g^{\otimes (N_0+L\cdot n)}]=e^{-i\theta_g}\Tr[u_g]^{N_0+L\cdot n},
\label{eq:importantfunc}\end{equation}
which gives $\int_Gf_L(g)dg = \dim V_\theta(N_0+L\cdot n)$.
To apply Lem.~\ref{lemma: delta} we need to check that $\lim_{N\rightarrow \infty}\frac{|\Tr[u_g]|^{N_0+L\cdot n}}{V_\theta(N_0+L\cdot n)}\to 0$ for $g\neq e$.
Following \cite{sala_spontaneous_2024} we use that $|\Tr[u_g]|\leq \dim \mathcal{H}_{\mathrm{loc}}$ with equality if and only if $u_g \propto \mathbb{1}$.
If $e$ was the only element acting proportional to the identity, then the assumptions of Lemma \ref{lemma: delta} are fulfilled, since we have $\frac{q^{N_L}}{\dim V_{{\theta}}(N)}\to 0$ for any $q<\dim \mathcal{H}_{\mathrm{loc}}$ according to Eq.~\eqref{eq: exponential growth}.
However, we cannot rule out the case where a subgroup $H\subset G$ acts as a global phase such that $u_h = e^{i\gamma_h}\mathbb{1}$ for $h \in H$.
This can be dealt with by using Lem.~\ref{lem:fubini} to divide out the subgroup $H$ such that for the resulting group $G/ H$ only $e$ acts proportional to $\mathbb{1}$, making the application of Lemma \ref{lemma: delta} possible.
Lem.~\ref{lemma: delta} is applicable in this case since $H$ is directly seen to be normal and is also closed as the preimage of the closed set $S^1\mathbb{1}$ under a continuous map (representations of compact Lie groups are smooth by definition~\cite{sepanski_compact_2007}).
Further we need to check that all functions involved are constant on the cosets $gH$ to enable the application of Lemma \ref{lemma: delta} on $G/H$.
For the function $f_L(g)$ of Eq.~(\ref{eq:importantfunc}) we are interested in, we need to check that for $f_L(g h) = e^{-i\theta_{gh}} \Tr[u_{gh}]^{N_0+L\cdot n}$ is independent of $h$ if $h \in \ker Ad$.
Noting that $\theta_{gh} = \theta_g + \theta_h$, and $u_{gh} = u_g e^{i \gamma_h}$, we obtain independence of $h$ only if $(N_0+L\cdot n) \gamma_h = \theta_h\ \mod\ 2\pi$.
To see that this holds we make the following observations:
\begin{itemize}
    \item For all $\ket{\Psi}\in V_\theta(N_0+L\cdot n)$ we have that $U_g\ket{\Psi}=e^{i\theta_g}\ket{\Psi}$. 
This holds in particular for all elements $h \in H \subset G$, which act like $U_h \ket{\Psi}=e^{i\gamma _h}\ket{\Psi}$ everywhere. 
Hence if $\dim V_\theta(N_0+L\cdot n)>0$ then $\gamma_h =\theta_h \mathrm{mod} 2\pi$ for all $h \in H$
\item If $V_\theta$ is $n$-regular, then the following holds for states $\ket{\Psi_{N_0}}$,$\ket{\Psi_{N_0+n}}$ on $N_0$/$N_0+n$ sites:
\begin{align}
    u_h^{\otimes {N_0}}\ket{\Psi_{N_0}}=e^{i{N_0}\gamma_h}\ket{\Psi_{N_0}},\quad u_h^{\otimes ({N_0}+n)}\ket{\Psi_{{N_0}+n}}=e^{i({N_0}+n)\gamma_h }\ket{\Psi_{{N_0}+n}} \label{eq: discussion}
\end{align}
If we now take both states to belong to $V_\theta$ (which is possible by $n$-regularity), then we find ($\mathrm{mod} \;2\pi$) that $\theta_h =N_0\gamma_h$ and $\theta_h =(N_0+n)\gamma_h$ which leaves us with $n\gamma_h=0 \,\mathrm{mod}\,2\pi$.
\end{itemize}
This shows that $f_L(g)$ takes constant values on the cosets $gH$ and thus descends to a function on $G/H$. When applying Lemma \ref{lemma: delta} we will also need to show that the function we test $f_L$ against (i.e., the function $\phi$ in the language of Lemma \ref{lemma: delta} is also well defined on $G/H$.)

\subsection{Evaluating Correlation Functions}
\begin{lemma}[Linear correlations at infinite temperature]\label{lemma: linear correlations}
Ordinary correlation function vanishes for the $\infty-$temperature states. 
\begin{equation}
    \lim_{N\to \infty}\Tr[\rho_{N_L}^\infty O^{\alpha \dagger}_{\Omega_{i_{N_L}}}O^{\alpha }_{\Omega_{j_{N_L}}}]=0 \quad \forall \alpha \in \mathcal{I}
\label{eq:ordcorrfunc}
\end{equation}
\end{lemma}
We will pick a sequence of system sizes of the form $N_L=N_0+L\cdot n$. From here the statement can be lifted to arbitrary increasing sequences of system sizes $\{N_L\}_{N\in \mathbb{N}}$ using steps that will be introduced around Eq.~\eqref{eq: independent sequence calculation}. We stick here with this type of sequence for pedagogical reasons.
\begin{proof}
To begin, note that the function $L(\mathcal{H}_{|\Omega|})\to\mathbb{C}, O \mapsto \Tr[O]$ obeys $\Tr[Ad^{|\Omega|}(g)O]=\Tr[U_gOU_g^\dagger]=\Tr[O]$ and is thus and intertwiner between $L(\mathcal{H}_{|\Omega|})$ with the adjoint representation and $\mathbb{C}$ with the trivial representation. 
Let $\{O^\alpha\}_\alpha$ be a set of local order parameters.
It follows by irreducibility and Schur's Lemma \ref{lemma: Schur} that $\Tr[O^\alpha]=0$ since order parameters transform nontrivially by definition.
For a chain of length $N_L$, $e^{-i\theta_h}U_h = e^{-i\theta_h}e^{i\lambda_h}\mathbb{1} = \mathbb{1}$ for $g \in H$ since on $H$ the functions $\gamma$ and $\theta$ agree with each other [see discussion around Eq.~\eqref{eq: discussion}].
Hence $e^{-i\theta_g}U_g$ takes constant value on the cosets $g H$ which we denote by $\left(e^{-i\theta_g}U_g\right)_{g H}$.
Then from Eq.~(\ref{eq: definition of maximally mixed}) and Lem.~\ref{lem:fubini} we find that we can express $\rho_{N_0+L\cdot n}^\infty$ as:
 \begin{align}
     &\mathrm{dim}(V_\theta(N_0+L\cdot n)) \cdot \rho_{N_0+L\cdot n}^\infty = \int_G e^{-i\theta_g}U_g dg 
     =\int_{G/H} \left( e^{-i\theta_{g} }U_{g}\right)_{gH}  \;d(gH).
     \label{eq: Infinite temp state quotient}
 \end{align}
The ordinary correlation function of Eq.~(\ref{eq:ordcorrfunc}) can then be written as
\begin{equation}
    \lim_{L\to \infty}\Tr[\rho_{N_0+L\cdot n}^\infty O^{\alpha \dagger}_{\Omega_{i_{N_L}}}O^{\alpha }_{\Omega_{j_{N_L}}}] = \frac{1}{\dim(V_\theta(N_0 + L\cdot n))}\int_{G/H} \Tr[\left( e^{-i\theta_{g} }U_{g}\right)_{gH} O^{\alpha \dagger}_{\Omega_{i_{N_L}}}O^{\alpha }_{\Omega_{j_{N_L}}}]  \;d(gH)
\end{equation}
Here we split the trace into two factors corresponding to regions of the size $2|\Omega|$ and $N_0 -(L-1)\cdot n=N_0 -(L-1)\cdot 2|\Omega|$ (since $2|\Omega|=n$ by assumption).
Now consider $N$ large enough such that $\Omega_{i_{N_L}}\cap \Omega_{j_{N_L}} = \emptyset$.
We can then write for all $g H \in G/H$:
\begin{equation}
    \Tr[(e^{-i\theta_{g}}U_{g})_{g H}O^{\alpha \dagger}_{\Omega_{i_{N_L}}}O^{\alpha }_{\Omega_{j_{N_L}}}]=\left(\Tr[u_{g
     }^{\otimes |\Omega|}O^{\alpha \dagger}] \Tr[u_{g
     }^{\otimes |\Omega|}O^{\alpha}]\right)_{g H} \left(e^{-i\theta_{g}}\Tr[u_{g}]^{N-2|\Omega|}\right)_{g H},
\end{equation}
since $2|\Omega|=n$ we have $n\gamma_h=2|\Omega|\gamma_h=0\;\mathrm{mod}\ 2\pi$ [see the discussion around Eq.~\eqref{eq: discussion}] and thus $\Tr[u_{g}^{\otimes |\Omega|}O^{\alpha }] \Tr[u_{g}^{\otimes |\Omega|}O^{\alpha }]$ takes constant value on the cosets $gH$.
Taking things together, we find that:
 \begin{align}
&\limsup_{L\to\infty}\left|\Tr[\rho_{N_0+L\cdot n}^\infty O^{\alpha \dagger}_{\Omega_{i_{N_L}}}O^{\alpha }_{\Omega_{j_{N_L}}}]\right| \\
 &\leq \underbrace{\lim_{L\to\infty}\left|\frac{\mathrm{dim}(V_\theta(N_0+(L-1)\cdot n)}{\mathrm{dim}(V_\theta(N_0+L\cdot n))}\right|}_{d^{-n}=d^{-2|\Omega|}}\times \lim_{L\to\infty}\Bigg|\int_{G/H} \underbrace{\left(\Tr[u_{g}^{\otimes |\Omega|}O^{\alpha \dagger}]\Tr[u_{g
     }^{\otimes |\Omega|}O^{\alpha }]\right)_{gH}}_{\mathrm{evaluate}\,\mathrm{at}\,gH=e} \underbrace{\frac{(e^{-i\theta_{g}}\Tr[u_{g}]^{N_0-(N-1)\cdot n})_{gH}}{\mathrm{dim}(V_\theta(N_0+(L-1)\cdot n))}}_{=\frac{f_{L-1}(gH)}{\int_{G/H} f_{L-1}(gH)}}\;d(g H )\Bigg| \nonumber
    \end{align}
The last factor becomes a Dirac-delta centered at the identity element of $G/H$ as per our discussion in Sec.~\ref{sec: proof strategy}.
Hence we can solve the integral by evaluating the remaining part at the identity to find:
\begin{align}
\limsup_{L\to\infty}\left|\Tr[\rho_{N_0+L\cdot n}^\infty O^{\alpha \dagger}_{\Omega_{i_{N_L}}}O^{\alpha }_{\Omega_{j_{N_L}}}]\right|\leq d^{-2|\Omega|}\left|\Tr[O^{\alpha \dagger}]\Tr[O^{\alpha }] \right|=0, 
\end{align}
\end{proof}
\subsection{Proof of Theorem}
\label{sec: proof of Thm 1}
\begin{proof}
Let us introduce the notation:
\begin{align}
    C^{sw}_{O^\alpha P}(i_{N_L},j_{N_L})[\rho^\infty_{N_0+L\cdot n}]= \frac{\Tr[\rho^\infty_{N_0+L\cdot n}O^{\alpha\dagger}_{\Omega_{i_{N_L}}}O^{\alpha}_{\Omega_{j_{N_L}}}\rho^\infty_{N_0+L\cdot n}P^{\dagger}_{\Omega_{j_{N_L}}}P^{}_{\Omega_{i_{N_L}}}]}{\Tr[(\rho_{N_0+L\cdot n}^\infty)^2]}.
\end{align}
Since $\rho_{N_0+L\cdot n}^\infty$ is a projection normalized to $\Tr[\rho_{N_0+L\cdot n}^\infty]=1$,
we find immediately that 
\begin{equation}
    \Tr[(\rho_{N_0+L\cdot n}^\infty)^2]=\frac{\Tr[\rho_{N_0+L\cdot n}^\infty]}{\mathrm{dim}V_\theta(N_0+L\cdot n)}=\mathrm{dim}V_\theta(N_0+L\cdot n)^{-1}
\end{equation}
Again, expressing the MMIS in terms of the quotient group $G/H$ as in Eq.~\eqref{eq: Infinite temp state quotient} we find that
\begin{align}
    &\lim_{L\to\infty}C^{sw}_{O^\alpha P}(i_{N_L},j_{N_L})= \lim_{L\to\infty}\frac{\Tr[\rho^\infty_{N_0+L\cdot n}O^{\alpha\dagger}_{\Omega_{i_{N_L}}}O^{\alpha}_{\Omega_{j_{N_L}}}\rho^\infty_{N_0+L\cdot n}P^{\dagger}_{\Omega_{j_{N_L}}}P^{}_{\Omega_{i_{N_L}}}]}{\Tr[(\rho_{N_0+L\cdot n}^\infty)^2]} \nonumber \\
    =&\lim_{L\to\infty}(\mathrm{dim}V_\theta(N_0+L\cdot n))^{-1}
   \int_{G/H)}d(g_2 H) \int_{G/H}d(g_1 H)\left(\Tr[\left(e^{-i\theta_{g_1}}U_{g_1}\right)_{g_1 H}O^{\alpha\dagger}_{\Omega_{i_{N_L}}}O^{\alpha}_{\Omega_{j_{N_L}}}\left(e^{-i\theta_{g_2}}U_{g_2} \right)_{g_2 H}P^{\dagger}_{\Omega_{j_{N_L}}}P_{\Omega_{i_{N_L}}}]\right) \nonumber \\
        =&\lim_{L\to \infty}\frac{\dim V_\theta (N_0+(L-1)\cdot n)}{\dim V_\theta (N_0+L\cdot n)}\lim_{L \to \infty}\int_{G/H}d(g_2 H) \int_{G/H}d(g_1 H) \Bigg(\nonumber \\
        &\hspace{1cm}\underbrace{\Tr[u_{g_1 H}^{\otimes 2|\Omega|}O^{\alpha\dagger}_{\Omega_{i_{N_L}}}O^{\alpha}_{\Omega_{j_{N_L}}}u_{g_2 H}^{\otimes 2|\Omega|}P^{\dagger}_{\Omega_{j_{N_L}}}P_{\Omega_{i_{N_L}}}]}_{\mathrm{evaluate}\,\mathrm{at\,}g_2=g_1^{-1}}
        \underbrace{\frac{\left(e^{-i\theta_{g_1g_2}}\Tr[u_{g_1 g_2}]^{N_0+(L-1)\cdot n}\right)_{g_1 H \cdot g_2H }}{\mathrm{dim}V_\theta(N_0+(L-1)\cdot n)}}_{f_{L-1}(g_1H\, g_2H)/\int_{G/H}f_{L-1}(g_1H\,g_2 H)} \Bigg)\label{eq: expand R2 expression}
\end{align}
In the above expression, we have split the trace into a product of traces over $\mathcal{H}_{2|\Omega|}=\bigotimes_{i \in \Omega_{i_{N_L}}\cup \Omega_{j_{N_L}}} \mathcal{H}_{\mathrm{loc}}$ in the first trace and the rest of $\mathcal{H}_{N_L}$ in the second trace.
Now we define $\widetilde{g} = g_1 g_2$ and use left invariance of the Haar measure.
Then as before, we apply Lem.~\ref{lemma: delta} to the second factor, which eliminates one integral by evaluating the first factor at $\widetilde{g}=e$.
Then, the above can be written as:
\begin{equation}
    \lim_{L \to \infty}C^{sw}_{O^\alpha P}(i_{N_L},j_{N_L})=d^{-2|\Omega|} \int_{G/H}d(g H) \left(\Tr[u_{g}^{\otimes |2\Omega|} O^{\alpha\dagger}_{\Omega_{i_{N_L}}}u^{{\otimes 2|\Omega|}\dagger}_{g}u^{\otimes 2|\Omega|}_{g}O^{\alpha}_{\Omega_{j_{N_L}}}u^{{2|\Omega|}\dagger}_{g}P^{\dagger}_{\Omega_{j_{N_L}}}P_{\Omega_{i_{N_L}}}]\right)_{g H}.
    \label{eq: R2 proof one integral gone}
\end{equation}
We also used that $\lim_{L\to \infty}\frac{\dim V_\theta(N_0+(L-1)\cdot n)}{\dim V_\theta(N_0+L\cdot n)}$.
The integrand in the above expression is well defined on the quotient $G/H$ since $H$ acts via multiplication with a phase.
Under the adjoint action of $G$ on $\mathcal{H}_{2|\Omega|}$, $O^{\alpha}_{\Omega}$  transforms under an irreducible representation with matrix elements $M_{\alpha \beta}$
\begin{align}
    Ad^{2|\Omega|} (g)[O^\alpha_\Omega]=u_g^{\otimes2|\Omega|}O^\alpha_\Omega (u_g^\dagger)^{\otimes2|\Omega|}=\sum_{\beta}M_{\alpha \beta}(g)O^{\beta}_\Omega \label{eq: matrix elements O}
\end{align}
This descends to an irreducible representation  $Ad^{2|\Omega|}(gH)$ on the quotient $G/H$ with the same matrix elements, since dividing out $H$ only forces the adjoint representation to be faithful but does not change the matrix elements.
By inserting Eq.~\eqref{eq: matrix elements O} into \eqref{eq: R2 proof one integral gone} and expressing the $uOu^\dagger$ in terms of the adjoint representation, we obtain
\begin{align}
&\lim_{L \to \infty}C^{sw}_{O^\alpha P}(i_{N_L},j_{N_L})=d^{-2|\Omega|}\int_{G/H}d(gH)\Tr\left[Ad^{2|\Omega|}(gH)[O^{\alpha}_{\Omega_{i_{N_L}}}] \overline{Ad^{2|\Omega|}}(gH)[O^{\alpha}_{\Omega_{j_{N_L}}}]P^{\dagger}_{\Omega_{j_{N_L}}}P_{\Omega_{i_{N_L}}}\right] \nonumber \\
    &=d^{-2|\Omega|} \sum_{\beta \beta'}\int_{G/H}d(gH)M_{\alpha \beta}(gH)\overline{M_{\alpha \beta'}}(gH)\Tr[O^{\beta\dagger}_{\Omega_{i_{N_L}}}O^{\beta'}_{\Omega_{j_{N_L}}}P^{\dagger}_{\Omega_{j_{N_L}}}P_{\Omega_{i_{N_L}}}]\nonumber \\
    &=\frac{1}{|\mathcal{I}_O|d^{2|\Omega|}}\sum_\beta\Tr[O^{\beta\dagger}_{\Omega_{i_{N_L}}}O^{\beta}_{\Omega_{j_{N_L}}}P^{\dagger}_{\Omega_{j_{N_L}}}P_{\Omega_{i_{N_L}}}] =\frac{1}{|\mathcal{I}_O|\dim (\mathcal{H}_{\mathrm{loc}})^{2|\Omega|}}\sum_\beta\left|\Tr[O^{\beta \dagger} P ]\right|^2,
\end{align}
where $\overline{\cdot}$ denotes element wise complex conjugation. 
From the second to the third line we used Schur orthogonality relations from Lem.~\ref{eq: Schur orthogonality}. 
This already yields the formula for the Rényi-2 correlator as presented in Thm.\ref{theorem: SWSSB of MMIS}.
As a mathematical detail, we have so far assumed that we increase the system size in steps of $n$. 
In general this need not be the case and we can consider arbitrary increasing sequences $\{N_L\}_{L \in \mathbb{N}}$, for as long as $\dim V_\theta(N_L)>0$ for all $L$, i.e. the MMIS is defined at every system size. 

\textbf{Independence of the sequence:}
While the above concludes the proof for the particular sequence $N_L=N_0+L\cdot n$ that we chose to approach the thermodynamic limit, we can also check that the result does not depend on this choice of sequence.
Consider an arbitrary increasing sequence $\{ N_L\}_{N \in \mathbb{N}}$ such that $\dim V_\theta(N_L)>0$ for all $N$.
Denote by $\widetilde{L}$ the largest integer such that $N_0 +\widetilde{L}\cdot n \leq N_L$.
The crucial step now is that the difference $\Delta N = N_L-(N_0 +\widetilde{L}\cdot n)$ only takes values $\Delta N = 0,1,\ldots ,n$ since the scalar symmetry sector is $n$-regular.
For a given $K \in \mathbb{N}$ we will consider the subsequence $\Lambda^K_{N}$ which contains only those $N_L$ with $\Delta N = K$.
Every $N_L$ will be contained in one of those subsequences. Hence if we can prove that $C^{sw}$ converges to the same value for all sequences $\Lambda^K_N$, we show that $\lim_{L\to \infty} C^{sw}[\rho_{N_L}]$ exists and is equal to $\lim_{L\to \infty} C^{sw}[\rho_{N_L^K}]$. We can proceed as in Eq.~\eqref{eq: expand R2 expression}, but instead of splitting the trace into two parts, we now split into three parts, corresponding to regions of the size $2|\Omega|=n$, $N_0+(\widetilde{L}-1)\cdot n$ and $K=\Lambda_K -(N_0+\widetilde{L}\cdot n)$:
\begin{align}
    &\lim_{L\to\infty}C^{sw}_{O^\alpha P}(i_{N_L},j_{N_L})[\rho^\infty_{\Lambda^K_N}]= \lim_{L\to\infty}\frac{\Tr[\rho^\infty_{\Lambda^K_N}O^{\alpha\dagger}_{\Omega_{i_{N_L}}}O^{\alpha}_{\Omega_{j_{N_L}}}\rho^\infty_{\Lambda^K_N}P^{\dagger}_{\Omega_{j_{N_L}}}P^{}_{\Omega_{i_{N_L}}}]}{\Tr[(\rho_{\Lambda^K_N}^\infty)^2]} \nonumber \\
    =&\lim_{L\to\infty}(\mathrm{dim}V_\theta({\Lambda^K_N}))^{-1}
   \int_{G/H)}d(g_2 H) \int_{G/H}d(g_1 H)\left(\Tr[\left(e^{-i\theta_{g_1}}U_{g_1}\right)_{g_1 H}O^{\alpha\dagger}_{\Omega_{i_{N_L}}}O^{\alpha}_{\Omega_{j_{N_L}}}\left(e^{-i\theta_{g_2}}U_{g_2} \right)_{g_2 H}P^{\dagger}_{\Omega_{j_{N_L}}}P_{\Omega_{i_{N_L}}}]\right)\nonumber \\
    =&\lim_{L \to \infty}\underbrace{\frac{\dim V_\theta(N_0 +(\widetilde{L} -1)\cdot n)}{\dim V_\theta(N_L^K)}}_{\to d^{N_0 +(\widetilde{L} -1)-\Lambda_K }=d^{-n-K}}\int_{G/H}d(g_2 H) \int_{G/H}d(g_1 H) \Bigg(\nonumber \\
        &\hspace{1cm}\underbrace{\Tr[u_{g_1 H}^{\otimes 2|\Omega|}O^{\alpha\dagger}_{\Omega_{i_{N_L}}}O^{\alpha}_{\Omega_{j_{N_L}}}u_{g_2 H}^{\otimes 2|\Omega|}P^{\dagger}_{\Omega_{j_{N_L}}}P_{\Omega_{i_{N_L}}}]\Tr[u_{g_1H}^{\otimes K}u_{g_2H}^{\otimes K}]}_{\mathrm{evaluate}\,\mathrm{at}\, g_2=g_1^{-1}} 
        \underbrace{\frac{\left(e^{-i\theta_{g_1g_2}}\Tr[u_{g_1 g_2}]^{N_0+(L-1)\cdot n}\right)_{g_1 H \cdot g_2H }}{\mathrm{dim}V_\theta(N_0+(L-1)\cdot n)}}_{f_{L-1}(g_1H\, g_2H)/\int_{G/H}f_{L-1}(g_1H\,g_2 H)}\Bigg) \label{eq: independent sequence calculation}
\end{align}
where we again used Eq.~\eqref{eq: constant} to calculate the limit $\lim_{L \to \infty} \frac{\dim V_\theta(N_0 +(\widetilde{L} -1)\cdot n)}{\dim V_\theta(N_L^K)} = d^{-n-K}$. Compared with Eq.~\eqref{eq: expand R2 expression} there is an additional factor of $\Tr[u_{g_1H}^{\otimes K}u_{g_2H}^{\otimes K}]$. (Using the same argument as in Eq.~\eqref{eq: discussion} we can show that this is constant on the cosets $gH$ and thus is a well defined function on $G/H$). The presence of this factor does not impact the application of Lemma \ref{lemma: delta} to the last factor and we can still evaluate the remaining part of the integral at $g_2=g_1^{-1}$. Since $\Tr[u_{Hg_1}^{\otimes K}u_{Hg_1^{-1}}^{\otimes K}]=d^K$, this gives:
\begin{align}
    &\lim_{L \to \infty}C^{sw}_{O^\alpha P}(i_{N_L},j_{N_L})=d^{-2|\Omega|-K}d^K \int_{G/H}d(g H) \left(\Tr[u_{g}^{\otimes |2\Omega|} O^{\alpha\dagger}_{\Omega_{i_{N_L}}}u^{{\otimes 2|\Omega|}\dagger}_{g}u^{\otimes 2|\Omega|}_{g}O^{\alpha}_{\Omega_{j_{N_L}}}u^{{2|\Omega|}\dagger}_{g}P^{\dagger}_{\Omega_{j_{N_L}}}P_{\Omega_{i_{N_L}}}]\right)_{g H}
\end{align}
This is the same expression as in Eq.~\eqref{eq: R2 proof one integral gone} and we can finish the proof as before.
\end{proof}

\section{Maximally Mixed Invariant State (MMIS) for a \texorpdfstring{$U(1)$}{Lg} Symmetry}
\label{sec: U1 calculations}
In this section, we provide an example of the MMIS of a $U(1)$ symmetry.
We will calculate the Rényi-2 correlator both directly and by using the formula Eq.~\eqref{eq: formula R2 maximally mixed}.
We consider a system of $N$ qubits and the $U(1)$ symmetry is generated by $\sum_iZ_i$ and the scalar symmetry sector we consider is that of charge $Q$, i.e., $V_Q(N)=\{\ket{\Psi} \in \mathcal{H}_N| \sum_iZ_i \ket\Psi = Q\ket{\Psi}\}$.
We keep $Q$ constant when taking the limit $N\to \infty$ and allow for arbitrary charge $Q$. $V_Q(N)$ is spanned by all product states in the $Z$-basis that have $N_{\uparrow}=\frac{N+Q}{2}$ spins pointing up.
\subsection{Application of the formula}
To be able to apply Thm. \ref{theorem: SWSSB of MMIS} we need to show that $\dim V_Q(N) \sim \frac{1}{\mathrm{poly}(N)}2^N$ whenever it is nonzero. 
Using Sterling's approximation $k! \sim \sqrt{2\pi k}\left(\frac{k}{e}\right)^k$ we find that for $N\to \infty$:
\begin{equation}
    \dim V_Q(N) = \binom{N}{N_\uparrow} = \frac{N!}{N_\uparrow! (N-N_\uparrow )!} = \frac{N!}{(\frac{N}{2}+\frac{Q}{2})! (\frac{N}{2}-\frac{Q}{2})! }\sim \frac{N!}{\left[\left(\frac{N}{2}\right)!\right]^2}\sim \frac{2}{\sqrt{2\pi N}}2^N,
\end{equation}
where we used that $N\gg Q$ which is always true since we keep $Q$ fixed and take the limit $N\to \infty$.
Hence Theorem \ref{theorem: SWSSB of MMIS} holds.
We take the order parameter to be supported on a single site ($|\Omega|=1$ in the language of Theorem \ref{theorem: SWSSB of MMIS}).
This choice is possible for Thm. \ref{theorem: SWSSB of MMIS} to work, since the $U(1)$ symmetry is 2-regular.
Using the generalized version of the Rényi-2 correlators from Eq.~\eqref{eq: generalized Rényi-2}, we can apply the formula to other closely related correlators that are commonly studied in the context of symmetry breaking, such as those of the inner product $\Tr[\rho\  \mathbf{S}_i\cdot\mathbf{S}_j]$.
Here $\mathbf{S}_i = \frac{1}{2}(X_i, Y_i, Z_i)$, where $(\cdots)$ denotes the vector of Pauli operators at site $i$.
Ref.~\cite{sala_spontaneous_2024} proposed to extend this to SW-SSB.
Thus we wish to compute 
\begin{align}
    C^{sw}_{U(1)} = \frac{\Tr[\rho \mathbf{S}_i \cdot \mathbf{S}_j\rho\mathbf{S}_j \cdot \mathbf{S}_i]}{\Tr[\rho^2]} =\frac{1}{4}\sum_{O,P\in \{X,Y,Z\} }\frac{\Tr[\rho O_i^\dagger O_j\rho P_i^\dagger P_j]}{\Tr[\rho^2]}.
\end{align}
The correlation functions on the right hand side can then be evaluated using Eq.~\eqref{eq: generalized Rényi-2} from Thm.~\ref{theorem: SWSSB of MMIS}.

$Z_i$ transforms under the trivial (hence $|\mathcal{I}_Z|=1$) representation of $U(1)$ and $X_i$ and $Y_i$ span a two-dimensional real irrep ( $|\mathcal{I}_{XY}|=2$).
Hence using the formula in Eq.~(\ref{eq: generalized Rényi-2}) with $\dim \mathcal{H}_{\mathrm{loc}}=2$ and $|\Omega|=1$, we find:
\begin{equation}
    C^{sw}_{U(1)}=\frac{1}{4}\Bigg(\frac{1}{|\mathcal{I}_Z|2^2}\sum_{P \in \{X,Y,Z\}}\left|\Tr[P^\dagger Z]\right|^2 +\frac{1}{|\mathcal{I}_{XY}|2^2}\sum_{P \in \{X,Y,Z\}}\left(\left|\Tr[P^\dagger X]\right|^2+\left|\Tr[P^\dagger Y]\right|^2\right)\Bigg)=\frac{3}{16}.
\label{eq:formularesultU(1)}
\end{equation}
\subsection{Direct calculation}
To confirm the result of Eq.~(\ref{eq:formularesultU(1)}), we can also calculate $C^{sw}_{U(1)}$ directly. 
We use that $\mathbf{S}_i\cdot\mathbf{S}_j=\frac{1}{4}\left(P_{ij}-\mathbb{1}\right)$, where $P_{ij}$ swaps sites $i$ and $j$.
In the symmetry sector with $Q$ as the eigenvalue of $\sum_iZ_i$ on a chain of $L$ qubits, there are $N_\uparrow=\frac{Q+L}{2}$ $\uparrow$ spins.
Let $\{\ket{b}\}$ be a product state in the $Z$ basis with $N_{\uparrow}$ spins pointing up. Then we can write 
\begin{equation}
    \rho = \binom{N}{N_\uparrow}^{-1}\sum_{b \in V_Q(N)}\ket{b}\bra{b}
\end{equation}
Hence $\Tr[\rho^2]=\binom{N}{N_\uparrow}^{-1}$.
We have the following
\begin{align}
    \Tr[P_{ij}\rho^2] &= \binom{N}{N_\uparrow}^{-2}\sum_{b \in V_Q(N)}\bra{b}P_{ij}\ket{b} =\binom{N}{N_\uparrow}^{-2}\left(\binom{N-2}{N_\uparrow-2}+\binom{N-2}{N_\uparrow}\right) \\
    \Tr[P_{ij}\rho P_{ji}\rho] &= \binom{N}{N_\uparrow}^{-2} \sum_{b,b' \in V_Q(N)}|\bra{b} P_{ij}\ket{b'}|^2 =\binom{N}{N_\uparrow}^{-1}.
\end{align}
Hence $C^{sw}(i,j)$ is independent of $|i-j|$ and thus we find
\begin{equation}
    C^{sw}(i,j)[\rho]=\frac{1}{16}\frac{N-3N^2+4Q^2}{N(1-N)} \xrightarrow[N\to\infty]{}\frac{3}{16}
\end{equation}
Hence as $N\to \infty$ our formula produces the correct result.

\section{Bounds for the fidelity correlator}
\label{sec: fidelity bounds}
In this section we will consider symmetric thermal states $\rho_\beta = \frac{P e^{-\beta H}}{\Tr[Pe^{-\beta H}]}$ where $H$ is a $G$-symmetric Hamiltonian and $P$ is the projection onto a scalar symmetry sector of $G$.
Note that by virtue of being a \textit{scalar} symmetry sector, we have that $[P, H] = 0$.
We will now show the following lemma.
\begin{lemma}[Bounds for the fidelity-correlator]\label{lemma: bounds for fidelity}
    Let $\rho_\beta$ be a symmetric thermal state and $O^\alpha_\Omega$ a local order parameter. Then \label{lemma: bounds}
    \begin{equation}
        \frac{C^{sw}(\Omega_1,\Omega_2)[\rho_{\beta/2}]}{\norm{O^\alpha}^2_\infty}\leq F_{O^\alpha_\Omega}(\Omega_1,\Omega_2)[\rho_\beta]\leq \norm{O^\alpha}_\infty^2
    \end{equation}
Where the upper bound holds for all states $\rho$.
\end{lemma}
To set the stage we recall the definition of the Schatten p-norm for any operator $A$ on the Hilbert space $\mathcal{H}$, i.e., $A \in L(\mathcal{H})$:
\begin{equation}
    \norm{A}_p=\Tr[|A|^p]^{\frac{1}{p}}.
\end{equation}
Note that we will often use the identity $\norm{A}_1 = \Tr[\sqrt{A^\dagger A}]$.
$\norm{\cdot}_p$ has the following important properties:
\paragraph{Hölder inequality} (1.174 in \cite{watrous_theory_2018})
Let $p,q \in [1,\infty]$ with $\frac{1}{p}+\frac{1}{q}=1$ (For $p=1$ set $q=\infty$) and $A,B \in L(\mathcal{H})$
\begin{equation}
    \Tr[A^\dagger B] \leq \Tr[|A^\dagger B|]=||A^\dagger B||_1 \leq ||A||_p||B||_q \label{eq: H-Inequality}
\end{equation}

\paragraph{Duality:} (1.173 in \cite{watrous_theory_2018})
For $p,q \in [1,\infty]$ with $\frac{1}{p}+\frac{1}{q}=1$ (For $p=1$ set $q=\infty$) and $O \in   L(\mathcal{H})$:
\begin{equation}
    ||O||_p = \max\left\{\left|\Tr(O^\dagger A)\right|, ||A||_q \leq 1\right\} \label{eq: H-Duality}
\end{equation}
\textbf{Proof of Lemma \ref{lemma: bounds for fidelity}:}
\begin{proof}
The lower bound is obtained as follows by using Eq.~(\ref{eq: H-Duality}), i.e.,  
\begin{align*}
    &F_O(\Omega_1,\Omega_2)[\rho_\beta]=\Tr\left[\sqrt{\sqrt{\rho_\beta}O_{\Omega_1}^\dagger O_{\Omega_2} \rho_\beta O_{\Omega_2}^\dagger O_{\Omega_1}\sqrt{\rho_\beta}}\right]
    =\Tr\left[\sqrt{\sqrt{\rho_\beta}O_{\Omega_1}^\dagger O_{\Omega_2} \sqrt{\rho_\beta}\sqrt{\rho_\beta} O_{\Omega_2}^\dagger O_{\Omega_1}\sqrt{\rho_\beta}}\right]\\
   & =\Tr\left[\sqrt{\left(\sqrt{\rho_\beta} O_{\Omega_2}^\dagger O_{\Omega_1}\sqrt{\rho_\beta}\right)^\dagger\sqrt{\rho_\beta} O_{\Omega_2}^\dagger O_{\Omega_1}\sqrt{\rho_\beta}}\right]=||\sqrt{\rho_\beta}O_{\Omega_1}^\dagger O_{\Omega_2}\sqrt{\rho_\beta}||_1\\
    &=\max\left\{\left|\Tr[\sqrt{\rho_\beta}O_{\Omega_2}^\dagger O_{\Omega_1}\sqrt{\rho_\beta} A]\right|, ||A||_\infty \leq 1\right\}
    \geq \Tr\left[\sqrt{\rho_\beta}O_{\Omega_2}^\dagger O_{\Omega_1}\sqrt{\rho_\beta}\frac{O_{\Omega_1}^\dagger O_{\Omega_2}}{{||O_{\Omega_1}^\dagger O_{\Omega_2}||_\infty}} \right]\\
    &=
    \frac{\Tr[Pe^{-\frac{\beta H}{2}}O_{\Omega_1}^\dagger O_{\Omega_1} Pe^{-\frac{\beta H}{2}}O_{\Omega_1}^\dagger O_{\Omega_2}]}{||O||^2_\infty \Tr[Pe^{-\beta H}]}= \frac{C^{sw}({\Omega_1},{\Omega_2})[\rho_{\beta/2}]}{||O||^2_\infty}
\end{align*}
where we used the fact, that $P$ and $H$ commute, such that $\Tr[(Pe^{-\beta H/2})^2]=\Tr[P^2e^{-\beta H}]=\Tr[Pe^{-\beta H}]$ since P is a projection and the fact that $||O^\dagger_{\Omega_1} O_{\Omega_2}||_\infty=||O||_\infty^2$ since $O^\dagger_{\Omega_1}$ and $O_{\Omega_2}$ are supported on different regions on the lattice.
For the upper bound on the fidelity correlator, we can use Eq.~\eqref{eq: H-Inequality} with $p = 2$ (i.e Cauchy-Schwarz inequality).
Defining $\sigma = O_{\Omega_1}^\dagger O_{\Omega_2} \rho_\beta O_{\Omega_2}^\dagger O_{\Omega_1}$, we obtain
\begin{align*}
    &F_O({\Omega_1},{\Omega_2})[\rho]=||\sqrt{\rho}\sqrt{\sigma}||_1\leq ||\sqrt{\rho}||_2||\sqrt{\sigma}||_2=\sqrt{\Tr[O_{\Omega_1}^\dagger O_{\Omega_2} \rho O_{\Omega_2}^\dagger O_{\Omega_1}]}\\
    &\leq\sqrt{||O_{\Omega_1}^\dagger O_{\Omega_2} O_{\Omega_2}^\dagger O_{\Omega_1}||_\infty||\rho||_1}=||O||_\infty^2,
\end{align*}
where we used that $\Tr[\rho]=\Tr[|\rho|]=||\rho||_1=1$ and that $\sqrt{\sigma}^2=\sigma$, since $\sigma$ is positive semidefinite and hermitian and for the last inequality we used Hölder`s inequality \eqref{eq: H-Inequality} with $p = 1$ and $q = \infty$ as well as cyclicity of the trace and $||A^\dagger A||_\infty= ||A||^2_{\infty}$ for all bounded operators A.
\end{proof}

\section{Equivalent measures of SW-SSB}
\label{sec: Equivalent measures of SW-SSB}
In this appendix, we will consider alternative measures of SW-SSB with the goal of showing that they are equivalent to the fidelity.
As mentioned in the main text in Sec. \ref{sec: 2 SW-SSB} these measures are \textit{inequivalent} to the Rényi-2 correlator and (to our knowledge) the only known rigorous relation between them is given by Lemma \ref{lemma: bounds for fidelity} for arbitrary states. 
However, as we argue in Sec.~\ref{subsubsec:identtrans}, for steady-state density matrices, we can resort to the structure of these correlators being boundary correlators acting on the $d$-dimensional boundary of a $(d+1)$-dimensional statistical mechanics model, which cannot change the phases of the system. 
Recall that the fidelity correlator is defined as
$F_O(i,j)=F(\rho,\sigma)$ where $\sigma = O_i^\dagger O_j \rho O_j^\dagger O_i$ for a local order parameter $O_i$.
It was conjectured in Ref.~\cite{lessa_strong--weak_2024} that SW-SSB can equivalently be defined in terms of quantities known as the trace distance or the sandwiched Rényi divergence, their definitions are reviewed below.
The conjecture was motivated by the observation that as the fidelity, trace distance, and divergence are functions that compare $\rho$ and $\sigma$ and obey a data processing inequality.
In this section we will give a positive answer to that conjecture.
Denote the set of positive semidefinite operators on the finite dimensional Hilbert space $\mathcal{H}$ by $P(\mathcal{H})$.
Let us recall the definitions we will use in the following.
\paragraph{Fidelity:} For $\rho,\sigma \in P(\mathcal{H})$, define the fidelity as:
\begin{equation}
F(\rho,\sigma)=||\sqrt{\rho}\sqrt{\sigma}||_1= \Tr \left[\sqrt{\sqrt{\sigma}^\dagger\sqrt{\rho}^\dagger\sqrt{\rho}\sqrt{\sigma}} \right] = \Tr[\sqrt{\sqrt{\rho}\sigma\sqrt{\rho}}], 
\end{equation}
where we have also used the fact that the fidelity is symmetric between $\rho$ and $\sigma$ \cite{nielsen_quantum_2010, watrous_theory_2018}.
\paragraph{Trace-distance}
\begin{equation}
    T(\rho,\sigma)=\frac{1}{2}\norm{\rho-\sigma}_1
\end{equation}
\paragraph{Sandwiched Rényi Divergence}
For positive $\rho, \sigma \in P(\mathcal{H})$ and $\alpha <1$ define the sandwiched Rényi divergence as:
\begin{small}
\begin{equation}
\widetilde{D}_\alpha(\rho||\sigma)=
        \frac{1}{\alpha-1}\log\left(\frac{1}{\Tr[\rho]}\Tr\left[\left(\sigma^{\frac{1-\alpha}{2\alpha}}\rho\sigma^{\frac{1-\alpha}{2\alpha}}\right)^\alpha\right]\right)
\end{equation}
\end{small}
It satisfies a few important properties:
\begin{enumerate}[label=(\roman*)]
    \item \textbf{Data-Processing Inequality:} $\widetilde{D}_\alpha(\rho||\sigma) \geq \widetilde{D}_\alpha(\mathcal{E}(\rho)||\mathcal{E}(\sigma))$ for $\frac{1}{2}\leq\alpha\leq \infty$ and any CPTP-map $\mathcal{E}$~\cite{frank_monotonicity2013}
    \item \textbf{Monotonicity:} The map $\alpha \mapsto \widetilde{D}_\alpha(\rho||\sigma)$ is monotonically increasing \cite{muller-lennert_quantum_2014}
\end{enumerate}
\paragraph{Petz-divergence}
The Petz divergence for $\alpha <1$ is defined as:
\begin{equation}
    \bar{D}_\alpha(\rho||\sigma)=
        \frac{1}{1-\alpha}\log\left(\frac{1}{\Tr[\rho]}\Tr\left[\rho^\alpha \sigma^{1-\alpha}\right]\right)
\end{equation}
Note that the Petz divergence satisfies 
\begin{equation}
    \bar{D}_{1-\alpha}(\sigma||\rho)=\frac{1-\alpha}{\alpha}\bar{D}_\alpha(\rho||\sigma)+\frac{1}{\alpha}\log \frac{\Tr[\sigma]}{\Tr[\rho]} \label{eq: Petz switch rho sigma}
\end{equation} 
It is worth mentioning that the last two quantities can also be defined for $\alpha \geq 1$ but are only finite for appropriate choices of $\rho$ and $\sigma$, but we will not consider these cases in the following.
We then have the following equivalent definitions of SW-SSB, resolving a conjecture from \cite{lessa_strong--weak_2024}:
\\
\\

\begin{theorem}[Equivalent definitions of SW-SSB for converging sequences]
    The following definitions of SW-SSB are equivalent
\begin{enumerate}[label=(\roman*)\setlength{\leftmargin}{1cm}]
    \item $\lim_{|i-j|\to \infty}F(\rho,O_i^\dagger O_j\rho O_j^\dagger O_i)>0$
    \item $\lim_{|i-j|\to \infty}T(\rho,O_i^\dagger O_j\rho O_j^\dagger O_i)<\frac{1}{2}\left(1+\lim_{|i-j|\to \infty}\Tr[O_i^\dagger O_j\rho O_j^\dagger O_i]\right)$
    \item $\lim_{|i-j|\to \infty}D_\alpha(\rho||O_i^\dagger O_j\rho O_j^\dagger O_i)<\infty$ for all $\alpha \in[\frac{1}{2},1)$
    \item $\lim_{|i-j|\to \infty}D_\alpha(O_i^\dagger O_j\rho O_j^\dagger O_i||\rho)<\infty$ for all $\alpha \in(0,\frac{1}{2}]$
\end{enumerate}
where $D_\alpha \in \{\widetilde{D}_\alpha,\bar{D}_\alpha\}$ can be the Petz divergence or the sandwiched Rényi divergence.
\end{theorem}

The above statement includes the assumption that all limits involved exist to begin with.
Even though we expect this to be true in all physical settings, from a mathematical point of view we cannot rule out that the sequences involved do not converge.
To address this, the more rigorous version of this statement, in Theorem \ref{theorem: QuantumInfo correlators} replaces the $\lim$ with $\liminf$ and $\limsup$, i.e., the smallest and largest accumulation points of the sequences.
$\liminf_{|i-j|\to \infty}F(\rho,O_i^\dagger O_j\rho O_j^\dagger O_i)>0$ then means that the sequence $F(\rho,O_i^\dagger O_j\rho O_j^\dagger O_i)$ remains bounded away from $0$ as $|i-j|\to \infty$.
Consequently, any convergent subsequence will converge to a positive value.
This ensures that any way of taking the thermodynamic limit, that results in a well-defined value for $F$ as $|i-j|\to \infty$, will yield $F>0$.
The more mathematically precise statement is then:
\begin{theorem}[Equivalent definitions of SW-SSB]\label{theorem: QuantumInfo correlators}
Let $\{\rho_N\}_{N\in \mathbb{N}}$ and $\{\sigma_N\}_{N\in \mathbb{N}}$ be sequences of positive semidefinite operators acting on the finite dimensional Hilbert spaces $\mathcal{H}_N$ for all $N$. Assume that there are constants $C_\rho,C_\sigma>0$ such that $\Tr[\rho_N]\geq C_\rho$ and $\Tr[\sigma_N]\leq C_\sigma$ for all $N\in \mathbb{N}$. Then the following are equivalent:
    
\begin{enumerate}[label=(\roman*)\setlength{\leftmargin}{1cm}]
    \item $\liminf_{N\to \infty}F(\rho_N,\sigma_N)>0$
    \item $\liminf_{N\to \infty}T(\rho_N,\sigma_N)<\frac{1}{2}\left(\liminf_{N\to \infty}\Tr[\rho_N]+\liminf_{N\to \infty}\Tr[\sigma_N]\right)$
    \item $\limsup_{N\to \infty}D_\alpha(\rho_N||\sigma_N)<\infty$ for all $\alpha \in[\frac{1}{2},1)$
    \item $\limsup_{N\to \infty}D_\alpha(\sigma_N||\rho_N)<\infty$ for all $\alpha \in(0,\frac{1}{2}]$

\end{enumerate}
where $D_\alpha \in \{\widetilde{D}_\alpha,\bar{D}_\alpha\}$ can be the Petz divergence or the sandwiched Rényi divergence.
\end{theorem}
The assumptions on the boundedness of the sequences $\rho_N$ and $\sigma_N$ are fulfilled in the context of SW-SSB, since $\Tr[\rho_N]=1$ as $\rho_N$ will always be a density matrix. Since $\sigma_N=O_{i_{N_L}}^\dagger O_{j_{N_L}} \rho_N O_{j_{N_L}}^\dagger O_{i_{N_L}}$ we also have $\Tr[\sigma_N] \leq \norm{O}_\infty^4$ by using the cyclicity of the trace, and applying Hölder's inequality of Eq.~(\ref{eq: H-Inequality}).

To prepare for the proof of the equivalence of the different SW-SSB measures, we will need to collect some results on the relation between different measures that ``compare" non-negative operators.
We will start by relating the Petz divergence and the sandwiched Rényi divergence.
\begin{lemma}[Corollary 2.3 in \cite{Iten2016PrettyGM}]
    For all $\alpha \in [0,1]$ and $\rho, \sigma  \in P(\mathcal{H})$, the Petz Divergence and sandwiched Rényi divergence are connected through the following inequalities:
    \begin{equation}
    \alpha \bar{D}_\alpha(\rho||\sigma)+(1-\alpha)(\log\Tr[\rho] -\log \Tr[\sigma])\leq \widetilde{D}_\alpha(\rho||\sigma)\leq \bar{D}_\alpha(\rho||\sigma) \label{eq: Divergences}
    \end{equation}
\end{lemma}
We know that for density matrices, the trace distance $T(\rho,\sigma)$ and the fidelity $F(\rho,\sigma)$ are related by the Fuchs-van de Graaf inequalities \cite{watrous_theory_2018, nielsen_quantum_2010}.
\begin{equation}
1-F(\rho,\sigma)\leq T(\rho,\sigma) \leq \sqrt{1-F(\rho, \sigma)^2},
\end{equation}
However, in the present situation, $\sigma$ is of the form $\sigma = O_i^\dagger O_j\rho O_j^\dagger O_i$ and thus not normalized to $\Tr[\sigma]=1$.
Hence we need the following generalization
\begin{lemma}[Generalized Fuchs-van de Graaf inequalities]\label{lemma: Fuchs-van de Graaf}
    Let $\rho, \sigma \in P(\mathcal{H})$ Then:
\begin{equation}
\frac{1}{2}\left(\Tr[\rho]+\Tr[\sigma]\right)-F(\rho,\sigma)\leq T(\rho,\sigma) \leq \sqrt{\frac{\left(\Tr[\rho]+\Tr[\sigma]\right)^2}{4}-F(\rho, \sigma)^2},
\end{equation}
where $T(\rho,\sigma)=\frac{1}{2}||\rho-\sigma||_1$ is the trace distance.
\end{lemma}
Note that if $\rho,\sigma$ are density matrices, the inequalities reduce to the usual Fuchs-van de Graaf inequalities.
\begin{proof} 
The proof is essentially the same proof as the one for the regular Fuchs-van de Graaf inequality, but keeping track of $\Tr[\sigma]$ and $\Tr[\rho]$. 
The left inequality is based on the fact that $||\rho-\sigma||_1\geq ||\sqrt{\rho}-\sqrt{\sigma}||^2_2$ for $\rho, \sigma \in P(\mathcal{H})$ (see Lemma 3.34 in \cite{watrous_theory_2018}).
Then we have that:
\begin{equation}
    ||\rho-\sigma||_1\geq ||\sqrt{\rho}-\sqrt{\sigma}||_2^2=\Tr[\rho+\sigma-2\sqrt{\rho}\sqrt{\sigma}] =\Tr[\rho]+\Tr[\sigma]-2\Tr[\sqrt{\rho}\sqrt{\sigma}]\geq  1+\Tr[\sigma]-2F(\rho,\sigma),
\end{equation}
where we used that $\Tr[\sqrt{\rho}\sqrt{\sigma}]\leq \Tr[|\sqrt{\rho}\sqrt{\sigma}|]=F(\rho,\sigma)$
For the right inequality, let $\mathcal{X}$ be a Hilbert space with $\dim(\mathcal{X})\geq \dim(\mathcal{H})$.
Since $\rho,\sigma$ are positive semidefinite, due to Uhlmann's theorem (Theorem 3.22 in \cite{watrous_theory_2018}) there exist purifications $\ket{\Psi},\ket{\phi}\in \mathcal{H}\otimes \mathcal{X}$ with $\Tr_\mathcal{X}[\ket{\Psi}\bra{\Psi}]=\rho$, $\Tr_\mathcal{X}[\ket{\phi}\bra{\phi}]=\sigma$ and $F(\rho,\sigma) = |\bra{\Psi}\ket{\phi}|$.
Let $\sket{\widetilde{\Psi}}$, $\sket{\widetilde{\phi}}$ be the unit vectors such that $\ket{\Psi}=\sqrt{\Tr[\sigma]}\sket{\widetilde{\Psi}}$ and $\ket{\phi}=\sqrt{\Tr[\rho]}\sket{\widetilde{\phi}}$.
Then using monotonicity of the trace norm under partial trace we have:
\begin{align}
    &||\rho-\sigma||_1\leq ||\Tr[\rho]\sket{\widetilde{\phi}}\sbra{\widetilde{\phi}}-\Tr[\sigma]\sket{\widetilde{\Psi}}\sbra{\widetilde{\Psi}}||_1=\sqrt{(\Tr[\rho]+\Tr[\sigma])^2-4\Tr[\sigma]\Tr[\rho]|\sbraket{\Psi}{\phi}|^2} \nonumber\\
    &=\sqrt{(\Tr[\rho]+\Tr[\sigma])^2-4F(\rho,\sigma)^2}
\end{align}
Here the first equality is just a restatement of 1.184 from \cite{watrous_theory_2018}.
\end{proof}
\textbf{Proof of Theorem \ref{theorem: QuantumInfo correlators}:}
We can now apply these Lemmas to prove Theorem \ref{theorem: QuantumInfo correlators}.
\begin{proof}
$(i) \iff (ii)$ follows from the generalized Fuchs-van-de Graaf inequalities in Lemma \ref{lemma: Fuchs-van de Graaf}.

$(i) \Rightarrow (iii):$ By symmetry of the fidelity, it follows that $\liminf_{n\to \infty}F(\sigma_n,\rho_n)>0$.
Let $1>\alpha \geq \frac{1}{2}$.
Then $1-\alpha \leq \frac{1}{2}$.
Since $F(\rho_n,\sigma_n)=\exp\left(-\frac{1}{2}\widetilde{D}_{\frac{1}{2}}(\rho_n,\sigma_n)\right)$ we can conclude that $\widetilde{D}_{\frac{1}{2}}(\rho_n,\sigma_n)$ remains bounded from above and it follows by monotonicity of the sandwiched Rényi divergences that $\limsup_{n\to \infty}\widetilde{D}_{1-\alpha}(\sigma_n||\rho_n)\leq \limsup_{n\to \infty}\widetilde{D}_{\frac{1}{2}}(\sigma_n||\rho_n)<\infty$.
Then, from the left inequality of Eq.~\eqref{eq: Divergences} it follows:
\begin{align*}
    &\widetilde{D}_{1-\alpha}(\sigma_n||\rho_n)\geq (1-\alpha) \bar{D}_{1-\alpha}(\sigma_n||\rho_n)+(1-(1-\alpha))(\log\Tr[\rho_n] -\log \Tr[\sigma_n])\\
    &\geq(1-\alpha) \bar{D}_{1-\alpha}(\sigma_n||\rho_n)+\alpha\left( \log C_\rho- \log C_\sigma\right).
\end{align*}
Thus it follows that $\limsup_{n\to \infty} \bar{D}_{1-\alpha}(\sigma_n||\rho_n) <\infty$. \\
$(iii)\Rightarrow (iv)$ By using Eq.~\eqref{eq: Petz switch rho sigma} we find:
\begin{align}
&\bar{D}_\alpha(\rho_n||\sigma_n)=\frac{\alpha}{1-\alpha}\bar{D}_{1-\alpha}(\sigma_n||\rho_n) +\frac{1}{1-\alpha}\left(\log\Tr[\sigma_n]-\log\Tr[\rho_n]\right) \nonumber\\
 &   \leq \frac{\alpha}{1-\alpha}\bar{D}_{1-\alpha}(\sigma_n||\rho_n)+\frac{1}{1-\alpha}\left(\log \frac{C_\sigma}{C_\rho}\right).
\end{align}
Hence $\limsup_{n\to \infty}\bar{D}_\alpha(\rho_n||\sigma_n)<\infty$.
Using the right inequality of Eq.~\eqref{eq: Divergences}:
\begin{equation}
    \limsup_{n\to\infty}\widetilde{D}_\alpha(\rho_n,\sigma_n)\leq \limsup_{n\to\infty}\bar{D}_\alpha(\rho_n,\sigma_n)
\end{equation}
$(iv) \Rightarrow (i):$ Assume $\limsup_{n\to \infty}\bar{D}_\alpha(\rho_n,\sigma_n) <\infty$ for some $\alpha \in [\frac{1}{2},0)$, it then follows again from the right inequality of Eq.~\eqref{eq: Divergences} that $\limsup_{n\to \infty}\widetilde{D}_\alpha(\rho_n,\sigma_n) <\infty$. 
Using monotonicity of $\widetilde{D}_\alpha$ in $\alpha$, it follows that $\widetilde{D}_\frac{1}{2}(\rho_n,\sigma_n)$ is bounded from above as $n\to \infty$ and hence $F(\rho_n,\sigma_n)=\exp\left(-\frac{1}{2}\widetilde{D}_{\frac{1}{2}}(\rho_n,\sigma_n)\right)$ remains bounded from below as $n\to \infty$, i.e., $\liminf_{n\to \infty}F(\rho_n,\sigma_n)>0$.
\end{proof}

\section{Convex combinations of SW-SSB states}\label{sec: convex combinations}
For ordinary symmetry breaking the correlation function is linear in the density matrix and thus convex combinations of symmetry broken states also show symmetry breaking.
Since all choices of correlation functions for SW-SSB are nonlinear, the analogous statement for SW-SSB is not immediately obvious.
In the following, we consider families of density matrices $\rho_{N}$ corresponding to increasing system sizes.
Further, we take sequences of sites $i_{N_L},j_{N_L}$ such that $\mathrm{dist}(i_{N_L},j_{N_L}) \xrightarrow[]{N\to \infty}\infty$.
\begin{corollary}
    Let $\{\rho_{1,N}\}_N, \{\rho_{2,N}\}_N$ be families of density matrices exhibiting SW-SSB in terms of either definition (Rényi-2 or fidelity).
    Then the convex combination $\tau_N = \lambda \rho_{1,N}+(1-\lambda)\rho_{2,N}$ also has SW-SSB (according to the same definition) for every $\lambda \in [0,1]$.  
\end{corollary}
\begin{proof}
    We start by assuming that $\{\rho_{1,N}\}_N, \{\rho_{2,N}\}_N$ have SW-SSB in terms of the fidelity correlator.
    By using the equivalent statement in terms of the trace distance from Thm. \ref{theorem: QuantumInfo correlators} we find:
    \begin{align}
        &\liminf_{N \to \infty}\norm{\tau -O_{i_{N_L}}^\dagger O_{j_{N_L}} \tau_N O_{j_{N_L}}^\dagger O_{i_{N_L}} }_1\\
        &\leq\liminf_{N \to \infty} \frac{1}{2}\left(\lambda\norm{\rho_{1,N}-O_{i_{N_L}}^\dagger O_{j_{N_L}} \rho_{1,N} O_{j_{N_L}}^\dagger O_{i_{N_L}}}_1+(1-\lambda)\norm{\rho_{2,N}-O_{i_{N_L}}^\dagger O_{j_{N_L}} \rho_{2,N} O_{j_{N_L}}^\dagger O_{i_{N_L}}}_1\right) \nonumber \\
        &<\frac{1}{2}\liminf_{N \to \infty}\left(\lambda \Tr[\rho_{1,N}]+\lambda \Tr[O_{i_{N_L}}^\dagger O_{j_{N_L}} \rho_{1,N} O_{j_{N_L}}^\dagger O_{i_{N_L}}] +(1-\lambda) \Tr[\rho_{2,N}]+(1-\lambda) \Tr[O_{i_{N_L}}^\dagger O_{j_{N_L}} \rho_{2,N} O_{j_{N_L}}^\dagger O_{i_{N_L}}]\right)\nonumber \\
        &=\frac{1}{2}\liminf_{N \to \infty}\left(\Tr[\tau_N]+\Tr[O_{i_{N_L}}^\dagger O_{j_{N_L}} \tau_{N} O_{j_{N_L}}^\dagger O_{i_{N_L}}]\right)
    \end{align}
Further using the equivalences in Thm.~\ref{theorem: QuantumInfo correlators}, we note that $\{\tau_N\}_N$ exhibits SW-SSB in terms of the fidelity correlator. 
To show this statement for the Rényi-2 correlator, let us assume that $\{\rho_{1,N}\}_N, \{\rho_{2,N}\}_N$ have SW-SSB in terms of the Rényi-2 correlator. 
We fix $N\in \mathbb{N}$ and assume w.l.o.g. that $\Tr[\rho_{1,N}^2]\leq \Tr[\rho_{2,N}^2]$.
From Cauchy-Schwarz inequality it then follows that $\Tr[\rho_{1,N}\rho_{2,N}]\leq \norm{\rho_{1,N}}_2 \norm{\rho_{2,N}}_2 = \sqrt{\Tr[\rho_{1,N}^2]\Tr[\rho_{2,N}^2]}\leq \Tr[\rho_{2,N}^2]$.
It follows that
\begin{align}
    &\frac{\Tr[\tau_N O_{i_{N_L}}^\dagger O_{j_{N_L}} \tau_N O_{j_{N_L}}^\dagger O_{i_{N_L}}]}{\Tr[\tau_N^2]} =\Bigg(
    \lambda^2  \Tr[\rho_{1,N} O_{i_{N_L}}^\dagger O_{j_{N_L}} \rho_{1,N} O_{j_{N_L}}^\dagger O_{i_{N_L}}]+    \lambda(1-\lambda)  \Tr[\rho_{1,N} O_{i_{N_L}}^\dagger O_{j_{N_L}} \rho_{2,N} O_{j_{N_L}}^\dagger O_{i_{N_L}}] \nonumber \\
    &+\lambda(1-\lambda)  \Tr[\rho_{2,N} O_{i_{N_L}}^\dagger O_{j_{N_L}} \rho_{1,N} O_{j_{N_L}}^\dagger O_{i_{N_L}}]+(1-\lambda)^2  \Tr[\rho_{2,N} O_{i_{N_L}}^\dagger O_{j_{N_L}} \rho_{2,N} O_{j_{N_L}}^\dagger O_{i_{N_L}}]\Bigg)\nonumber \\
&\times \Bigg(\lambda^2\Tr[\rho_{1,N}^2]+2\lambda(1-\lambda)\Tr[\rho_{1,N}\rho_{2,N}]+
    (1-\lambda )^2\Tr[\rho_{2,N}^2 ]\Bigg)^{-1}
\geq \lambda ^2\frac{\Tr[\rho_{2,N} O_{i_{N_L}}^\dagger O_{j_{N_L}} \rho_{2,N} O_{j_{N_L}}^\dagger O_{i_{N_L}}]}{\Tr[\rho_{2,N}^2]}
\end{align}
Here we used the fact that the numerator is composed of traces of a products of positive semidefinite matrices, which are nonnegative. 
Taking the limit $N \to \infty$, the statement follows.

\end{proof}

\section{{Derivations of effective Hamiltonians}}
\label{app:effH}
In this section we will derive the expressions for the various effective Hamiltonians given in the main text. 
\subsection{\texorpdfstring{$\mathbb{Z}_2$}{Lg} symmetric circuit with postselection}\label{subsec:Z2postselec}
We start by deriving the expression in Eq.~\eqref{eq: effective Hamiltonian Z2 postselection} for the effective Hamiltonian of probabilistic postselection.
The quantum operation we consider is given by
\begin{align}
    &\mathcal{E}_{dt}[\rho] = (1-p\ dtN)\rho + \frac{p\ dt\ N}{N}\sum_i \left[(1-s)\left(\pi^+_i\rho\pi^+_i+\pi^-_i\rho\pi^-_i\right) +s\pi^+_i\rho\pi^+_i\right],
    \label{eq: postselection operation appendix}
\end{align}
where $\pi_i^\pm =\frac{1\pm Z_i}{2}$. The action of $\mathcal{E}_{dt}$ on the vectorized density matrix $\opket{\rho}$ is given by
\begin{align}
    &\mathcal{E}_{dt}\opket{\rho} = (1-p\ dt\ N)\opket{\rho}+ p\ dt\sum_i \left[\pi^+_i\otimes\pi^{+*}_i+(1-s)\pi^-_i\otimes \pi^{-*}_i\right]\opket{\rho} \nonumber \\
    &= (1-p\ dt\ N)\opket{\rho}+ \frac{p\ dt }{4}\sum_i \left[(1+Z_i^f)(1+Z_i^b)+(1-s)(1-Z_i^f)(1-Z_i^b)\right]\opket{\rho} \nonumber \\
    &= (1-p\ dt\ N)\opket{\rho}+ \frac{p\ dt }{4}\sum_i \left[2-s+(2-s)Z^f_iZ^b_i+s(Z^f_i+Z^b_i)\right]\opket{\rho} \nonumber\\
    &= \opket{\rho}- \frac{p\ dt }{4}\sum_i \left[2+s+(s-2)Z^f_iZ^b_i-s(Z^f_i+Z^b_i)\right]\opket{\rho} \nonumber\\
    &= \bigg[1- dt\underbrace{\left(\frac{p (s-2)}{4}\sum_i \left(\frac{s+2}{s-2}+Z^f_iZ^b_i\right)-\frac{ps}{4}\sum_i(Z^f_i+Z^b_i)\right)}_{P_{\mathbb{Z}_2}^{\mathrm{meas}}}\bigg] \opket{\rho}.
\end{align}
This is the form of $P_{\mathbb{Z}_2}$ given in Eq.~\eqref{eq: effective Hamiltonian Z2 postselection}.
$P_{\mathbb{Z}_2}$ commutes with $Z_i^fZ_i^b$ for every $i$ and thus has a $\mathbb{Z}_2^N$-symmetry.
The initial state is $\rho_0=\ket{\uparrow \ldots \uparrow}\bra{\uparrow \ldots \uparrow}$, which has a positive parity $Z_i^fZ_i^b=+1$ across each $f-b$ rung (see Fig.~\ref{fig: doubled Hspace}), and due to the $\mathbb{Z}_2^N$-symmetry, the state remains in that subspace of states with $Z_i^fZ_i^b=+1$ for all sites $i$, and for all times.
Hence we can, as in the derivation of Eq.~\eqref{eq: P for Z2 base model}, restrict the analysis to the reduced Hilbert space $\widetilde{\mathcal{H}}=\mathrm{span}\{\bigotimes_i\opket{\widetilde{p}_i}, \; p_i \in \{\uparrow,\downarrow\}\}$, where we have defined
\begin{equation}
|\widetilde{\uparrow}\rangle_i=\opket{{\centering \uparrow}}_{i,f}\otimes\opket{{\centering \uparrow}}_{i,b},\;\;\;|\widetilde{\downarrow} \rangle_i=\opket{{\centering \downarrow}}_{i,f}\otimes\opket{{\centering \downarrow}}_{i,b}.
\end{equation}
On this subspace $Z^f_i+Z_i^b$ acts as $2\widetilde{Z}_i$ and thus $P_{\mathbb{Z}_2}$ acts as
\begin{equation}
    P_{\mathbb{Z}_2}^{\mathrm{meas}}=\frac{psN}{2}+\frac{ps}{2}\sum_i\widetilde{Z}_i.
\end{equation}
Combining this with the effective Hamiltonian $P_{\mathbb{Z}_2}$ of the Brownian dynamics  (i.e., Eq.~\eqref{eq: P for Z2 base model}) we find that $P_{\mathbb{Z}_2}^{\mathrm{tot}}=P_{\mathbb{Z}_2}+P_{\mathbb{Z}_2}^{\mathrm{meas}}$ is given by:
\begin{align}
     P_{\mathbb{Z_2}}^{\mathrm{tot}}=2J\sum_{\langle i,j \rangle}[1-\widetilde{X}_j\widetilde{X}_i]+\frac{ps}{2}\sum_i\widetilde{Z}_i+\frac{psN}{2},
\end{align}
which is the Hamiltonian of the transverse field Ising model as given in Eq.~\eqref{eq: TFI}.
\subsection{\texorpdfstring{$S_3$}{Lg} symmetric circuit with postselection}
We will start by deriving the effective Hamiltonian for the Brownian circuit without postselection.
We will then add the quantum operation Eq.~\eqref{eq: postselection operation S3} and implement the probabilistic postselection protocol.
Starting from the Potts Hamiltonian with random Brownian couplings on a lattice from Eq.~\eqref{eq: Potts}
\begin{equation}
H(t)=\sum_iU_i(t)\left(\tau_i+\tau^\dagger_i\right)+\sum_{\langle i,j\rangle}J_{ij}(t)\left(\sigma_i^\dagger\sigma_j+\sigma_i\sigma_j^\dagger\right), 
\label{eq:Pottsapp}
\end{equation}
and using the formula for the effective Hamiltonian $P_{S_3}$ of a Brownian circuit from Eq.~\eqref{eq: Brownian formula} we find:
\begin{align}
    P_{S_3}&=U\sum_i \left(\tau_i^f+\tau_i^{f\dagger}-\tau_i^b-\tau_i^{\dagger b}\right)^2+J\sum_{\langle i,j \rangle}\left(\sigma_{i}^{f\dagger}\sigma_{j}^{ f}+ \sigma_{i}^{f}\sigma_{j}^{f\dagger}-\sigma_{i}^{\dagger b}\sigma_{j}^{b}- \sigma_{i}^{b}\sigma_{j}^{\dagger b}\right)^2 \nonumber \\
    &= U\sum_i\left[\left( \tau_i^f +\tau_i^{f\dagger} +\tau_i^b +\tau_i^{b\dagger}+4\right) -2\left(\tau_i^f +\tau_i^{f\dagger}\right)\left(\tau_i^b +\tau_i^{b\dagger}\right)\right] \nonumber \\
    &+ J\sum_{\langle i,j \rangle}\left[\left(\sigma_{i}^{f\dagger}\sigma_{j}^{ f}+ \sigma_{i}^{f}\sigma_{j}^{f\dagger} + \sigma_{i}^{b\dagger}\sigma_{j}^{ b}+ \sigma_{i}^{b}\sigma_{j}^{b\dagger} +4\right)-2\left(\sigma_{i}^{f\dagger}\sigma_{j}^{ f}+ \sigma_{i}^{f}\sigma_{j}^{f\dagger}\right) \left(\sigma_{i}^{b\dagger}\sigma_{j}^{b}+ \sigma_{i}^{b}\sigma_{j}^{b\dagger}\right)\right] \nonumber \\
    &=H^f_{\mathrm{Potts}}(U,J)+H^b_{\mathrm{Potts}}(U,J)-2U\sum_i\left(\tau_i^f +\tau_i^{f\dagger}\right)\left(\tau_i^b +\tau_i^{b\dagger}\right)-2J\sum_{\langle i,j \rangle}\left(\sigma_{i}^{f\dagger}\sigma_{j}^{ f}+ \sigma_{i}^{f}\sigma_{j}^{f\dagger}\right) \left(\sigma_{i}^{b\dagger}\sigma_{j}^{b}+ \sigma_{i}^{b}\sigma_{j}^{b\dagger}\right) \nonumber \\
    &+4N(U+J)-4J,\label{eq:PS3}
\end{align}
where $H^{f/b}_{\mathrm{Potts}}(U,J)$ are Potts Hamiltonians of the form Eq.~(\ref{eq:Pottsapp}) with fixed coupling constants $U,J$ acting only on the froward/backward copy.
Here, we have repeatedly used the fact that $\tau^\dagger_i = \tau_i^2$, $\tau_i^3=1$, $\sigma^\dagger_i = \sigma_i^2$, and $\sigma_i^3=1$.
Let us now consider the action of the quantum operation
\begin{align}
    &\mathcal{E}_{dt}[\rho] = (1-p\ dt\ N)\rho + \frac{p\ dt\ N}{N} \sum_i \left[(1-s)\left(\pi^0_i\rho\pi^0_i+\pi^1_i\rho\pi^1_i\right) +s\pi^0_i\rho\pi^0_i\right]
\end{align}
on the vectorized density matrix $\opket{\rho}$.
Here, $\pi_i^1=Q_i^2$ and $\pi_0=1-Q_i^2$, where $Q_i$ is defined in Eq.~(\ref{eq:Qidefn}).
We find, similar to Eq.~(\ref{eq: postselection operation appendix}):
\begin{align}
    \mathcal{E}_{dt}\opket{\rho} &= (1-p\ dt\ N)\opket{\rho} + p\ dt\sum_i \left[\pi^0_i\otimes\pi^0_i+(1-s)\pi^1_i\otimes\pi^1_i\right]\opket{\rho} \nonumber \\
    &=(1-p\ dt\ N)\opket{\rho} + p\ dt \sum_i \left[1-(Q_i^{f})^2-(Q_i^{b})^2+(2-s) (Q_i^{f} Q_i^{b})^2\right]\opket{\rho} \nonumber \\
    &=\bigg(1-dt \ \underbrace{p\sum_i\left[(Q_i^{f})^2 + (Q_i^{b})^2 +(s-2)(Q_i^{f}Q_i^{b})^2\right]}_{P_{S_3}^{\mathrm{meas}}} \bigg)\opket{\rho}.
\end{align}
$P_{S_3}^{\mathrm{meas}}$ can be written in terms of the $\tau_i$ operators by substituting $Q_i=\frac{i}{\sqrt{3}}\left(\tau_i^\dagger -\tau_i\right)$ and using that $\tau_i^2=\tau_i^\dagger$ and $\tau_i^3=\mathbb{1}$:
\begin{align}
  P_{S_3}^{\mathrm{meas}} &=p\sum_i\left[(Q_i^{f})^2 + (Q_i^{b})^2 +(s-2)(Q_i^{f}Q_i^{b})^2\right] \nonumber \\
  &= p\sum_i \Bigg[ \frac{-1}{3}\left(\tau_i^{f}+\tau_i^{f \dagger} +\tau_i^{b}+\tau_i^{b \dagger}-4\right)
  +\frac{s-2}{9}\left(\tau_i^{f}+\tau_i^{f \dagger}-2\right)\left(\tau_i^{b}+\tau_i^{b \dagger}-2\right)\Bigg]
\label{eq:PS3measexpr}
\end{align}
The total effective Hamiltonian of Eq.~\eqref{eq: Potts-2copy} is then (up to an overall constant)  $P_{S_3} + P_{S_3}^{\mathrm{meas}}$, where $P_{S_3}$ and $P_{S_3}^{\mathrm{meas}}$ are given by Eqs.~(\ref{eq:PS3}) and (\ref{eq:PS3measexpr}), respectively. 
\subsection{\texorpdfstring{$\mathbb{Z}_2$}{Lg} symmetric Lindbladian model}
Let us now derive the effective Hamiltonian Eq.~\eqref{eq: total spin hamiltonian}.
We start with the Hamiltonian part.
We extend the effective Hamiltonian $\mathbb{Z}_2$ for the Brownian from Eq.~\eqref{eq: P for Z2 base model} from a nearest neighbor to one with an all-to-all coupling:
\begin{align}
    P_{\mathbb{Z}_2} &= 2J\sum_{i<j}(1-\widetilde{X}_i\widetilde{X}_j)
    =J\sum_{i,j}(1-\widetilde{X}_i\widetilde{X}_j)-J\sum_i(1-\widetilde{X}_i\widetilde{X}_i)\nonumber \\
    &=J N^2-\left(\sum_i\widetilde{X}_i\right)^2=J N^2-4\widetilde{S}^x\widetilde{S}^x.
\end{align}
where we have defined the total spin operators as $\widetilde{S}^{x,y,z}=\sum_i\widetilde{S}_i^{x,y,z}$ and $\widetilde{S}_i^x=\frac{1}{2}\widetilde{X}_i$ (and similarly for $\widetilde{S}_i^{y,z}$).
The channel as discussed in Sec.~\ref{subsec:Lindbladexample} is given by:
\begin{align}
    &\mathcal{E}_{dt}[\rho]=(1-p\frac{dt\ N(N-1)}{2})\rho +p\ dt\frac{\ N(N-1)}{2} \frac{1}{\frac{\ N(N-1)}{2}}\sum_{i<j}[\pi^+_i\pi^+_j\rho\pi^+_j\pi^+_i+\pi^+_i\pi^-_j\rho\pi^-_j\pi^+_i
    +\pi^-_i\pi^+_j\rho\pi^+_j\pi^-_i \nonumber\\
    &+X_iX_j\pi^-_i\pi^-_j\rho\pi^-_j\pi^-_iX_jX_i],  \label{: channel adaptive}
\end{align}
where $\pi^{\pm}_{i}=\frac{1\pm Z_i}{2}$.
Without the $X$ operators, the above channel would be one where a measurement of the $Z_i$ and $Z_j$ operators are performed for a random pair of $i\neq j$ (with a probability of $\frac{ p\ dt\ N(N-1)}{2}$).
By including the $X$ operators, we incorporate classical feedback by flipping both spins only in the case where both measurement outcomes correspond to $\downarrow$. 
We can then directly compute the effective Hamiltonian in the doubled Hilbert space to be:
\begin{align}
    P_{\mathbb{Z}_2}^{\mathrm{meas}} &= -\lim_{dt\to 0}\frac{\mathcal{E}_{dt}-1}{dt}=
    p\frac{ N(N-1)}{2} -p \sum_{i<j}[\pi^+_i\pi^+_j\otimes\pi^+_i\pi^+_j \nonumber \\
    &+\pi^+_i\pi^-_j\otimes\pi^+_i\pi^-_j
    +\pi^-_i\pi^+_j\otimes\pi^-_i\pi^+_j +X_iX_j\pi^-_i\pi^-_j\otimes X_iX_j\pi^-_i\pi^-_j].
\end{align}
The last term simplifies to
\begin{align}
    X_i^fX_j^f\pi^{-\ f}_i\pi^{-\ f}_j X_i^bX^b_j\pi^{-\ b}_i\pi^{-\ b}_j=S_i^{+ \ f}S_j^{+ \ f}S_i^{+ \ b}S_j^{+ \ b}.
\end{align}
As in the $Z_2$ case in Sec.~\ref{subsec:Z2postselec}, the effective Hamiltonian commutes with the operators $Z_i^fZ_i^b$ for every $i$, and since the initial state is an eigenstate of $Z_i^fZ_i^b$ for every $i$, the dynamics will be confined to the reduced Hilbert space $\widetilde{\mathcal{H}}=\mathrm{span}\{\bigotimes_i\opket{\widetilde{p}_i}, \; p_i \in \{\uparrow,\downarrow\}\}$ for all times.
We have defined $\opket{\widetilde{\uparrow}}_i=\ket{\uparrow}_{i,f}\otimes\ket{\uparrow}_{i,b}$ and $\opket{\widetilde{\downarrow}} _i=\ket{\downarrow}_{i,f}\otimes\ket{\downarrow}_{i,b}$.
Here, $S_i^{+ \ f}S_i^{+ \ b}=\widetilde{S}_i^+$ and ${\pi}_i^{\pm\ f}{\pi}_i^{\pm\ b}=\widetilde{\pi}_i^\pm$.
Taking everything together we find:
\begin{align}
    P_{\mathbb{Z}_2}^{\mathrm{meas}} &= p\frac{ N(N-1)}{2} -p \sum_{i<j}\left[\widetilde{\pi}_i^+\widetilde{\pi}_j^++\widetilde{\pi}_i^-\widetilde{\pi}_j^++\widetilde{\pi}_i^+\widetilde{\pi}_j^-+\widetilde{S}_i^+\widetilde{S}_j^+\right] \nonumber \\
    &=p\frac{ N(N-1)}{2} -\frac{p}{4} \sum_{i<j}\left[3+\widetilde{Z}_i+\widetilde{Z}_j-\widetilde{Z}_i\widetilde{Z}_j+4\widetilde{S}_i^+\widetilde{S}_j^+\right] \nonumber \\
    &=p\frac{ N(N-1)}{2} -\frac{p}{8} \sum_{i,j}\left[3+\widetilde{Z}_i+\widetilde{Z}_j-\widetilde{Z}_i\widetilde{Z}_j+4\widetilde{S}_i^+\widetilde{S}_j^+\right]+\frac{p}{8} \sum_{i}\left[2+2\widetilde{Z}_i\right] \nonumber \\
    &=-\frac{p}{2}\left(N\widetilde{S}^z-\widetilde{S}^z\widetilde{S}^z+\widetilde{S}^+\widetilde{S}^+-\frac{N^2}{4}\right) -\frac{pN}{4}+\frac{p}{2}\widetilde{S}^z, \label{eq: Lindblad measurement Hamiltonian}
\end{align}
where we have defined $\widetilde{S}^z = \sum_j{\widetilde{Z}_j}$.
From the second to the third line we replaced the sum over pairs $\sum_{i<j}$ by a sum $\sum_{i,j}$ at the expense of a factor of $\frac{1}{2}$ to avoid double counting and subtracted the diagonal where $i=j$.
Since the dynamics are Lindbladian, we can find the jump operators. 
Given a quantum channel of the form $\mathcal{E}_{dt}[\rho]=(1-dt\cdot \gamma)\rho +\gamma\cdot dt\;\mathcal{E}[\rho]$, where $\mathcal{E}[\rho]=\sum_iK_i\rho K_i^\dagger$ is another quantum channel, we can express this as:
\begin{align}
    \mathcal{E}_{dt}[\rho]=\exp\left(-dt \;\gamma\left(1-\sum_iK_i \rho K_i^\dagger\right)\right)=\exp\left(-dt \left(-\sum_i \gamma\left[K_i \rho K_i^\dagger -\frac{1}{2}\{K_i^\dagger K_i,\rho\}\right]\right)\right) +\mathcal{O}(dt^2), \label{eq: Kraus to Lindblad}
\end{align}
where we used that $1=\sum_iK_i^\dagger K_i$ since $\mathcal{E}$ is trace preserving. 
Comparing Eq.~\eqref{eq: Kraus to Lindblad} with Eq.~(\ref{eq: Lindbladian}), we see that the Kraus operators act as Lindblad jump operators.
Hence by using Eq.~\eqref{: channel adaptive}, the jump operators are $\pi^+_i\pi_j^+$, $\pi^+_i\pi_j^-$, $\pi^-_i\pi_j^+$ and $X_iX_j\pi^-_i\pi_j^-$. We also have $\gamma = p\frac{N(N-1)}{2}$.

In this appendix, we derive results on the SW-SSB of MMIS for on-site symmetries that form a group $G$ and have exponentially growing scalar symmetry sectors.
We also provide the proof for our formula Eq.~\eqref{eq: formula R2 maximally mixed} for the Rényi-2 correlators of MMIS.
\section{Fidelity correlator for \texorpdfstring{$\mathbb{Z}_2$}{Lg} symmetric circuits} \label{app:fidelity}
\subsection{Fidelity as series of Rényi correlators}
We will sketch an analytical argument for why we expect the fidelity correlator to behave similarly to the Rényi correlators. 
We consider the $\mathbb{Z}_2$ symmetric circuits and define $\rho(t\to \infty) = \lim_{t \to \infty}\mathbb{E}[\rho(t)]$, where $\mathbb{E}[\ldots]$ denotes the average over the circuit randomness.
We also pick the initial state $\ket{\uparrow \ldots \uparrow}$.
Let $\rho^\text{eq}=\frac{\rho(t\to\infty)}{\Tr[\rho(t\to \infty)]}$.
The fidelity is then given by:
\begin{align}
    F_X(i,j)[\rho^\text{eq}]=\Tr[\sqrt{\sqrt{\rho^\text{eq}}X_iX_j\rho^\text{eq}X_jX_i\sqrt{\rho^\text{eq}}}]
    =\frac{\Tr[\sqrt{\sqrt{\rho(t\to \infty))}X_iX_j\rho(t\to \infty)X_jX_i\sqrt{\rho(t\to \infty)}}]}{\Tr[\rho(t\to \infty)]}. \label{eq: fidelity renormalized}
\end{align}
We now express the fidelity in terms of quantities that can be computed in terms of the Ising ground state, which has the form of the vectorized density matrix $\opket{\rho(t \rightarrow \infty)}$.
We will use the fact that $\opket{\rho(t\to \infty)}$ lies in the subspace $\widetilde{\mathcal{H}}$ of the doubled Hilbert space $\mathcal{H}_f\otimes \mathcal{H}_b$, as discussed in Sec.~\ref{subsec:abelianZ2}.
To do so, we first focus on the denominator of Eq.~(\ref{eq: fidelity renormalized}) $\Tr[\rho(t\to \infty)]$ by evaluating it in the $Z$-basis as
\begin{align}
    \Tr[\rho(t\to\infty)] = \sum_{i}\rho(t\to\infty)_{\alpha\alpha}=\lim_{t\to \infty}\sum_\alpha\bra{\alpha_f}\otimes\bra{\alpha_b}e^{-tP_{\mathbb{Z_2}}^{\mathrm{tot}}}\opket{\rho_0}=\lim_{t\to \infty}\sum_{\widetilde{\alpha}\in\{\widetilde{\uparrow},\widetilde{\downarrow}\}^N} \opbra{\widetilde{\alpha_1}\ldots \widetilde{\alpha_N} }e^{-t P_{\mathbb{Z_2}}^{\mathrm{tot}}}\opket{\rho_0},
\end{align}
where in the last equality we have expressed the sum in terms of the basis states of the reduced Hilbert space $\widetilde{\mathcal{H}}$, introduced above Eq.~\eqref{eq: P for Z2 base model} $\opket{\widetilde{\alpha}}=\ket{\alpha_f}\otimes \ket{\alpha_b}$.
Since the initial state is $\opket{\rho_0}=\opket{\widetilde{\uparrow}\ldots \widetilde{\uparrow}}$, we can obtain  $\opket{\widetilde{\alpha_1}\ldots \widetilde{\alpha_N} }$ from $\opket{\rho_0}$ by applying a spin-flip to all sites where the local state is $\opket{\widetilde{\downarrow}}$, i.e., $\opket{\widetilde{\alpha_1}\ldots \widetilde{\alpha_N} }=\prod_{\widetilde{\alpha_j}=\widetilde{\downarrow}}\widetilde{X}_{j}\opket{\rho_0}$.
Since the initial state lies in the sector of  fixed positive parity, and the dynamics are strongly $\mathbb{Z}_2$ symmetric, the overlap of $e^{-tP_{\mathbb{Z_2}}^{\mathrm{tot}}}\opket{\rho_0}$ with all states containing an odd number of $\widetilde{\downarrow}$ vanishes. 
$\prod_{\widetilde{\alpha_j}=\widetilde{\downarrow}}\widetilde{X}_{j}$ contains an even number of $\widetilde{X}$ operators, the product of which in turn commute with $P_{\mathbb{Z_2}}^{\mathrm{tot}}$. 
We therefore find that
\begin{align}
    &\rho_{\alpha\alpha} = \lim_{t \to \infty}\opbra{\widetilde{\alpha_1}\ldots \widetilde{\alpha_N} }e^{-tP_{\mathbb{Z_2}}^{\mathrm{tot}}}\opket{\rho_0} =\lim_{t \to \infty}\opbra{\rho_0}\prod_{\widetilde{\alpha_j}=\widetilde{\downarrow}}\widetilde{X}_{j}e^{-tP_{\mathbb{Z_2}}^{\mathrm{tot}}}\opket{\rho_0} \nonumber\\
    =&\lim_{t \to \infty}\opbra{\rho_0} e^{-\frac{t}{2}P_{\mathbb{Z_2}}^{\mathrm{tot}}}\prod_{\widetilde{\alpha_j}=\widetilde{\downarrow}}\widetilde{X}_{j}e^{-\frac{t}{2}P_{\mathbb{Z_2}}^{\mathrm{tot}}}\opket{\rho_0}=\opbra{\rho(t\to\infty)}\prod_{\widetilde{\alpha_j}=\widetilde{\downarrow}}\widetilde{X}_{j}\opket{\rho(t\to\infty)}=:C^{\mathrm{Ising}}_{\{s_1,\ldots s_n\}}. \label{eq: definition of C Ising}
\end{align}
Since $\opket{\rho(t\to \infty)}$ is the unnormalized groundstate of the Transverse Field Ising Model (TFIM), $C^{\mathrm{Ising}}_{\{s_1,\ldots s_n\}}$ is an $n$-point correlation function of the TFIM ground state where $s_1, \ldots s_n$ label the sites where a spin was flipped by applying the $\widetilde{X}_j$ operator.
Hence $\Tr[\rho(t\to \infty)]$ corresponds to a sum of correlation functions of the TFIM.
We now turn to the expression of the numerator in Eq.~(\ref{eq: fidelity renormalized}).
Written in terms of density matrices, $\widetilde{\mathcal{H}}$ corresponds to matrices that are diagonal in the $Z$-basis.
Hence $X_iX_j\rho(t\to \infty)X_jX_i$ is also diagonal in this basis, and we can write
\begin{align}
    \Tr\left[\sqrt{\sqrt{\rho(t\to\infty)} X_iX_j \rho(t\to\infty) X_jX_i\sqrt{\rho(t\to\infty)}}\right]=\sum_\alpha\sqrt{\rho(t\to\infty)_{\alpha\alpha}} \sqrt{\left[X_iX_j\rho(t\to\infty)X_j X_i\right]_{\alpha\alpha}}. \label{eq: expression square root}
\end{align}
As seen from Eq.~\eqref{eq: definition of C Ising}, $\rho(t\to\infty)_{\alpha\alpha}=C_{s}^{\mathrm{Ising}}$ for an appropriate set of sites $s=\{s_1, \ldots ,s_n\}$.
Likewise
\begin{align}
    \left[X_iX_j\rho(t\to\infty)X_j X_i\right]_{\alpha\alpha}=\lim_{t \to \infty}\opbra{\widetilde{\alpha_1}\ldots \widetilde{\alpha_N} }\widetilde{X}_i\widetilde{X}_je^{-tP_{\mathbb{Z_2}}^{\mathrm{tot}}}\opket{\rho_0}=C^{\mathrm{Ising}}_{s \widetilde{\cup}\{i,j\}},
\end{align}
where $s \widetilde{\cup}  \{i\}$ ($j$ respectively) means $s \cup  \{i\}$  if $i \not \in s$ and $s\setminus \{i\}$ if $i \in s$ (since in this case the two $\widetilde{X}_i$ operators multiply to the identity). 
Combining Eqs.~\eqref{eq: fidelity renormalized}-\eqref{eq: expression square root}, the fidelity correlator from Eq.~\eqref{eq: fidelity renormalized} is found to be:
\begin{align}
    F_X(i,j)=\frac{\sum_{s \in \mathcal{P}([N])}\sqrt{C^{\mathrm{Ising}}_{s} C^{\mathrm{Ising}}_{s \widetilde{\cup}  \{i,j\}}}}{\sum_{s \in \mathcal{P}([N])}C^{\mathrm{Ising}}_{s}}, \label{eq: fidelity as npoint functions}
\end{align}
where $\mathcal{P}([N])$ denotes all subsets $s$ of $[N]=\{1,2,\ldots N\}$.
The importance in the above expression lies in the fact that the fidelity can be fully expressed in terms of correlation functions in the groundstate of the transverse field Ising model.
Hence, we expect the same onset of SW-SSB from the fidelity $F_X$ and from the Rényi-2 correlator. 
\subsection{Simulation data}
In this section we provide additional data for the fidelity correlator of the $\mathbb{Z}_2$ symmetric circuit both in the postselection model and in the adaptive model studied in Secs.~\ref{subsec:abelianZ2} and \ref{subsec:Lindbladexample} respectively.
For models with local interactions, the Rényi-2 correlator can be evaluated efficiently using tensor network methods in 1-d or analytically in the case of the postselection model in Sec.~\ref{subsec:abelianZ2} by mapping it to the transverse field Ising model.
For the fidelity correlator, these methods are not available. 
The fidelity correlator for the average steady-state density matrix is defined in Eq.~\eqref{eq: fidelity renormalized}. 
To arrive at the plot, the transverse field Ising model was diagonalized and the fidelity was calculated using Eq.~\eqref{eq: expression square root}, using the fact that the vectorized density matrix $\opket{\rho(t \rightarrow \infty)}$ is the ground state of the transverse field Ising model Eq.~\eqref{eq: TFI}.
For the postselection model, the data depicted in Fig.~\ref{fig: Z2 fidelity}b was further used to train the neural network that was able to estimate the fidelity for larger system sizes in Sec.~\ref{subsec:abelianZ2}.
For the Lindbladian model involving the feedback and measurement protocol, the fidelity plotted in Fig.~\ref{fig: Z2 fidelity} shows no signs of a phase transition as the measurement rate is increased.
This is consistent with the behavior of the Rényi-2 correlator as shown in Fig.~\ref{fig: adaptive Rényi} in the main text.

\begin{figure*}
    \centering
    \includegraphics[width=1\linewidth]{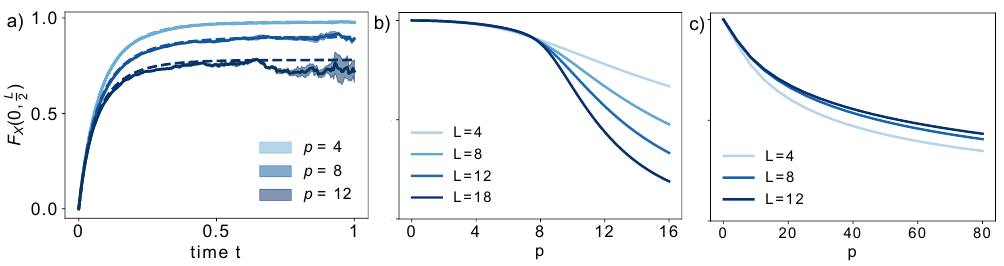}
    \caption{\textbf{Fidelity for the $\mathbb{Z}_2$ circuit:} \textbf{a):} Simulated dynamics for the fidelity in the $\mathbb{Z}_2$-circuit with postselection for $J=U=1$, $N=4$ and $s=0.5$. \textbf{b)} Steady-state value of the fidelity from exact diagonalization of the transverse field Ising model with the same parameters. \textbf{c)} Fidelity evaluated in the steady state of the adaptive Lindblad dynamics with $J=U=1$.}
    \label{fig: Z2 fidelity}
\end{figure*}

\twocolumngrid

\end{document}